\documentclass[reprint, amsmath,amssymb,showpacs,floatfix,longbibliography, aps, twocolumn, superscriptaddress,nofootinbib]{revtex4-1}
\usepackage{CJKutf8}
\pdfoutput=1
\usepackage{amsthm}
\makeatletter
\newcommand{\newreptheorem}[2]{\newtheorem*{rep@#1}{\rep@title}\newenvironment{rep#1}[1]{\def\rep@title{#2 \ref*{##1}}\begin{rep@#1}}{\end{rep@#1}}}
\makeatother

\newtheorem{theorem}{Theorem}
\newreptheorem{theorem}{Theorem}
\newtheorem{corollary}{Corollary}[theorem]
\newtheorem{lemma}[theorem]{Lemma}
\newreptheorem{lemma}{Lemma}
\usepackage{multirow}
\usepackage[utf8]{inputenc}
\usepackage{amsmath,amssymb,amsfonts,stmaryrd}
\usepackage{physics}
\usepackage{dsfont}
\usepackage{graphicx}
\usepackage{xcolor}
\usepackage{colortbl}
\usepackage{graphicx}
\usepackage{cancel}
\usepackage{natbib}
\usepackage{color}
\usepackage{tikz}
\usetikzlibrary{calc}
\usetikzlibrary {matrix}
\usetikzlibrary{positioning}
\usetikzlibrary{decorations.markings,arrows}
\usetikzlibrary { decorations.pathmorphing, decorations.pathreplacing, decorations.shapes}
\usetikzlibrary{shapes.geometric}
\usetikzlibrary{arrows.meta}
\usetikzlibrary {3d}
\usetikzlibrary {patterns,patterns.meta}
\usetikzlibrary {shadows}
\usepackage{mathrsfs}
\usepackage{relsize}

\usepackage[all]{xy}
\usepackage{yhmath}
\usepackage{blkarray}
\usepackage{tabularx}
\usepackage{array}
\usepackage{makecell}

\usepackage{diagbox}
\usepackage{algorithm}
\usepackage{algorithmic}
\usepackage{lipsum}
\newcommand*\colvec[3][]{
    \begin{pmatrix}\ifx\relax#1\relax\else#1\\\fi#2\\#3\end{pmatrix}
}

\usepackage{amsmath}
\usepackage{bm} %allows for bold math type
\usepackage{subfigure} %allows for 2x2 figures
\usepackage[percent]{overpic}
\usepackage{environ}
\NewEnviron{eqs}{%
\begin{equation}\begin{split}
    \BODY
\end{split}\end{equation}}
\usepackage{url}
\usepackage[margin=0.75in]{geometry}
\usepackage{scalerel}
\usepackage{stackengine,wasysym}
\usepackage{braket}

\usepackage{enumitem}

\makeatletter

\newcommand*{\wideboxed}[1]{\setlength{\fboxsep}{1ex}%
  \fbox{\m@th$\displaystyle#1$}}
\makeatother

\usepackage{hyperref}
\hypersetup{colorlinks=true,citecolor=blue,linkcolor=blue, urlcolor=blue, breaklinks=true}
\hypersetup{linktocpage}

\allowdisplaybreaks

\usepackage{mathtools}

\newcommand{\ZZ}{{\mathbb Z}}

\renewcommand{\bar}{\overline}
\renewcommand{\tilde}{\widetilde}
\renewcommand{\hat}{\widehat}

\newcommand{\BBox}{\text{BBox}}
\newcommand{\supp}{\text{supp}}

\newcommand{\lm}{\operatorname{lm}}

\newcommand{\LT}{\operatorname{LT}}
\newcommand{\LM}{\operatorname{LM}}
\newcommand{\lcm}{\operatorname{lcm}}
\newcommand{\Ann}{\operatorname{Ann}}

\graphicspath{{./Figures/}}

\newcolumntype{C}{>{\centering\arraybackslash}p{0.72em}}

\newcommand{\prlsection}[1]{\vspace{6pt}\noindent\textbf{\textsc{#1.}}—}

\begin{document}

\author{Peilun Han}
\thanks{equal contribution.}
\affiliation{International Center for Quantum Materials, School of Physics, Peking University, Beijing 100871, China}

\author{Zijian Liang}
\thanks{equal contribution.}
\affiliation{International Center for Quantum Materials, School of Physics, Peking University, Beijing 100871, China}

\author{Yifei Wang}
\thanks{equal contribution.}
\affiliation{Institute for Advanced Study, Tsinghua University, Beijing, 100084, China}

\author{Bowen Yang}
%\email[E-mail: ]{bowen\_yang@g.harvard.edu}
\affiliation{Center of Mathematical Sciences and Applications, Harvard University, Cambridge, Massachusetts 02138, USA}
\affiliation{Department of Mathematics, Harvard University, Cambridge, Massachusetts 02138, USA}

\author{Yingfei Gu}
%\email[E-mail: ]{guyingfei@tsinghua.edu.cn}
\affiliation{Institute for Advanced Study, Tsinghua University, Beijing, 100084, China}

\author{Yu-An Chen}
\email[E-mail: ]{yuanchen@pku.edu.cn}
\affiliation{International Center for Quantum Materials, School of Physics, Peking University, Beijing 100871, China}

\date{\today}
\title{Resolving spurious topological entanglement entropy in stabilizer codes}

\begin{abstract}
Topological entanglement entropy (TEE) is a key diagnostic of long-range entanglement in two-dimensional gapped phases of matter, but it can suffer from spurious contributions that overestimate the total quantum dimension of the underlying topological order.
In this work, we identify the microscopic origin of spurious TEE and introduce a concave partition for computing the Levin-Wen TEE of translation-invariant stabilizer codes of prime-dimensional qudits. We rigorously prove that this prescription is free of spurious contributions.
As a complementary probe, we study bivariate bicycle codes on a bipartite cylinder and show that the entanglement entropy depends sensitively on the cylinder circumference, revealing topological frustration of the underlying anyons.
\end{abstract}

\maketitle

%%%%%%%%%%%%%%%%%%%%%%%%%%%%%%%%%%%%%%%%%%%%%%%%%%%%%%%
\prlsection{Introduction}
Entanglement provides a powerful framework for characterizing the structure of quantum many-body states; particularly, the scaling behavior of entanglement entropy reveals universal features of quantum phases of matter~\cite{Wen1990RIGID, Vidal2003Entanglement, Chen2010Local}. In gapped systems, this behavior is typically governed by an area law, while subleading corrections can encode nonlocal information about the underlying phase~\cite{Hamma2005Bipartite, Hamma2005GroundStateEntanglement, Hastings2007area, Eisert2010AreaLaws}. For ground states of gapped two-dimensional Hamiltonians, the entanglement entropy of a region $A$ takes the form
\begin{equation}
    S_A = \alpha\,|\partial A| - \gamma + \ldots,
    \label{eq:TEE}
\end{equation}
where the ellipsis denotes terms that vanish in the limit $|\partial A| \to \infty$, $\alpha$ is a nonuniversal constant, and $|\partial A|$ is the boundary length of $A$. Here, the constant $\gamma$, often referred to as the \emph{topological entanglement entropy} (TEE), is universal and serves as a robust signature of intrinsic topological order under broad conditions~\cite{Kitaev2006Topological, Levin2006Detecting, dennis2002topological, Flammia2009TopologicalEntanglementRenyi, Shiying2008Topologicalentanglement, bravyi2010topological, Jiang2012Identifying}.

In lattice realizations of topological phases, however, the extraction of TEE can be subtle. In particular, \emph{spurious} topological entanglement entropy (sTEE) may arise even in the absence of intrinsic topological order, obscuring the interpretation of $\gamma$~\cite{zouSpuriousLongrangeEntanglement2016}. This issue is especially prominent in stabilizer codes, where subsystem symmetries can generate artificial constant contributions to entanglement entropy~\cite{Williamson2019sTEEsubsystem,katoToyModelBoundary2020}. Recently, it has been shown that such spurious contributions are non-negative; the TEE always obeys a universal lower bound set by the total quantum dimension~\cite{Kim2023UniversalLowerBoundTEE, Levin:2024ngk}.

In this work, we introduce a prescription that eliminates spurious contributions for translation-invariant stabilizer codes of prime-dimensional qudits. We rigorously prove that our construction recovers the genuine TEE, the logarithm of the total quantum dimension~\cite{Kitaev2006Topological, Levin2006Detecting}. Illustrative examples and the full details of the proof are presented in the Appendices. 
In addition, for systems with periodic boundary conditions, we show that the entanglement entropy is sensitive to the system size, relating our analysis to entanglement probes of topological order on cylinders and tori~\cite{Zhang2012Quasiparticle, Cincio2013Characterizing} and to topological frustration in quantum low-density parity-check (qLDPC) codes~\cite{Panteleev2022QuantumLDPC, Bravyi2024HighThreshold, liang2025generalized, chen2025anyon, liang2025planar, liang2026generalizedmathbbzptoriccodes}.

%%%%%%%%%%%%%%%%%%%%%%%%%%%%%%%%%%%%%%%%%%%%%%%%%%%%%%%
\prlsection{Concave partition and spurious-free TEE}
We improve the Levin-Wen prescription that extracts the TEE by the conditional mutual information~\cite{Levin2006Detecting}
\begin{equation}
    \gamma =\frac{1}{2} I(A:C|B)_\rho := \frac{1}{2} (S_{AB}+S_{BC}-S_{ABC}-S_B)_\rho,
    \label{eq: TEE}
\end{equation}
where $ABC$ refers to a partition of a region that is topologically equivalent to an annulus, illustrated in Fig.~\ref{fig:ABCD_partition}. In this work, we focus on translation-invariant stabilizer codes, which can always be coarse-grained to a square lattice with qudits on each site~\cite{haah_module_13}.
The standard rectangular partition (Fig.~\ref{fig: Levin-Wen(a)}) can yield spurious contributions in translation-invariant models~\cite{Williamson2019sTEEsubsystem}. We show that a simple geometric deformation of the partition (Fig.~\ref{fig: Levin-Wen(b)}) eliminates these spurious contributions and recovers the genuine TEE:
\begin{theorem}\label{thm: improved Levin-wen}
Consider a translation-invariant topological\footnote{Here, the topological order condition means that any local operator commuting with all stabilizers is a finite product of stabilizers.}
stabilizer code on a square lattice of $\mathbb{Z}_p$ (with prime $p$) qudits with $q$ qudits per site.
Assume that the stabilizer generators have range $r$, i.e., each generator is supported within an $r\times r$ square.
Then, for the concave partition shown in Fig.~\ref{fig: Levin-Wen(b)}, if the linear size $L$ satisfies
\begin{equation}
        L \ge 32r^3q^4+2r^4+27r + 1~,
\label{eq: L bound for concave partition}
\end{equation}
the 
conditional mutual information $\gamma$ in Eq.~\eqref{eq: TEE}
is guaranteed to be free of spurious contributions. 
\end{theorem}

\begin{figure}[t]
\centering
\subfigure[Rectangular partition]{\includegraphics[scale=0.08]{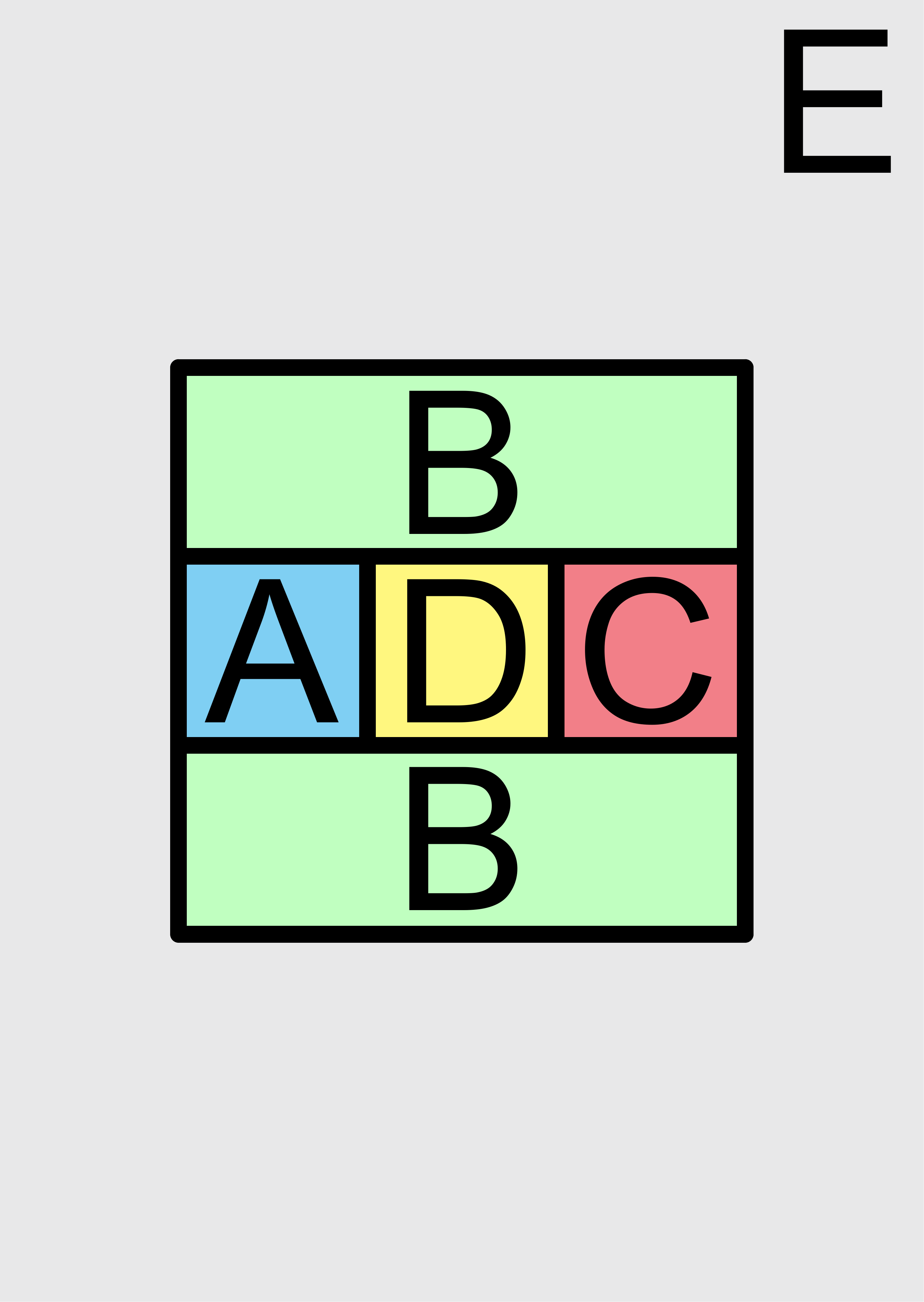}\label{fig: Levin-Wen(a)}
}
\hspace{1em}
\subfigure[Concave partition]{\includegraphics[scale=0.08]{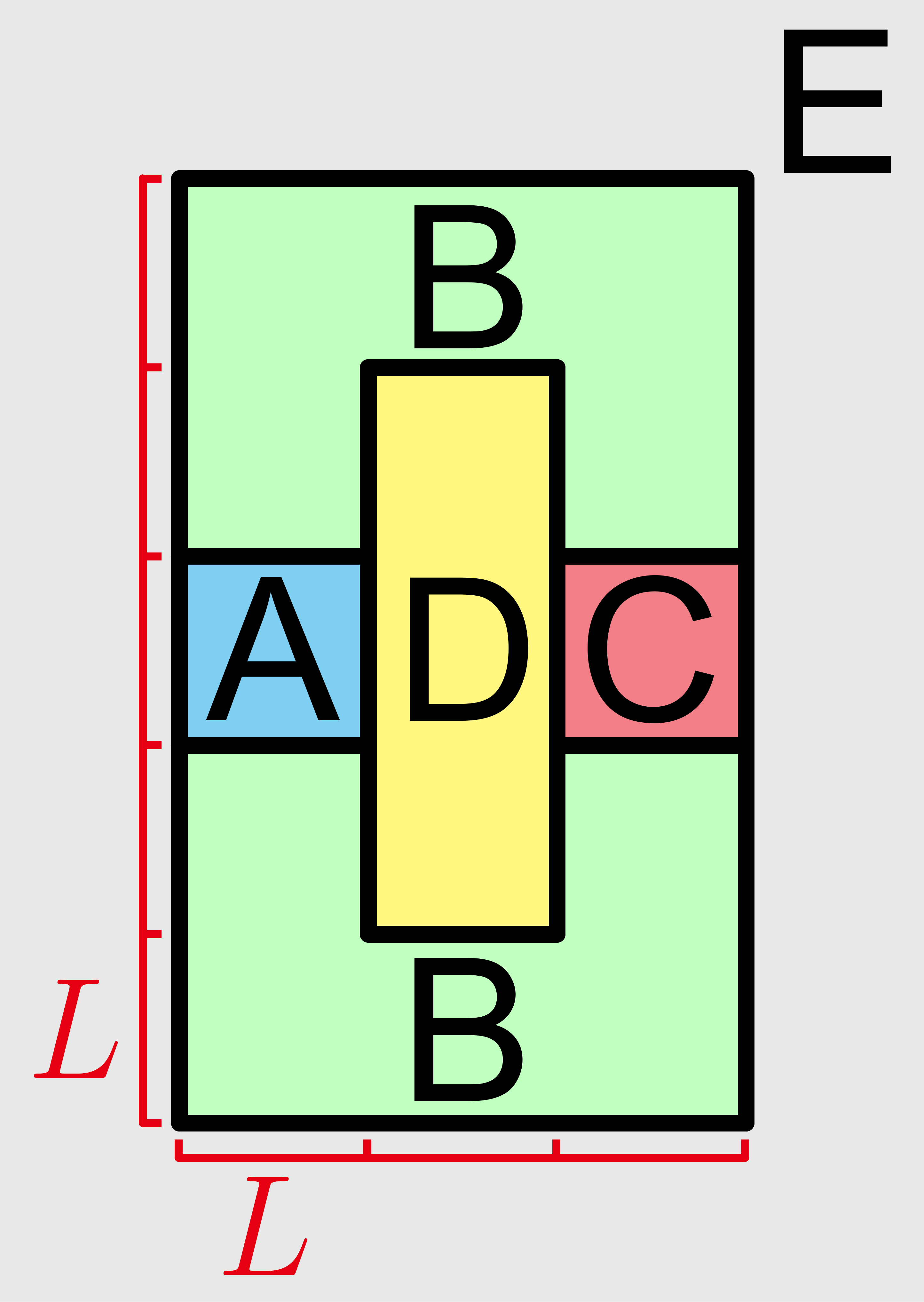}\label{fig: Levin-Wen(b)}}
\caption{Two tripartitions for computing the Levin-Wen topological entanglement entropy. (a) The standard rectangular partition. (b) The concave partition introduced in this work. Our main result shows that the concave partition is immunized against spurious contributions to the topological entanglement entropy.}
\label{fig:ABCD_partition}
\end{figure}

\begin{figure*}[thb]
    \centering
    \subfigure{
    \begin{overpic}[scale=0.07]{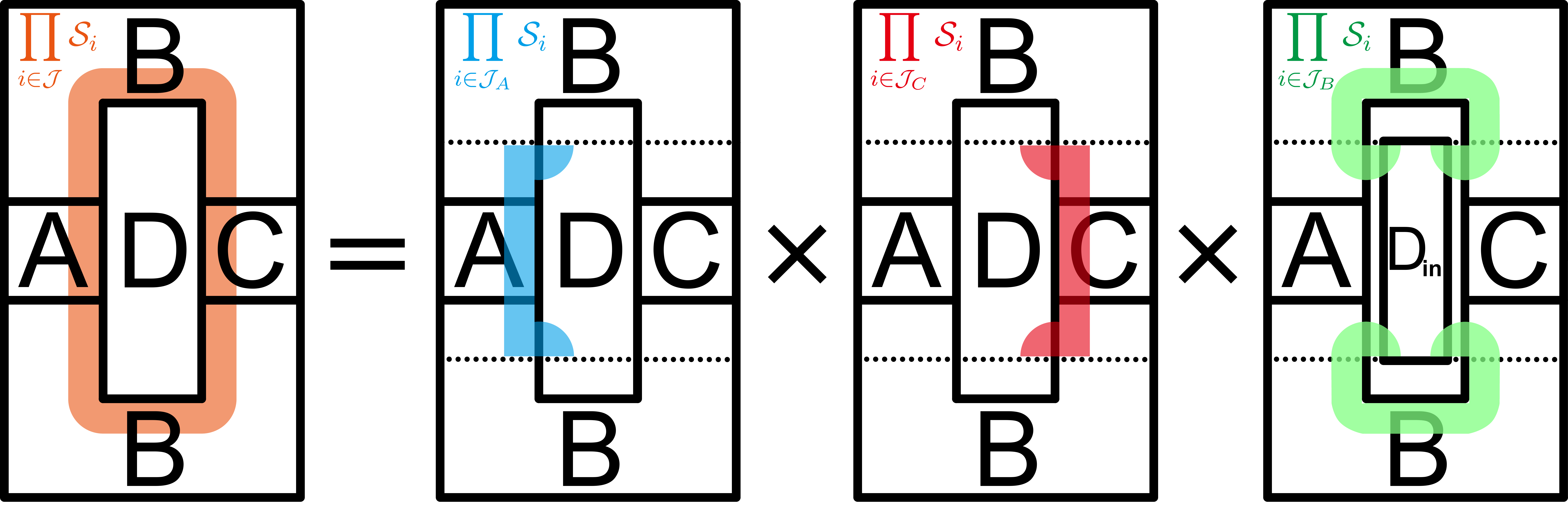}
        \put(-5,34){(a)}
    \end{overpic}
    \label{fig: G_decomposition_ABC}
    }
    \hspace{0.6cm}
    \subfigure{
    \begin{overpic}[scale=0.07]{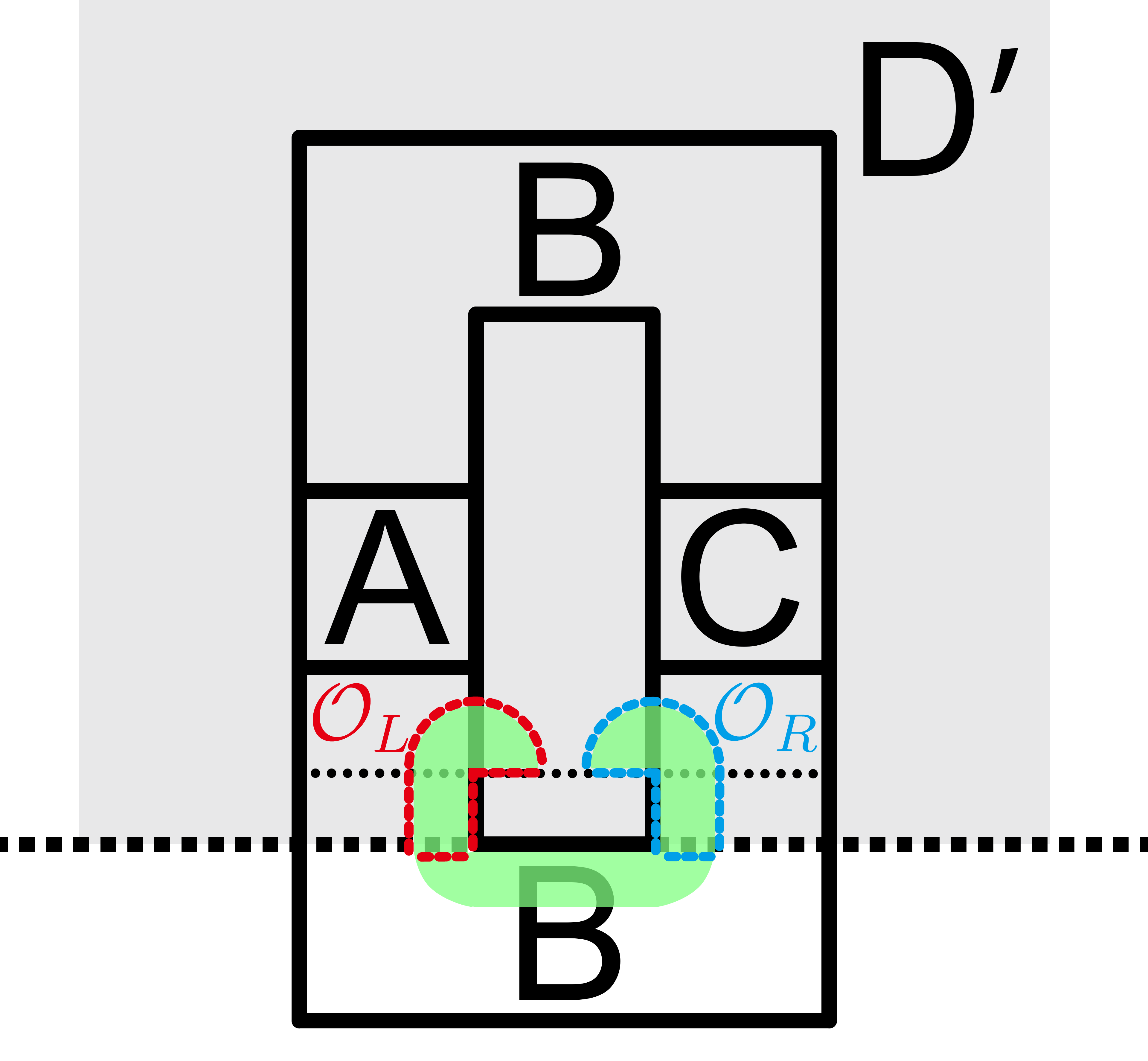}
        \put(-6,84){(b)}
    \end{overpic}
    \label{fig: Omega_L_R}
    }
    \caption{(a) Decomposition of a stabilizer $\mathcal{S} = \prod_{i \in \mathcal{J}} \mathcal{S}_i$, supported on $A\cup B\cup C$, into partial products over the generator subsets $\mathcal{J}_A$, $\mathcal{J}_B$, and $\mathcal{J}_C$. The subset $\mathcal{J}_B$ consists of stabilizer generators near the top and bottom boundaries of $D$: its top part contains all generators in $\mathcal{J}$ that have nonvanishing overlap with the semi-infinite plane above the top boundary of $D_\mathrm{in}$, while its bottom part is defined analogously. The remaining generators in $\mathcal{J}$ are assigned to $\mathcal{J}_A$ or $\mathcal{J}_C$ according to the side on which they are supported. The colored regions indicate the supports of the corresponding partial products $\prod_{i\in\mathcal{J}_I} \mathcal{S}_i$. At this stage, these partial products still overlap with region $D$ near the corners.
    (b) Restriction of the lower portion of $\prod_{i\in\mathcal{J}_B} \mathcal{S}_i$ to the semi-infinite plane $D'$ above the lower boundary of $D$. Provided that $D$ is sufficiently wide, the restricted operator factorizes as $\mathcal{O}_L \mathcal{O}_R$, where $\mathcal{O}_L$ and $\mathcal{O}_R$ are localized near the left and right corners, respectively, and each commutes with all bulk stabilizers in $D'$.
    }
\end{figure*}

We emphasize that this bound is a worst-case estimate derived from Dub\'e’s theorem on the degree growth of Gr\"obner bases and its extension to submodules, as discussed below and in Appendix~\ref{app:Grobner}. In practice, this bound is very loose.
For a given model, the actual size required for the concave partition can instead be obtained from the constructive proofs in the Appendices, with the algorithms for determining the effective range of the bulk stabilizer generators near the boundary and the size of the anyon string operators given in Refs.~\cite{liang2023extracting, liang2024operator}.

We now sketch the proof of Theorem~\ref{thm: improved Levin-wen} based on the following lemmas, with details deferred to the Appendices.

\begin{lemma}\label{lemma: TEE contribution}
For stabilizer codes, a nonzero Levin-Wen topological entanglement entropy arises from each independent stabilizer $\mathcal{S}$ supported in $A \cup B \cup C$ that cannot be factorized as
\begin{equation}
    \mathcal{S} = \mathcal{S}_{AB}\,\mathcal{S}_{BC},
    \label{eq: indivisibility condition}
\end{equation}
where $\mathcal{S}_{AB}$ and $\mathcal{S}_{BC}$ are stabilizers supported on $A \cup B$ and $B \cup C$, respectively. We refer to this failure of factorization as the \textbf{indivisibility condition}, and call such a stabilizer $\mathcal{S}$ \textbf{indivisible}.
\end{lemma}

Lemma~\ref{lemma: TEE contribution} follows directly from the counting formula for the entanglement entropy of stabilizer codes introduced in Appendix~\ref{app: Computing entanglement entropy}.

Given a set of stabilizer generators, an \textbf{anyon} is defined as a violation of finitely many stabilizer generators. Two anyons belong to the same type (superselection sector) if their violation patterns differ by the action of a local Pauli operator~\cite{kitaev2006anyons,liang2023extracting, liang2024operator}.
A stabilizer is said to be fully supported in a region if all of its Pauli operators are supported entirely within that region. Similarly, an anyon is said to be located in a region if the position\footnote{The position of a stabilizer generator is the coordinate $(x_{\min}, y_{\min})$, where $x_{\min}$ and $y_{\min}$ are the minimal $x$- and $y$-coordinates among the Pauli operators contained in the generator.} of every stabilizer generator it violates lies within that region.
A trivial anyon, which belongs to the same equivalence class as the vacuum, is created by a local operator. In two spatial dimensions, every nontrivial anyon can be created by a semi-infinite string operator, namely, a stringlike operator extending from the anyon to infinity~\cite{bombin_Stabilizer_14, haah_classification_21, ruba2024homological}. 

We now describe a necessary and sufficient condition for spurious topological entanglement entropy.
\begin{lemma}\label{lemma: sTEE contribution}
Let \(D_{\mathrm{in}}\subset D\) be an inner region large enough to support representatives of all anyon types. Assume that the geometry satisfies the following two conditions:
\begin{enumerate}
    \item every trivial anyon supported in \(D_{\mathrm{in}}\) can be created by a local Pauli operator supported in \(D\);
    \item every nontrivial anyon supported in \(D_{\mathrm{in}}\) can be created by a semi-infinite string operator supported in \(A\cup D\cup E\), as shown in Fig.~\ref{fig:ABCD_partition}.
\end{enumerate}
Then, the stabilizers $\mathcal{S}$ supported in $A \cup B \cup C$ that contribute to the spurious Levin-Wen topological entanglement entropy are precisely those that satisfy the indivisibility condition and commute with every such string operator.
Equivalently, such a stabilizer $\mathcal{S}$ does not satisfy Eq.~\eqref{eq: indivisibility condition} and can be written as
\begin{equation}
    \mathcal{S} = \prod_{j \in \mathcal{J}} \mathcal{S}_j,
\label{eq: S as product of stabilizer generators}
\end{equation}
where $\mathcal{J}$ is a finite set of stabilizer generators such that no generator $\mathcal{S}_j$ is fully supported in $D_{\mathrm{in}}$.\footnote{Each $\mathcal{S}_j$ in Eq.~\eqref{eq: S as product of stabilizer generators} may still overlap with $D_{\mathrm{in}}$, provided that its support inside $D$ cancels out in the product.}
\end{lemma}

Lemma~\ref{lemma: sTEE contribution} is motivated by the proof of the lower bound on the conditional mutual information in Ref.~\cite{Levin2024TEE}.  
A detailed proof is given in the Appendix~\ref{proof: lemma 3}.

Next, we show that once the concave partition is chosen, any stabilizer $\mathcal{S}$ satisfying Eq.~\eqref{eq: S as product of stabilizer generators} is divisible as Eq.~\eqref{eq: indivisibility condition}, so spurious TEE does not occur.
Starting from the decomposition of $\mathcal{S}$ in Eq.~\eqref{eq: S as product of stabilizer generators}, we further partition the set $\mathcal{J}$ into disjoint subsets,
\begin{equation}
    \mathcal{J} = \mathcal{J}_A \cup \mathcal{J}_B \cup \mathcal{J}_C ,
\end{equation}
as illustrated in Fig.~\ref{fig: G_decomposition_ABC}.
The set $\mathcal{J}_B$ consists of all stabilizer generators in $\mathcal{J}$ near the top and bottom boundaries of $D$.
All remaining stabilizer generators in $\mathcal{J}$ are assigned to $\mathcal{J}_A$ or $\mathcal{J}_C$, depending on whether they are left or right.
At this stage, the resulting factorization is not of the form of Eq.~\eqref{eq: indivisibility condition}. The middle portion of each product $\prod_{j \in \mathcal{J}_I} \mathcal{S}_j$ has no overlap with region $D$, since these operators are locally identical to $S=\prod_{j \in \mathcal{J}} \mathcal{S}_j$, which is fully supported in $A \cup B \cup C$; however, near the corners of region $D$, they overlap with $D$. We will further modify each $\mathcal{J}_I$ by multiplying appropriate stabilizer generators, so that the resulting products become fully disjoint from region $D$.

\begin{figure}[thb]
    \centering
    \subfigure[ ]{\includegraphics[scale=0.05]{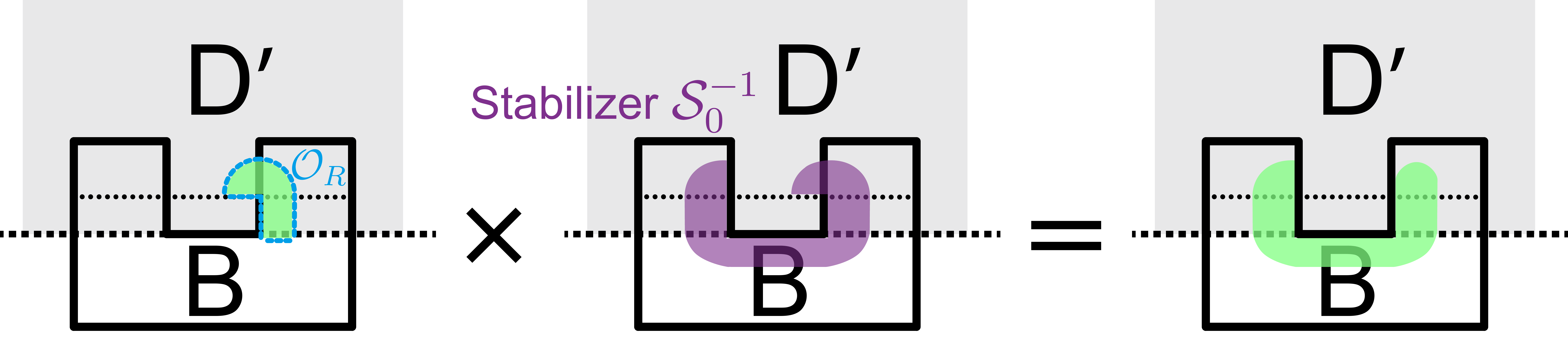}\label{fig: boundary_gauge_moving1}
    }
    \\
    \subfigure[ ]{\includegraphics[scale=0.05]{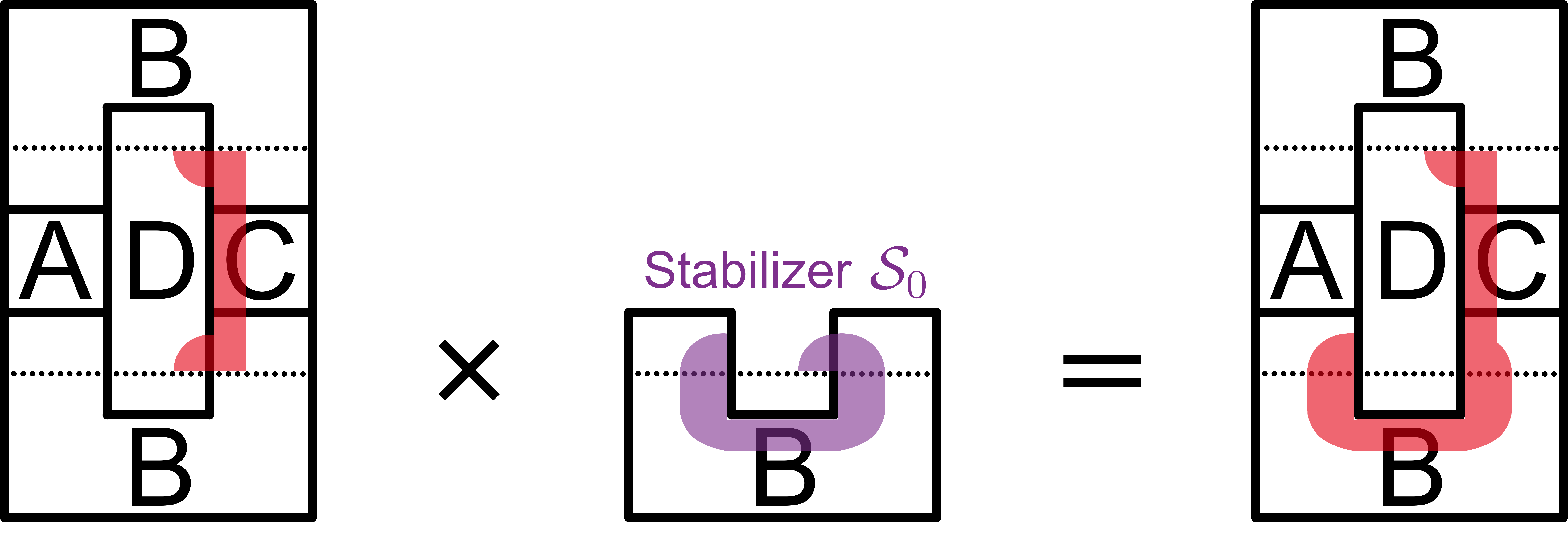}\label{fig: boundary_gauge_moving2}
    }
    \caption{(a) By multiplying an appropriate stabilizer $\mathcal{S}_0^{-1}$ near the bottom boundary, the support of $\mathcal{O}_R$ is shifted into region $B$, removing the residual support of $\prod_{i\in\mathcal{J}_B}\mathcal{S}_i$ in region $D$ near the lower-right corner. (b) Correspondingly, multiplying $\mathcal{S}_0$ removes the residual support of $\prod_{i\in\mathcal{J}_C}\mathcal{S}_i$ in region $D$ near the lower-right corner, while leaving the total product of $\prod_{i\in\mathcal{J}_B}\mathcal{S}_i \prod_{i\in\mathcal{J}_C}\mathcal{S}_i$ unchanged.}
    \label{fig: boundary_gauge_moving}
\end{figure}

We introduce the relevant terminology~\cite{liang2024operator}. Given a region, the \textbf{bulk stabilizers} consist of all local stabilizers supported fully in that region.
This includes not only the original stabilizer generators fully supported in the region, but also products of generators whose support outside the region cancels.
The \textbf{boundary gauge operators} are local operators that commute with all bulk stabilizers, modulo multiplication by bulk stabilizers. Among them, the \textbf{primary boundary gauge operators} are those inherited from truncated local stabilizers of the infinite plane, while the \textbf{secondary boundary gauge operators} are those that cannot be represented in this way~\cite{liang2024operator}.

We first focus on the lower portion of the product $\prod_{i \in \mathcal{J}_B} \mathcal{S}_i$. Let $D'$ denote the semi-infinite plane that shares the lower boundary of $D$, as shown in Fig.~\ref{fig: Omega_L_R}. Restricting the product to $D'$, we obtain a decomposition into left and right components,
\begin{equation}
    \prod_{i \in \mathcal{J}_B^{\mathrm{(bottom)}}} \mathcal{S}_i \Big|_{D'}
    = \mathcal{O}_L \,\mathcal{O}_R~.
\end{equation}
The combined operator $\mathcal{O}_L \mathcal{O}_R$ is a boundary gauge operator, commuting with all bulk stabilizers supported in $D'$. 
The generating set of bulk stabilizers in $D'$ may contain elements beyond the original stabilizer generators, but their spatial range remains bounded:
\begin{equation}
    r_\mathrm{bulk} \le 2(r+q)^4 ,
\end{equation}
as proved in Appendix~\ref{app:Grobner}.

Once the linear size of $D$ is sufficiently large, no single bulk stabilizer generator can overlap both $\mathcal{O}_L$ and $\mathcal{O}_R$. Therefore, $\mathcal{O}_L$ and $\mathcal{O}_R$ are each boundary gauge operators. To characterize the structure of such operators, we introduce the following lemma, whose proof is given in Appendix~\ref{app: The proof of lemmas and theorems}.

\begin{lemma}\label{lemma: height limit}
Consider a two-dimensional translation-invariant topological stabilizer code whose stabilizer generators have range $r$, with the upper half-plane $y \ge 0$ taken to be the bulk region. Then every boundary gauge operator supported in the region $0 \le y \le h$, where $h\ge r$, admits a decomposition as an (infinite) product of
\begin{enumerate}
    \item bulk stabilizer generators in the upper half-plane;
    \item stabilizer generators of the infinite plane, truncated to their support in $y\ge 0$.
\end{enumerate}
Moreover, every factor is supported in
\begin{equation}
y \le 8r^3q^4 + 5r + 3h~.
\end{equation}
\end{lemma}

Using this lemma, we construct an (infinite) product of stabilizer generators whose restriction to the upper half-plane $D'$ is precisely $\mathcal{O}_R$. From this product, we keep only those stabilizer generators that overlap with $D$, and denote their product by $\mathcal{S}_0$. We then multiply $\prod_{i \in \mathcal{J}_B} \mathcal{S}_i$ by $\mathcal{S}_0^{-1}$ to cancel its support in $D$, and multiply $\prod_{i \in \mathcal{J}_C} \mathcal{S}_i$ by $\mathcal{S}_0$ so that the total product remains unchanged, as illustrated in Fig.~\ref{fig: boundary_gauge_moving}. After this transformation, neither operator has support near the lower-right corner of $D$. Repeating the same construction at the other corners of $D$, we obtain a new index set $\mathcal{J}'_B$ such that $\prod_{i \in \mathcal{J}'_B} \mathcal{S}_i$ is fully supported in region $B$.

After redistributing all such stabilizer generators, we arrive at updated index sets $\mathcal{J}'_A$, $\mathcal{J}'_B$, and $\mathcal{J}'_C$ such that $\prod_{i \in \mathcal{J}'_A} \mathcal{S}_i$ is supported in $A \cup B$, while $\prod_{i \in \mathcal{J}'_C} \mathcal{S}_i$ is supported in $B \cup C$. The stabilizer $\mathcal{S}$ therefore factorizes as
\begin{equation}
    \mathcal{S}
    =
    \left( \prod_{i \in \mathcal{J}'_A} \mathcal{S}_i \prod_{i \in \mathcal{J}'_B} \mathcal{S}_i \right)
    \left( \prod_{i \in \mathcal{J}'_C} \mathcal{S}_i \right)
    :=
    \mathcal{S}_{AB}\,\mathcal{S}_{BC},
\end{equation}
which satisfies Eq.~\eqref{eq: indivisibility condition}. Hence, any stabilizer $\mathcal{S}$ satisfying Eq.~\eqref{eq: S as product of stabilizer generators} does not contribute to the spurious topological entanglement entropy.

In the argument above, the concave geometry of region $B$ is essential: it allows a secondary boundary gauge operator $\mathcal{O}_R$ to be absorbed into region $B$.\footnote{By contrast, for a primary boundary gauge operator, one can directly multiply by a stabilizer to cancel its support in $D'$.}
For the standard rectangular partition, however, a secondary boundary gauge operator cannot be absorbed into $B$, and the above argument therefore breaks down. It follows that any spurious contribution to the TEE must involve a secondary boundary gauge operator in $D'$.
Conversely, given a secondary boundary gauge operator, Lemma~\ref{lemma: height limit} implies that it can be written as an infinite product. By truncating this product to a sufficiently long finite segment, one obtains secondary boundary gauge operators localized near the two ends, chosen to lie in regions $A$ and $C$. This construction yields a spurious TEE.
Taking into account both the upper and lower boundaries of region $D$, we arrive at the following corollary:
\begin{corollary}
For a sufficiently large standard rectangular partition shown in Fig.~\ref{fig: Levin-Wen(a)}, each nontrivial secondary boundary gauge operator supported in the upper half-plane contributes \(1/2\) to the spurious Levin-Wen topological entanglement entropy, and the same holds for each one supported in the lower half-plane.

Equivalently, the standard rectangular partition is free of spurious contributions if and only if every boundary gauge operator in the upper or lower half-plane is a truncated stabilizer.
\end{corollary}
Appendix~\ref{app: examples} presents several \(\mathbb{Z}_2\) examples, in which we explicitly compute the TEE on different geometries, and the corresponding secondary boundary gauge operators, thereby verifying this corollary.

Finally, to obtain the bound in Eq.~\eqref{eq: L bound for concave partition} of Theorem~\ref{thm: improved Levin-wen}, it is not enough to require that the size of $D$ exceed the range $r_\mathrm{bulk}$ of the bulk stabilizers. One must also account for the size of the anyons and the support of the Pauli operators that create them.
This is captured by the following lemma, proved in Appendix~\ref{app: The proof of lemmas and theorems}.
\begin{lemma}\label{lemma: string size bound}
In a two-dimensional translation-invariant topological stabilizer code with $q$ qudits per unit cell and stabilizer generators of range $r$, every trivial anyon fully supported in an $r' \times r'$ square can be created by a Pauli operator supported in the concentric square of linear size $2r' + 8r^3q^4 + 2r$.
\end{lemma}

For each anyon type, the cleaning lemma on a sufficiently large torus gives an infinite string operator of width at most $r$, whose truncation creates endpoint anyons of range at most $r'=2r$~\cite{haah2016algebraic, Ellison2023paulitopological}; see also Appendix~\ref{app: proof of theorem 1}.
To obtain a single-generator violation in $D_{\mathrm{in}}$ labeled by anyon $\alpha$, we truncate such a thin string of the corresponding anyon type so that one endpoint creates an anyon $\beta$ in $D_{\mathrm{in}}$, at the position of $\alpha$, with $\alpha-\beta$ being a trivial anyon. This trivial anyon can then be created by a local Pauli operator $\mathcal{O}$. Decorating the truncated string by $\mathcal{O}$, we obtain a new string operator whose endpoint creates the desired single-generator violation $\alpha$. By Lemma~\ref{lemma: string size bound}, this endpoint decoration is supported within an additional distance $8r^3q^4+6r$ from the endpoint.
Consequently, provided the separation between $D_{\mathrm{in}}$ and $D$ is at least $8r^3q^4+6r$, and the width of $A$ exceeds $2r$, every anyon supported in $D_{\mathrm{in}}$ can be created by a semi-infinite string operator supported in $A \cup D \cup E$.
This also implies that every trivial anyon supported in $D_{\mathrm{in}}$ can be created by a local Pauli operator supported in $D$, by bringing the other endpoints of the thin strings together inside $D_{\mathrm{in}}$ and applying Lemma~\ref{lemma: string size bound} again.

Next, we choose the linear size of $D_{\mathrm{in}}$ to be $2(r+q)^4+2r$, so that no bulk stabilizer overlaps both $\mathcal{O}_L$ and $\mathcal{O}_R$. The extra $2r$ provides a buffer, since $\mathcal{O}_L$ and $\mathcal{O}_R$ can each penetrate into $D_{\mathrm{in}}$ by depth at most $r$ from opposite sides. It therefore suffices to choose the size of $D$ such that
\begin{eqs}
    L \ge &~ 2(r+q)^4 + 2r + 2(8r^3 q^4 + 6r) \\
    = &~ 2(r+q)^4 + 16r^3q^4 + 14r~.
    \label{eq: inequality 1}
\end{eqs}
In addition, substituting the height bound for $\mathcal{O}_L$ and $\mathcal{O}_R$, namely $h=8r^3q^4+7r$, into Lemma~\ref{lemma: height limit} shows that the decorating stabilizer $\mathcal{S}_0$ has height less than $32r^3q^4+27r$. For this decoration to be fully supported in $B$, its height must be smaller than the vertical separation between the lower boundaries of $D$ and $A$. We therefore also require
\begin{equation}
    L \ge 32r^3q^4 + 27r~.
    \label{eq: inequality 2}
\end{equation}
For $q\ge 1$ and $r\ge 1$, both inequalities~\eqref{eq: inequality 1} and~\eqref{eq: inequality 2} are satisfied by choosing
\begin{equation}
    L \ge 32r^3q^4 + 2r^4 + 27r +1~.
\end{equation}
This completes the proof of Theorem~\ref{thm: improved Levin-wen}.\footnote{We focus on stabilizers near the inner boundary of \(A\cup B\cup C\); the outer boundary case follows analogously. See Appendix~\ref{app: The proof of lemmas and theorems}.}

%%%%%%%%%%%%%%%%%%%%%%%%%%%%%%%%%%%%%%%%%%%%%%%%%%%%%%%

\prlsection{Cylinder entanglement and topological frustration}
\label{sec: period} 
Having shown that the concave partition removes spurious contributions to the Levin-Wen topological entanglement entropy, we now study a complementary entanglement diagnostic: the bipartite entanglement entropy on an infinite cylinder. We focus on its dependence on the cylinder circumference. Consider an infinite cylindrical surface~\cite{Williamson2019sTEEsubsystem}, shown in Fig.~\ref{fig: geomotry_period}, bipartitioned into regions $A$ and $B$ by a cut perpendicular to the cylinder axis.

\begin{figure}[t]
    \centering
    \begin{minipage}{1\linewidth}
        \centering
        \begin{overpic}[width=\linewidth]{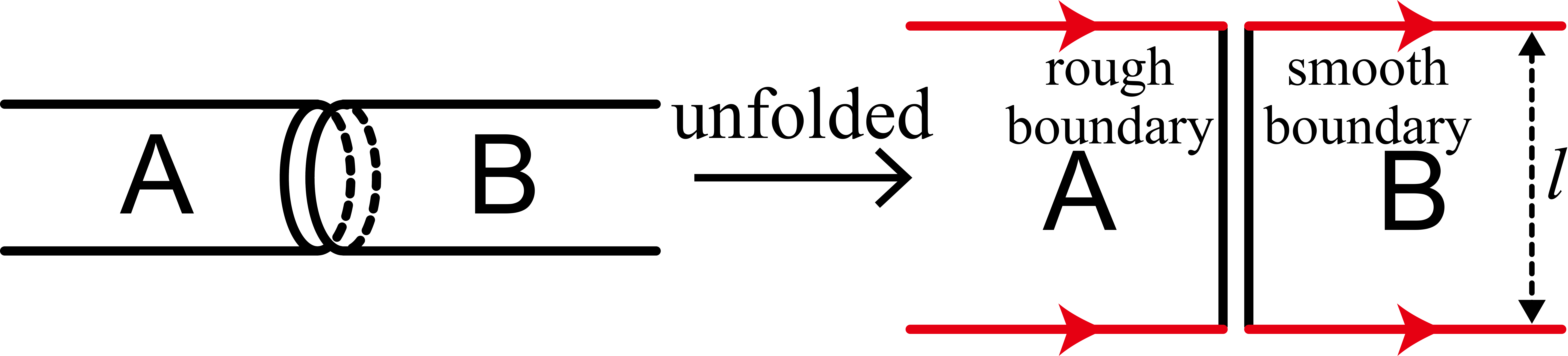}
            \put(0,20){(a)}
        \end{overpic}
    \end{minipage}\\
    \begin{minipage}{\linewidth}
        \centering
        \begin{overpic}[width=\linewidth]{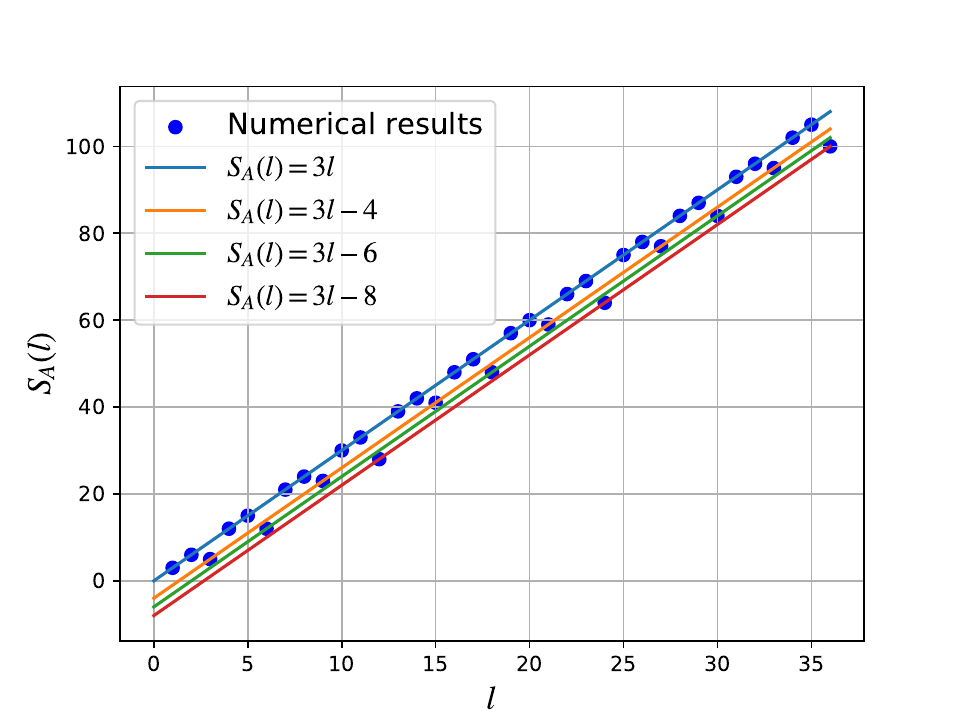}
            \put(0, 63){(b)}
        \end{overpic}
    \end{minipage}
    \caption{Bipartite entanglement entropy on an infinite cylinder. (a) An infinite cylinder is bipartitioned into regions $A$ and $B$ by a cut perpendicular to the cylinder axis. After unfolding, the cut induces a rough boundary on the right edge of region $A$ and a smooth boundary on the left edge of region $B$. The top and bottom edges are identified, and the cylinder circumference is $l$. (b) The entanglement entropy $S_A(l)$ of region $A$ for the $(3,3)$-BB code~\cite{Bravyi2024HighThreshold}. The numerical data separate into four linear branches, $S_A(l)=3l$, $3l-4$, $3l-6$, and $3l-8$, depending on $\gcd(l,12)$. The constant offset is in one-to-one correspondence with $k(l)/2$, where $k(l)$ is the dimension of logical operators winding around the $y$ direction when the cylinder is glued to form a torus.}
    \label{fig: geomotry_period}
\end{figure}

Following the algorithm of Ref.~\cite{fattalEntanglementStabilizerFormalism2004}, we compute the entanglement entropy $S_A(l)$ of region $A$ in Fig.~\ref{fig: geomotry_period} for the $(3,3)$-BB code with stabilizer generators~\cite{Bravyi2024HighThreshold, liang2025generalized}
\begin{eqs}
\begin{gathered}
    \mathcal{S}_1 = \vcenter{\hbox{\includegraphics[scale=.17]{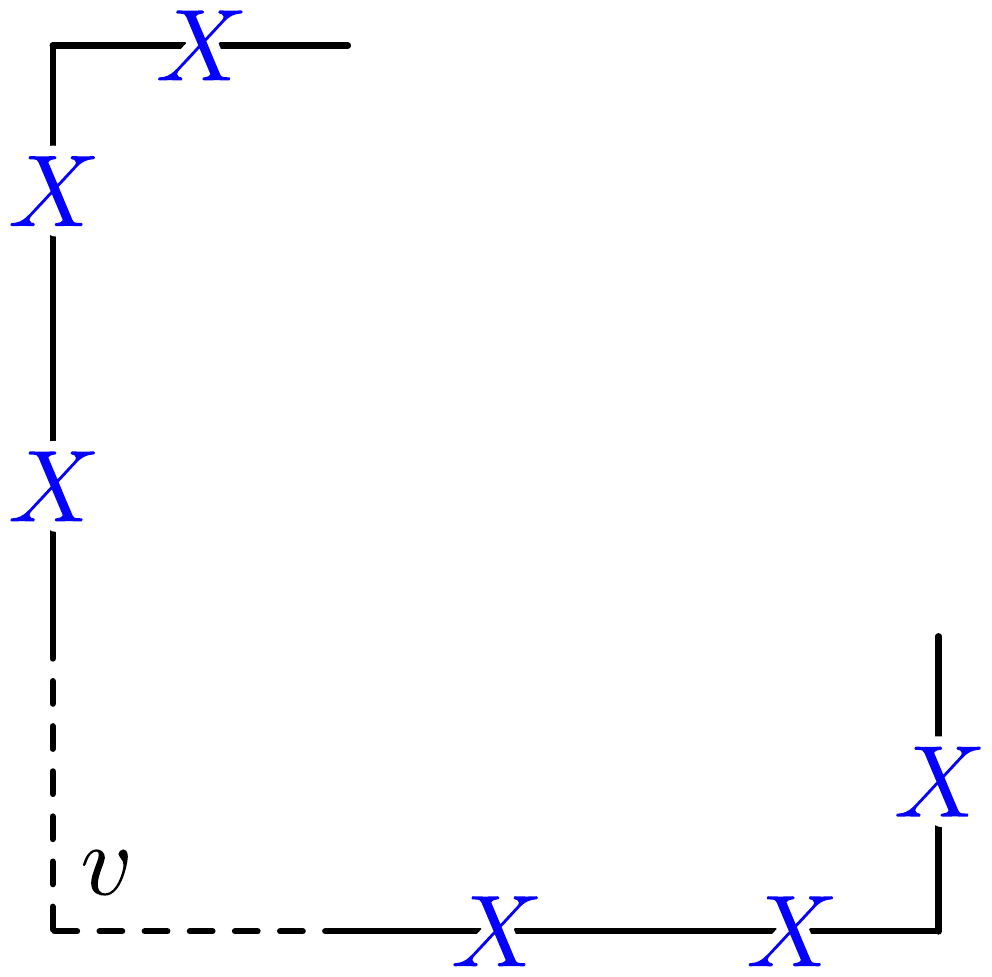}}}~,~
    \mathcal{S}_2 = \vcenter{\hbox{\includegraphics[scale=.17]{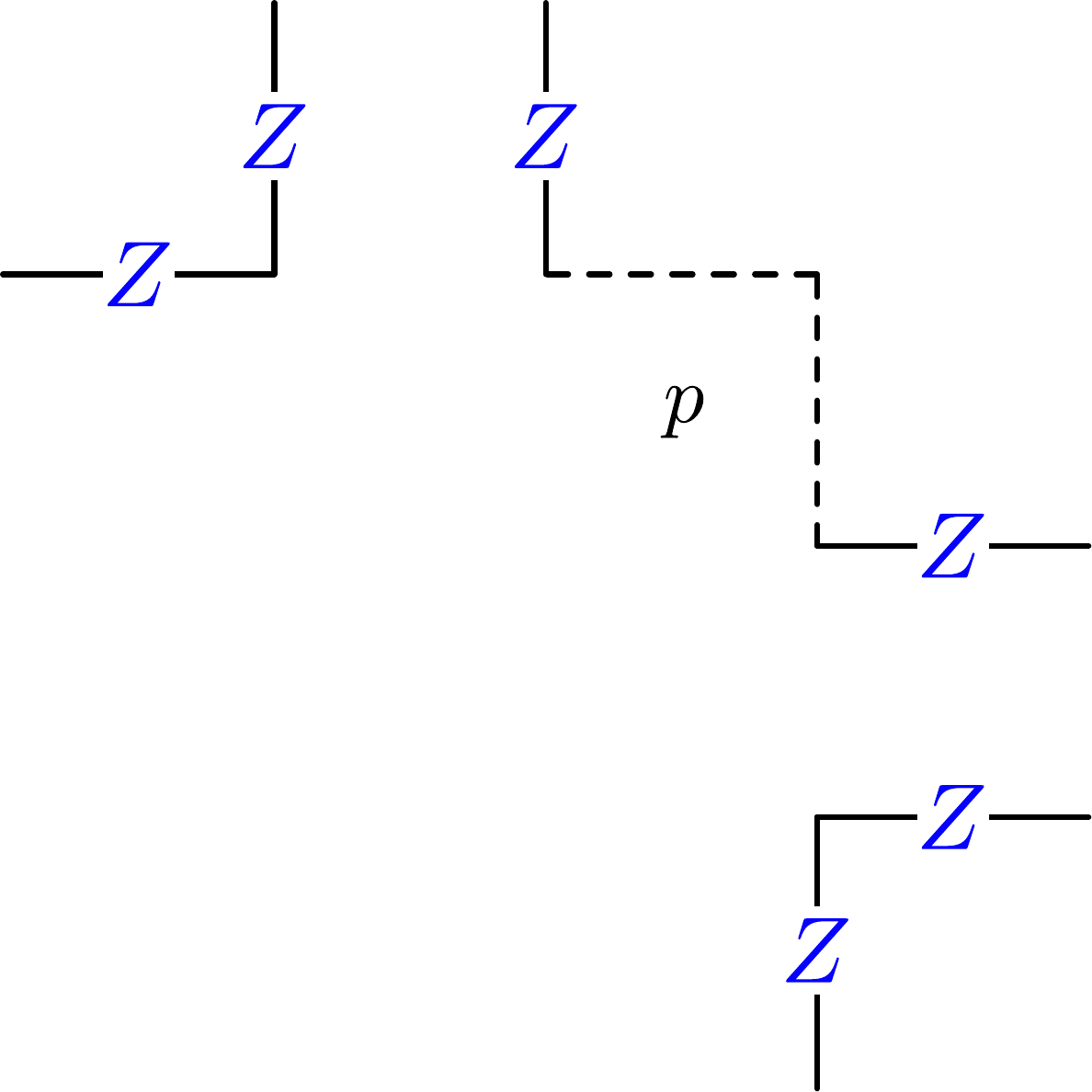}}}~.
\end{gathered}\nonumber
\end{eqs}
As the circumference $l$ varies, the data separate into four linear branches.
Remarkably, the constant offset is proportional to the logical dimension $k(l)$ obtained when the cylinder is glued to form a torus\footnote{Here we assume that the period in the $x$ direction is compatible with the stabilizers, so that it does not remove the logical operators winding in the $y$ direction.}:
\begin{equation}
\bigl(S_A(l),\,k(l)\bigr) =
\begin{cases}
(3l-8,\,16), & \gcd(l,12)=12,\\
(3l-6,\,12), & \gcd(l,12)=6,\\
(3l-4,\, 8), & \gcd(l,12)\in\{3,9\},\\
(3l,\,0),    & \text{otherwise}.
\end{cases}\nonumber
\end{equation}
These results agree with the anyon frustration studied in Refs.~\cite{liang2025generalized, chen2025anyon}, which determines the anyons allowed in the $y$ direction, and hence the logical dimension, for each period $l$.

\prlsection{Summary and outlook}
In this work, we identified the microscopic origin of spurious contributions to the Levin-Wen topological entanglement entropy in Pauli stabilizer models, showing that they arise from secondary boundary gauge operators. For two-dimensional translation-invariant stabilizer codes of prime-dimensional qudits, we proved that the concave partition eliminates these spurious contributions whenever the linear size $L$ exceeds a polynomial bound determined by the stabilizer range $r$ and the number $q$ of qudits per unit cell.
We present several examples in Appendix~\ref{app: examples}, including the shifted toric code, bivariate bicycle (BB) codes, and stabilizer codes with subsystem symmetries. For BB codes, we further show that the dependence of the bipartite entanglement entropy on the cylinder circumference $l$ provides an entanglement signature of topological frustration in high-performance qLDPC codes.

Our results suggest several directions for future work. First, the analytic lower bound on the size of the concave partition is conservative, and it would be valuable to sharpen it and develop more practical, model-dependent criteria for the absence of spurious contributions. Second, extending the proof to general $\mathbb{Z}_d$ stabilizer codes, including the role of zero divisors in the algebraic formalism, remains an important open problem. More broadly, it would be interesting to generalize spurious-free entanglement partitions to higher-dimensional stabilizer models and to entanglement-based diagnostics beyond the Pauli stabilizer setting~\cite{Grover2011Entanglement}.

\prlsection{Acknowledgments}
%We would like to thank XXX for their valuable discussions.
P.H., Z.L., and Y.-A.C. are supported by the National Natural Science Foundation of China (Grant No.~12474491) and the Fundamental Research Funds for the Central Universities, Peking University. 
Y.W. and Y.G. are supported by the National Key R\&D Program of China 2023YFA1406702, NSFC12575022, SRICSPYF-ZY2025157, the Shanghai Committee of Science and Technology grant No. 25LZ2600800 and the Tsinghua Dushi program.

\bibliography{bib.bib}

\appendix
\onecolumngrid
\newpage

%%%%%%%%%%%%%%%%%%%%%%%%%%%%%%%%%%%%%%%%%%%%%%%%%%%%%%%%%%%%%%%%%%%%%%%%%%%%%%

\section{Examples comparing rectangular and concave partitions}
\label{app: examples}

In this Appendix, we compare the Levin-Wen topological entanglement entropy obtained from rectangular and concave partitions for several stabilizer models for $\mathbb{Z}_2$ stabilizer codes on a square lattice, as summarized in Table~\ref{tab: examples}. The examples include the standard toric code, two models obtained from it by CNOT circuits, the two-dimensional cluster state, and the $(3, 3)$-BB code. All of these models, except for the standard toric code, exhibit spurious TEE under the rectangular partition, while the concave partition reproduces the genuine value.
%We then turn to bivariate bicycle (BB) codes~\cite{Bravyi2024HighThreshold}, for which the concave partition again yields the expected result, $\gamma = k/2$.
Finally, we compute the secondary boundary gauge operators on the upper and lower half-planes for these models and find that the total number of their types, taken together on both half-planes, is always equal to twice the spurious TEE. 

\begin{table*}[thb]
\centering
\resizebox{\textwidth}{!}{
\begin{tabular}{|c|c|c|c|c|c|}
\hline
$\begin{gathered}
\text{\large Stabilizers} \\
\text{\large $\mathcal{S}_1$ and $\mathcal{S}_2$}
\end{gathered}$
&
$\begin{gathered}
    \vspace{-1.15em}\\
    \vcenter{\hbox{\includegraphics[scale=.2]{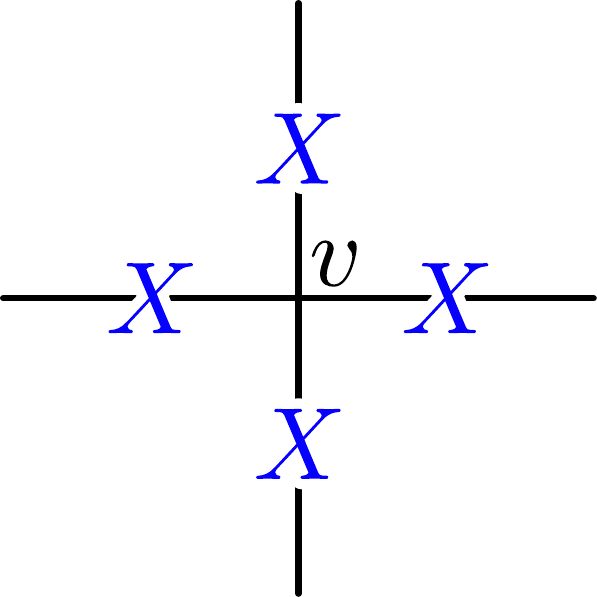}}}\\
    \vspace{2em} \\
    \vcenter{\hbox{\includegraphics[scale=.2]{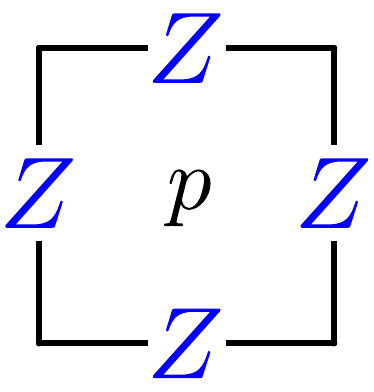}}}\\
    \vspace{-1em}
\end{gathered}$
&
$\begin{gathered}
    \vspace{1.25em}\\
    {\hbox{\includegraphics[scale=.2]{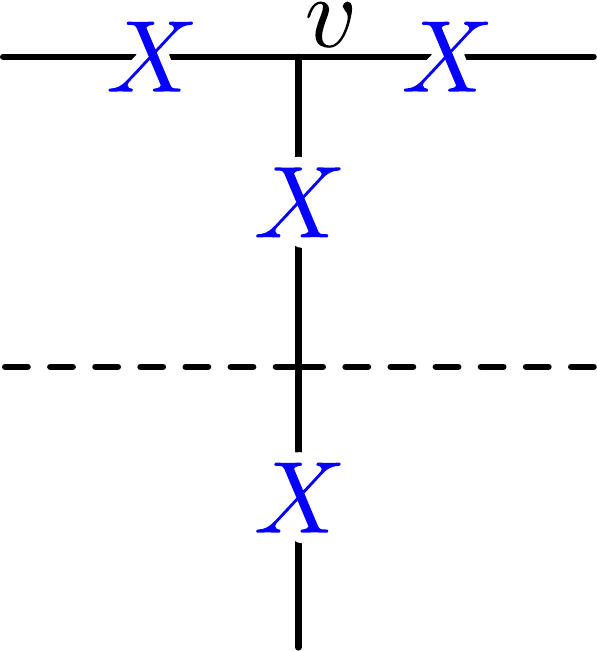}}}\\
    \vspace{2em} \\
    {\hbox{\includegraphics[scale=.2]{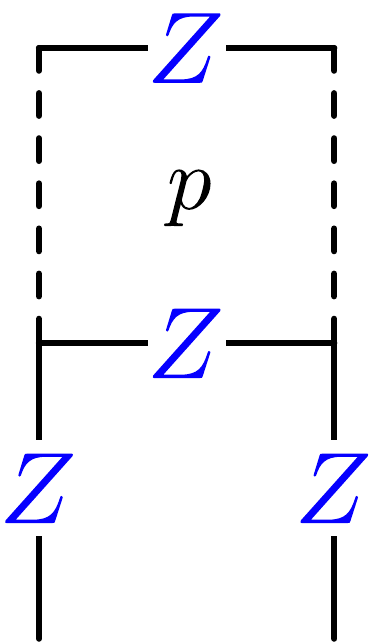}}}\\
    \vspace{-1em}
\end{gathered}$
&
$\begin{gathered}
    \vspace{-1em}\\
    \vcenter{\hbox{\includegraphics[scale=.2]{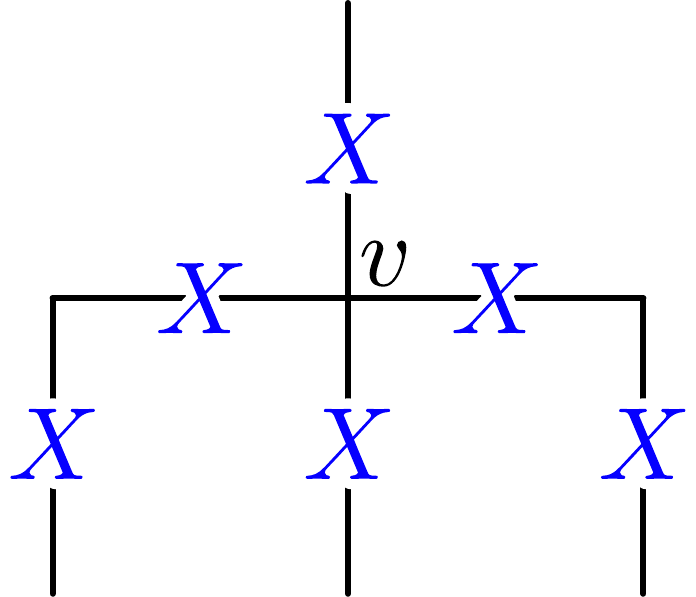}}} \\
    \vspace{2em} \\
    \quad\vcenter{\hbox{\includegraphics[scale=.2]{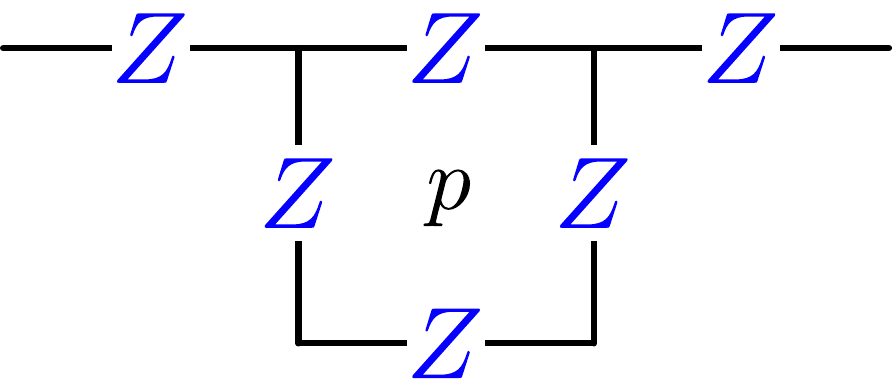}}}\quad \\
    \vspace{-1em}
\end{gathered}$
&
$\begin{gathered}
    \vspace{-1em}\\
    \vcenter{\hbox{\includegraphics[scale=.2]{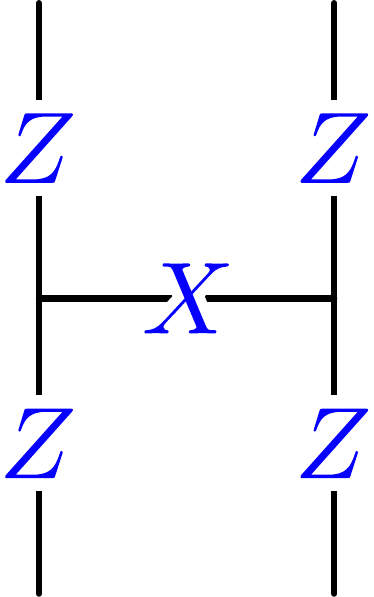}}} \\
    \vspace{2em} \\
    \quad\vcenter{\hbox{\includegraphics[scale=.2]{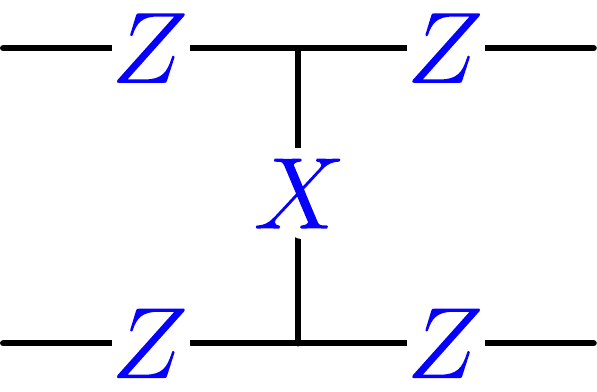}}}\quad\\
    \vspace{-1em}
\end{gathered}$
&
$\begin{gathered}
    \vspace{-1em}\\
    \vcenter{\hbox{\includegraphics[scale=.2]{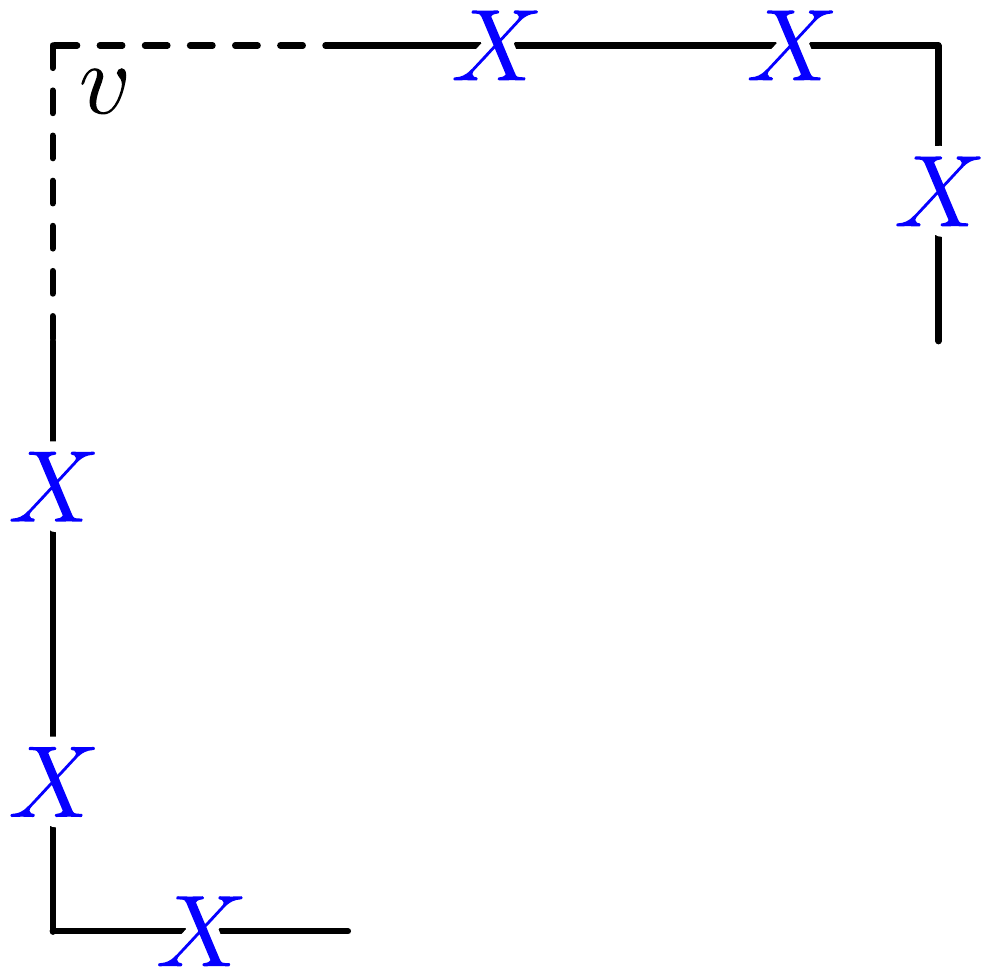}}} \\
    \quad\vcenter{\hbox{\includegraphics[scale=.2]{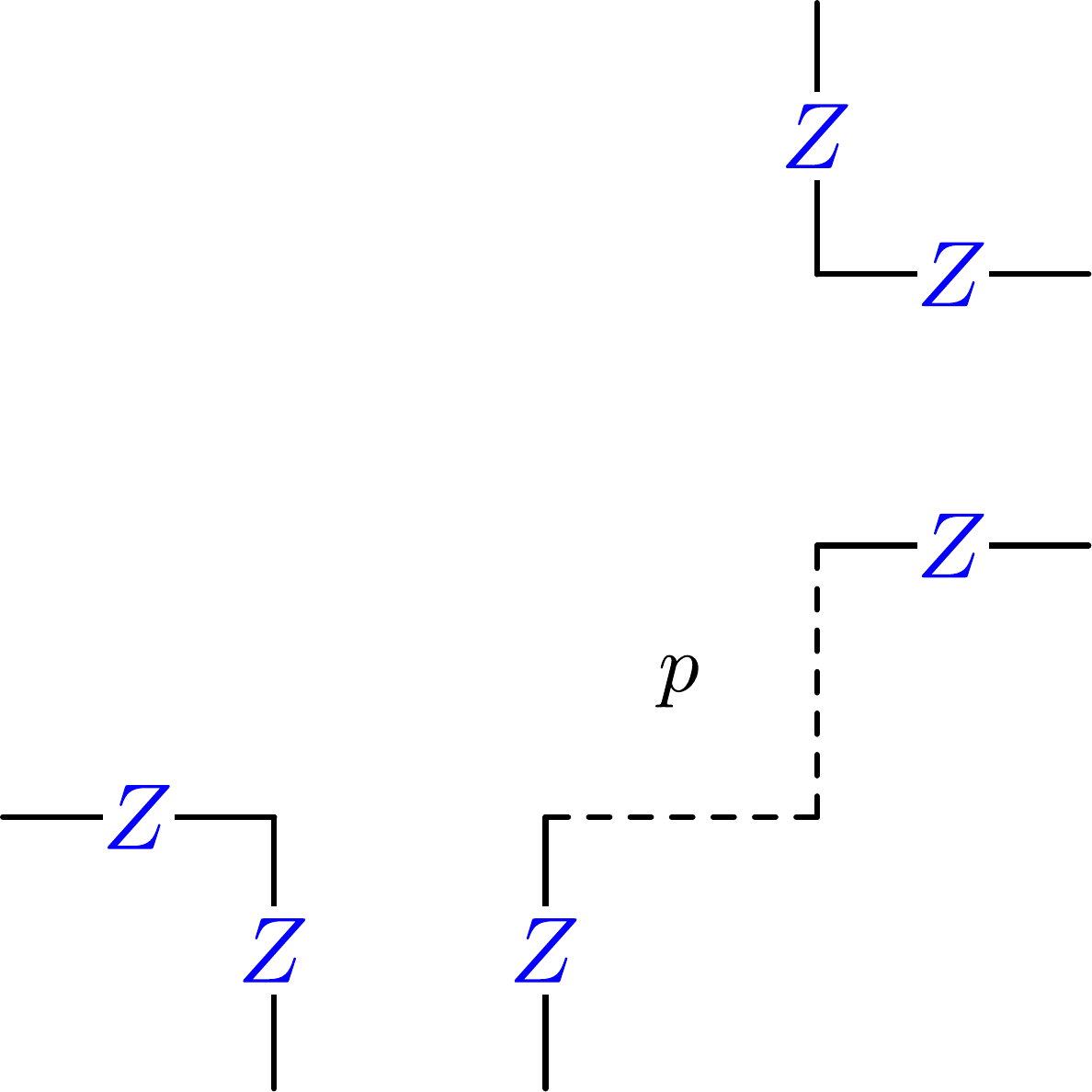}}}\quad\\
    \vspace{-1em}
\end{gathered}$
\\
\hline
$\begin{gathered}
    \vcenter{\hbox{\includegraphics[scale=.03]{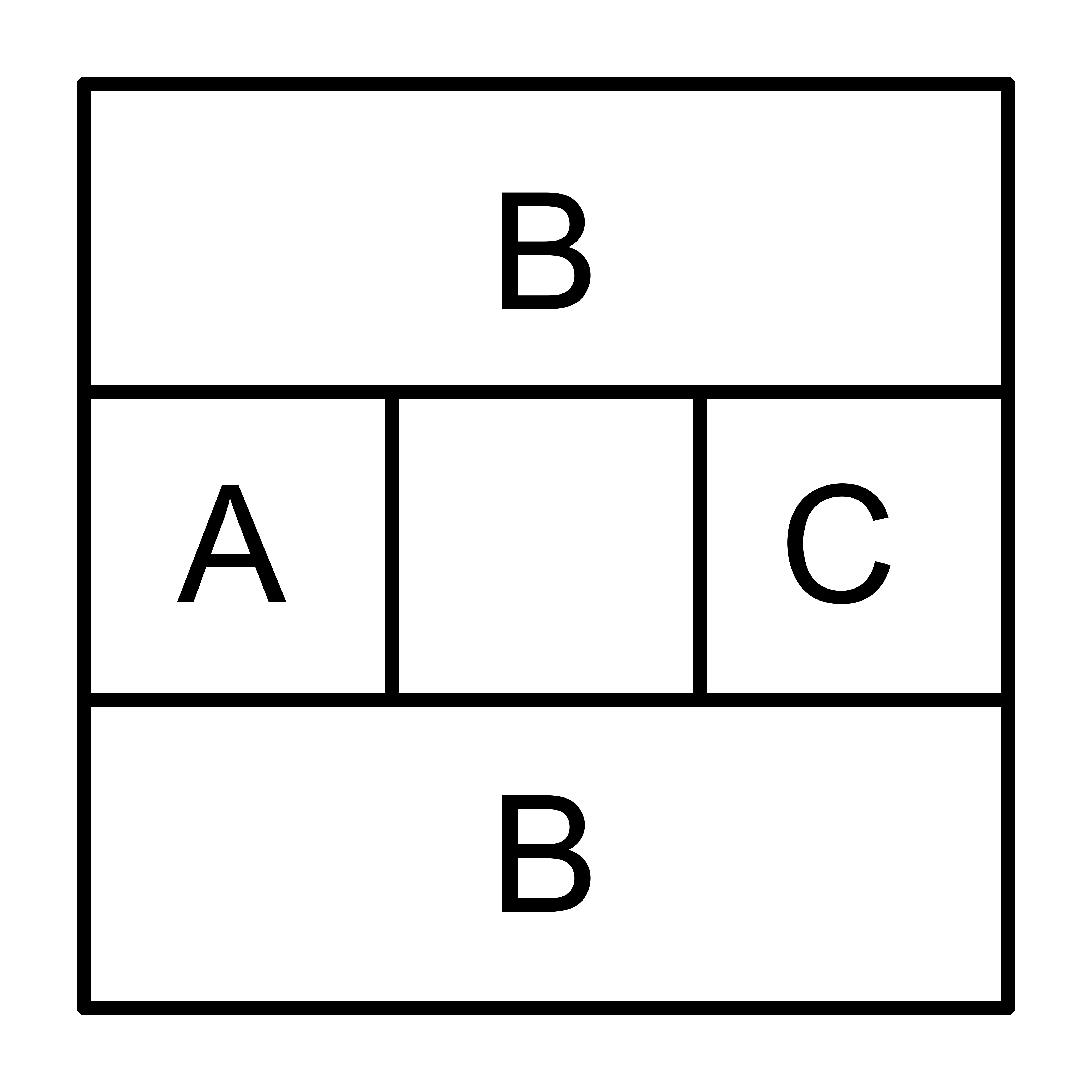}}}
\end{gathered}$
&
\text{\large $\gamma=1$}
&
\text{\large $\gamma=2$}
&
\text{\large $\gamma=3$}
&
\text{\large $\gamma=1$}
&
\text{\large $\gamma=9$}
\\
\hline
$\begin{gathered}
    \vcenter{\hbox{\includegraphics[scale=.03]{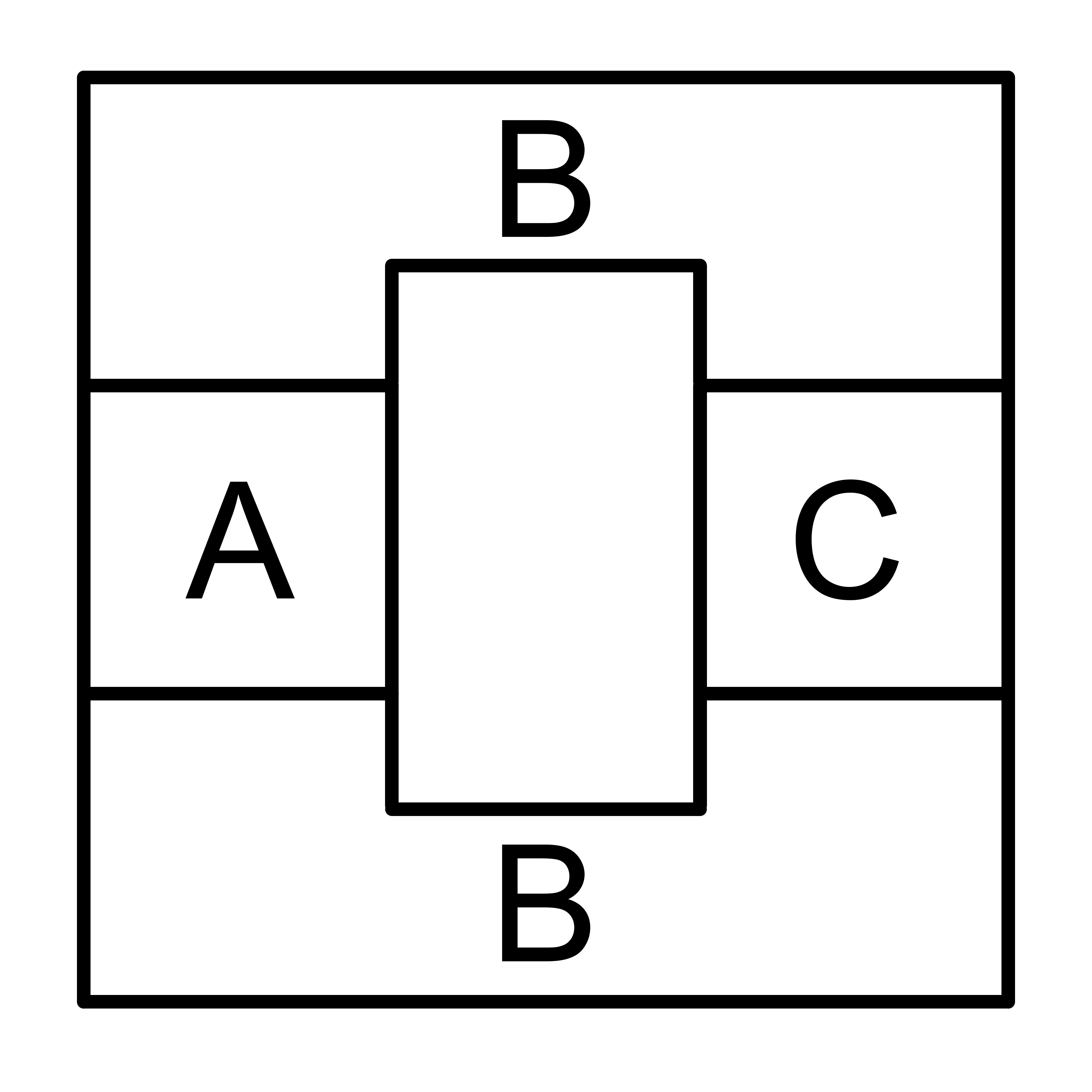}}}
\end{gathered}$
&
\text{\large $\gamma=1$}
&
\text{\large $\gamma=1$}
&
\text{\large $\gamma=1$}
&
\text{\large $\gamma=0$}
&
\text{\large $\gamma=8$}
\\
\hline
 $\begin{gathered}
\text{\large Secondary boundary}\\\text{\large gauge operators}\\\text{(\large upper half-plane)} 
\end{gathered}$& $\begin{gathered} \text{\large None} \end{gathered}$& $\begin{gathered}
    \vcenter{\hbox{\includegraphics[scale=.15]{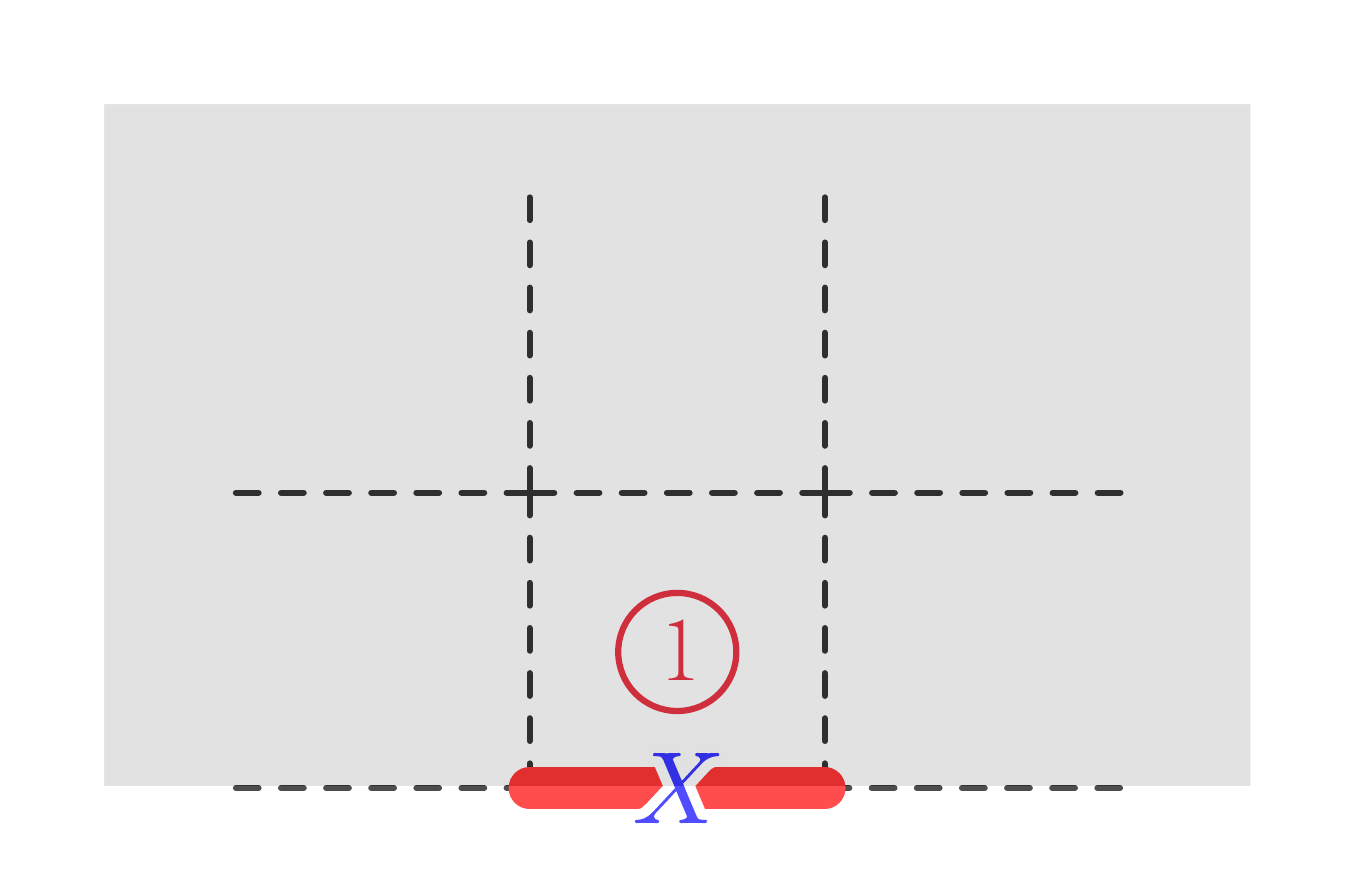}}}
\end{gathered}$& $\begin{gathered}
    \vcenter{\hbox{\includegraphics[scale=.15]{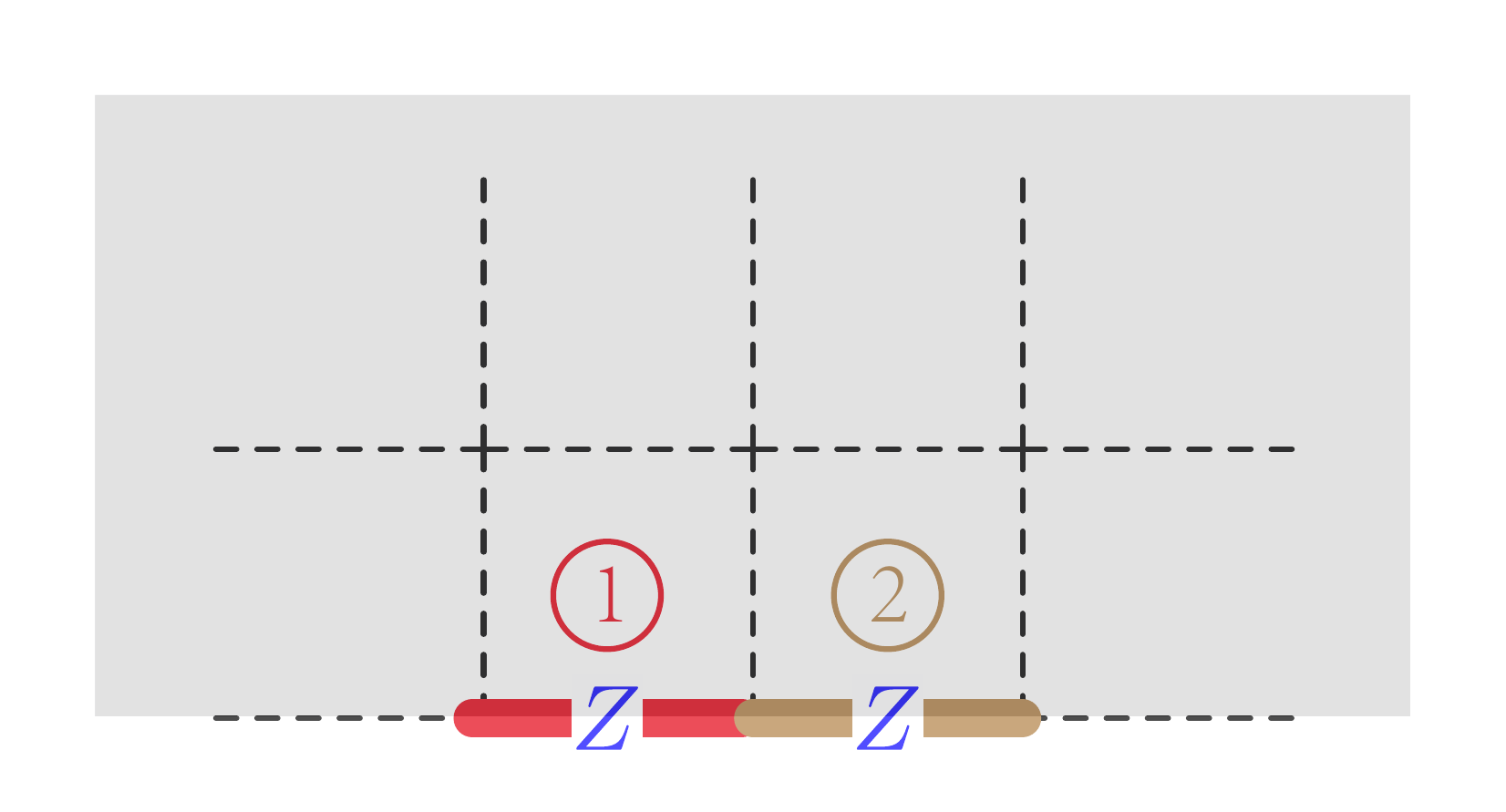}}}
\end{gathered}$& $\begin{gathered}
    \vcenter{\hbox{\includegraphics[scale=.15]{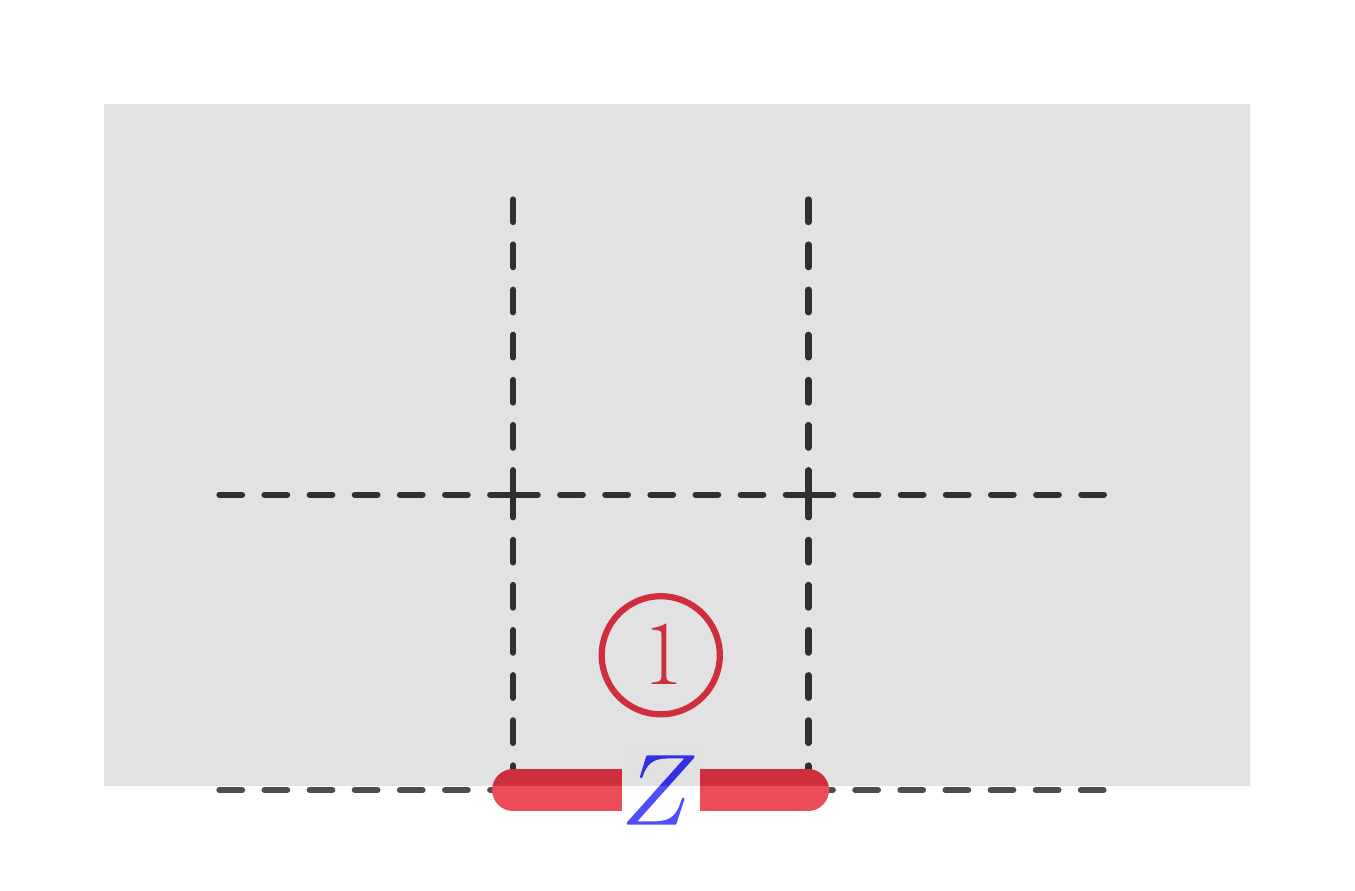}}}
\end{gathered}$&$\begin{gathered}
    \vcenter{\hbox{\includegraphics[scale=.15]{upper_1.pdf}}}
\end{gathered}$\\\hline
 $\begin{gathered}
\text{\large Secondary boundary}\\\text{\large gauge operators}\\\text{\large (lower half-plane)} 
\end{gathered}$& $\begin{gathered} \text{\large None} \end{gathered}$& $\begin{gathered}
    \vcenter{\hbox{\includegraphics[scale=.15]{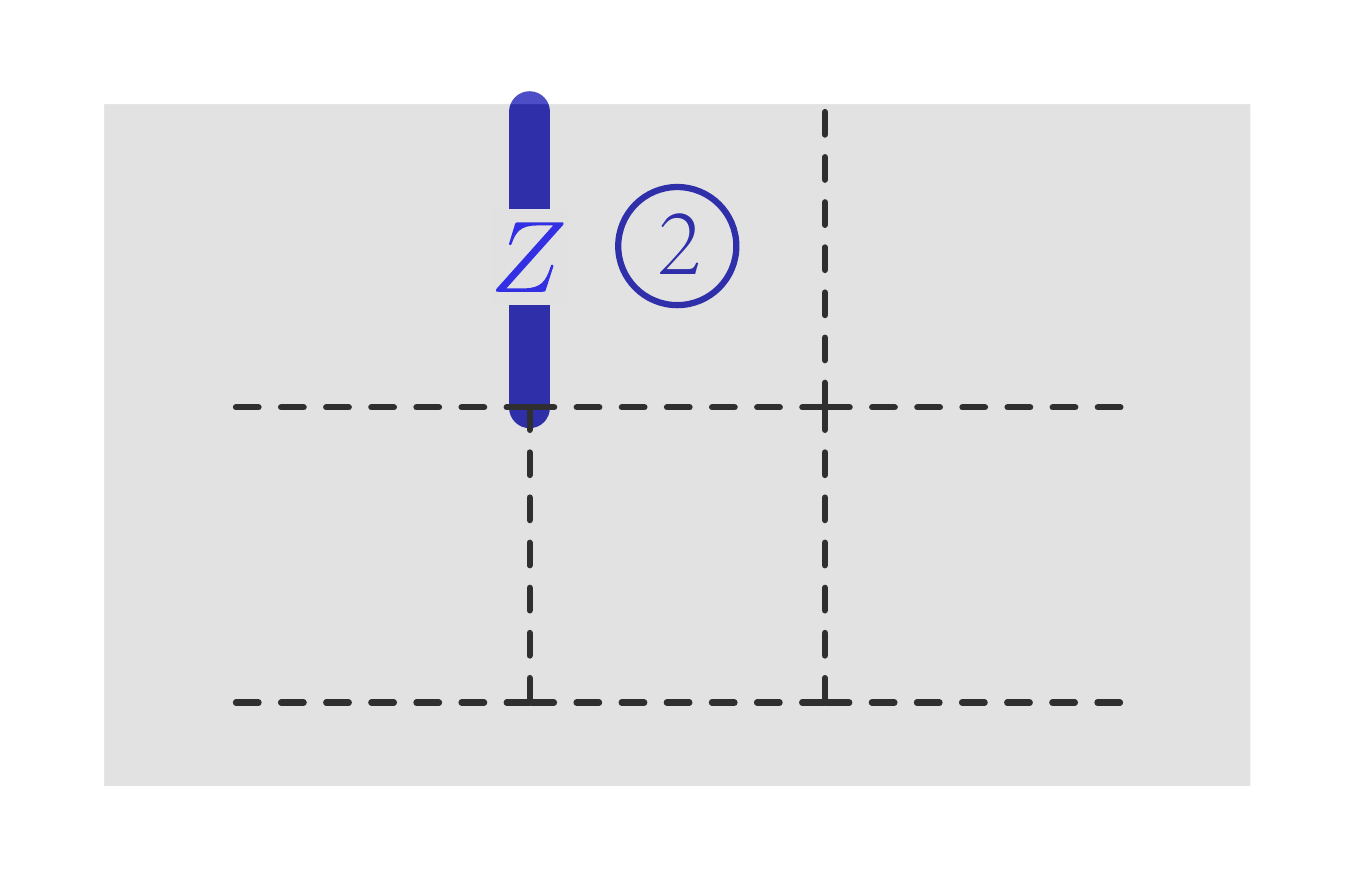}}}
\end{gathered}$& $\begin{gathered}
    \vcenter{\hbox{\includegraphics[scale=.15]{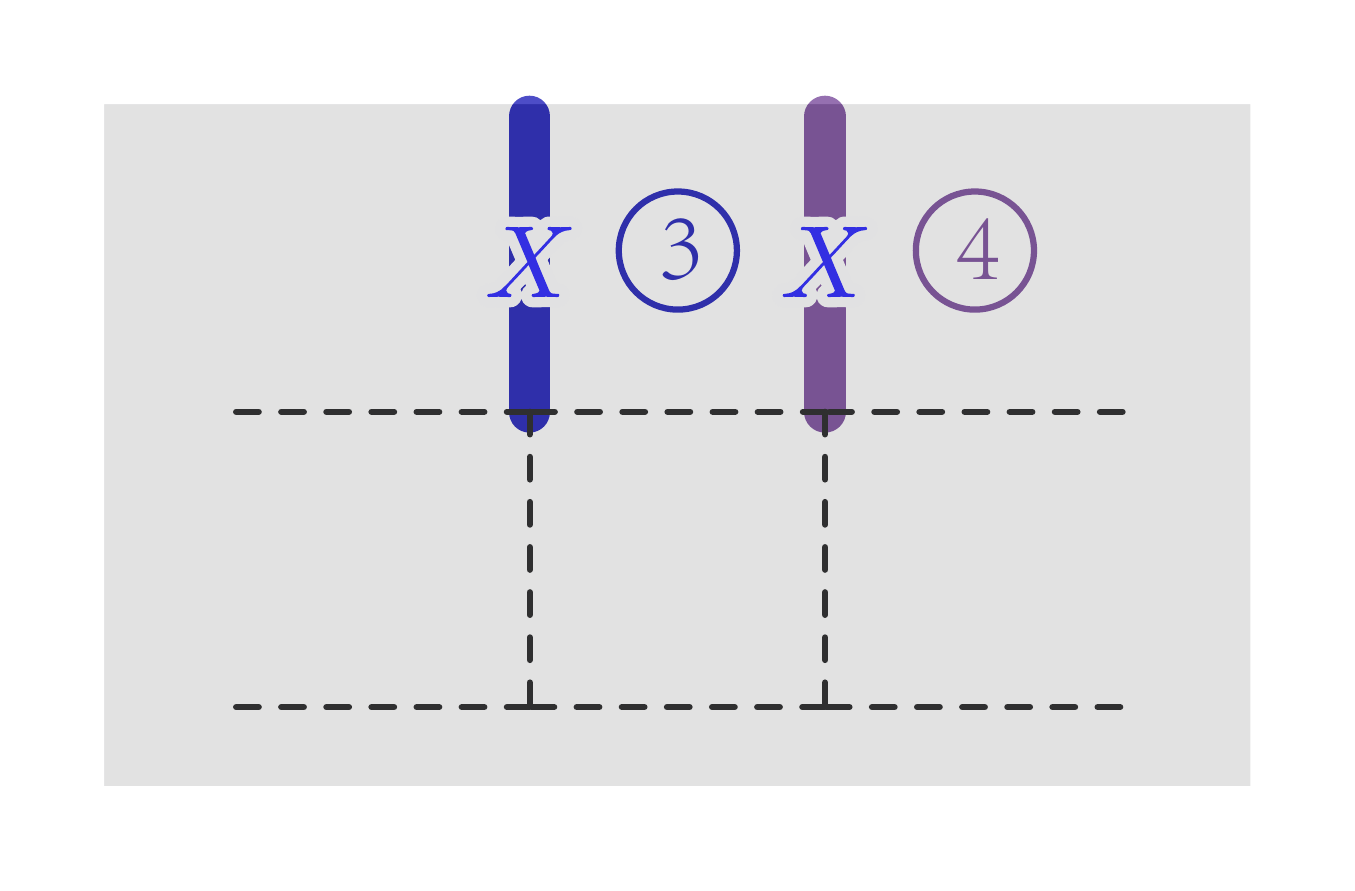}}}
\end{gathered}$& $\begin{gathered}
    \vcenter{\hbox{\includegraphics[scale=.15]{lower_1.pdf}}}
\end{gathered}$&$\begin{gathered}
    \vcenter{\hbox{\includegraphics[scale=.15]{lower_1.pdf}}}
\end{gathered}$\\\hline
\end{tabular}
}
\caption{The Levin-Wen TEE $\gamma$ for the rectangular and concave partitions and secondary boundary gauge operators. The rectangular partition produces spurious contributions in the last four examples. By contrast, the concave partition yields the correct TEE: the first three models are equivalent to the standard toric code up to qubit shifts or local unitary conjugation, whereas the fourth model is equivalent to the trivial Ising Hamiltonian $H=-\sum_e X_e$ conjugated by $CZ$ gates. The last column is the $(3,3)$-BB code, equivalent to $8$ copies of toric codes.
The total number of types of secondary boundary gauge operators on the upper and lower half plane is always equal to twice the spurious TEE, or equivalently, twice the difference between $\gamma$ in the rectangular partition and $\gamma$ in the concave partition. The boundary gauge operators (more explicitly, the generators of boundary gauge operators) in the third example are two different single $X$ or single $Z$ operators. Although they differ only by a translation, neither can be obtained from the other by multiplying by a finite product of primary boundary gauge operators and bulk operators. Therefore, they represent independent elements of the group of secondary boundary gauge operators. }
\label{tab: examples}
\end{table*}

To compute the TEE $\gamma$, we use the formula in Lemma~\ref{lemma: TEE contribution},
\begin{equation}
    \gamma = \frac{1}{2}\log\!\left|\frac{J_{ABC}}{J_{BC}\,J_{AB}}\right|~,
\end{equation}
where $J_R$ denotes the stabilizer group supported in region $R$. A derivation of this formula is given in Appendix~\ref{app: Computing entanglement entropy}. This formula applies to both qubit and qudit stabilizer codes. 

We now examine the shifted toric code shown in Fig.~\ref{fig: stabilizer_CNOT_TC}. For the rectangular partition in Fig.~\ref{fig: sTEE_shifted_TC_decomposed}, we find $\gamma = 2$.
The closed $e$- and $m$-string operators winding once around the region $D$ contribute to the genuine TEE and reflect the toric code topological order. In addition, there are extra contributions from the $X$-type and $Z$-type stabilizers obtained by taking products of the $\mathcal{A}_v$ operators and $\mathcal{B}_p$ operators marked by the red and blue dots in Fig.~\ref{fig: sTEE_shifted_TC_decomposed}. These stabilizers are supported on $A\cup B\cup C$, but neither can be factorized into stabilizers supported on $A\cup B$ and $B\cup C$. They therefore contribute to $\gamma$ as spurious topological entanglement entropy. Its restriction to the upper half-plane has no Pauli operators in the middle and produces only two secondary boundary gauge operators at the endpoints.

% We now present a detailed analysis of the sTEE for the shifted toric code shown in
% Fig.~\ref{fig: stabilizer_CNOT_TC}. Using the rectangular partition in
% Fig.~\ref{fig: sTEE_shifted_TC_decomposed}, we compute the topological entanglement
% entropy as
% \begin{equation}
%     \gamma = \frac{1}{2}\log\!\left|\frac{J_{ABC}}{J_{BC}\,J_{AB}}\right|
%     = 2,
% \end{equation}
% with the equation provided in Appendix~\ref{app: Computing entanglement entropy}. Here $J$ denotes the stabilizer group of a region. (In this paper, the base of $\log$ equals to the qudits' dimension $d$ if it is not specified.) The $e$- and $m$-anyon string operators winding once around the central region contribute
% to the TEE and reflect the topological order of toric code. However, the red $X$-type and blue $Z$-type stabilizers
% given by the products of the $\mathcal{A}_v$ operators (red dots) and the $\mathcal{B}_p$ operators (blue
% dots) contribute to the sTEE: they are supported on $A\cup B\cup C$, but neither can be factorized into stabilizers
% supported on $A\cup B$ and on $B\cup C$.

\begin{figure}[htb]
    \centering
    \hspace{3em}
    \subfigure[Stabilizer $\mathcal{A}_v$ and $\mathcal{B}_p$.]{%
    \raisebox{2cm}{\includegraphics[scale=0.06]{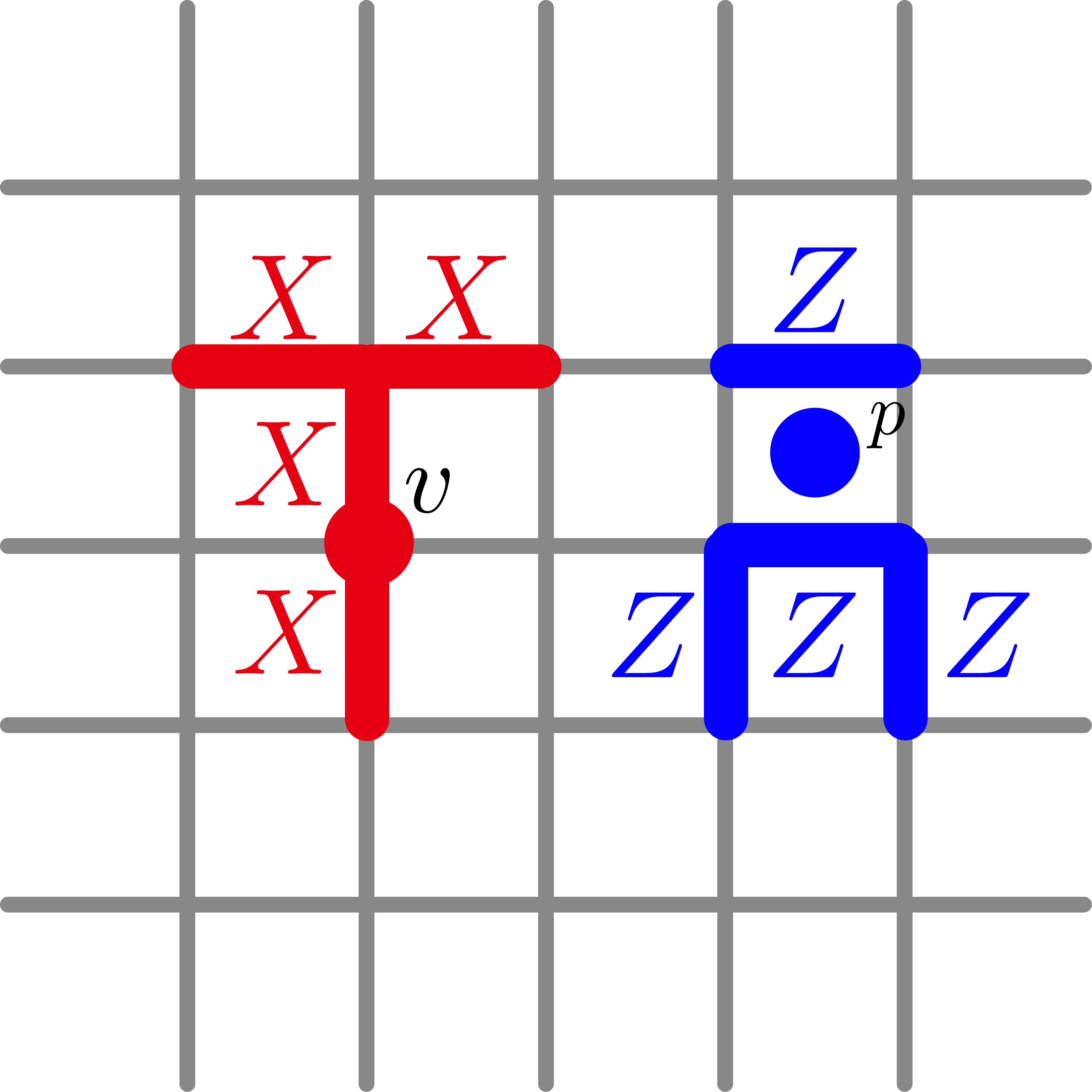}}%
    \label{fig: stabilizer_CNOT_TC}%
    }
    \hspace{6em}
    \subfigure[sTEE contribution.]{%
    \includegraphics[scale=0.06]{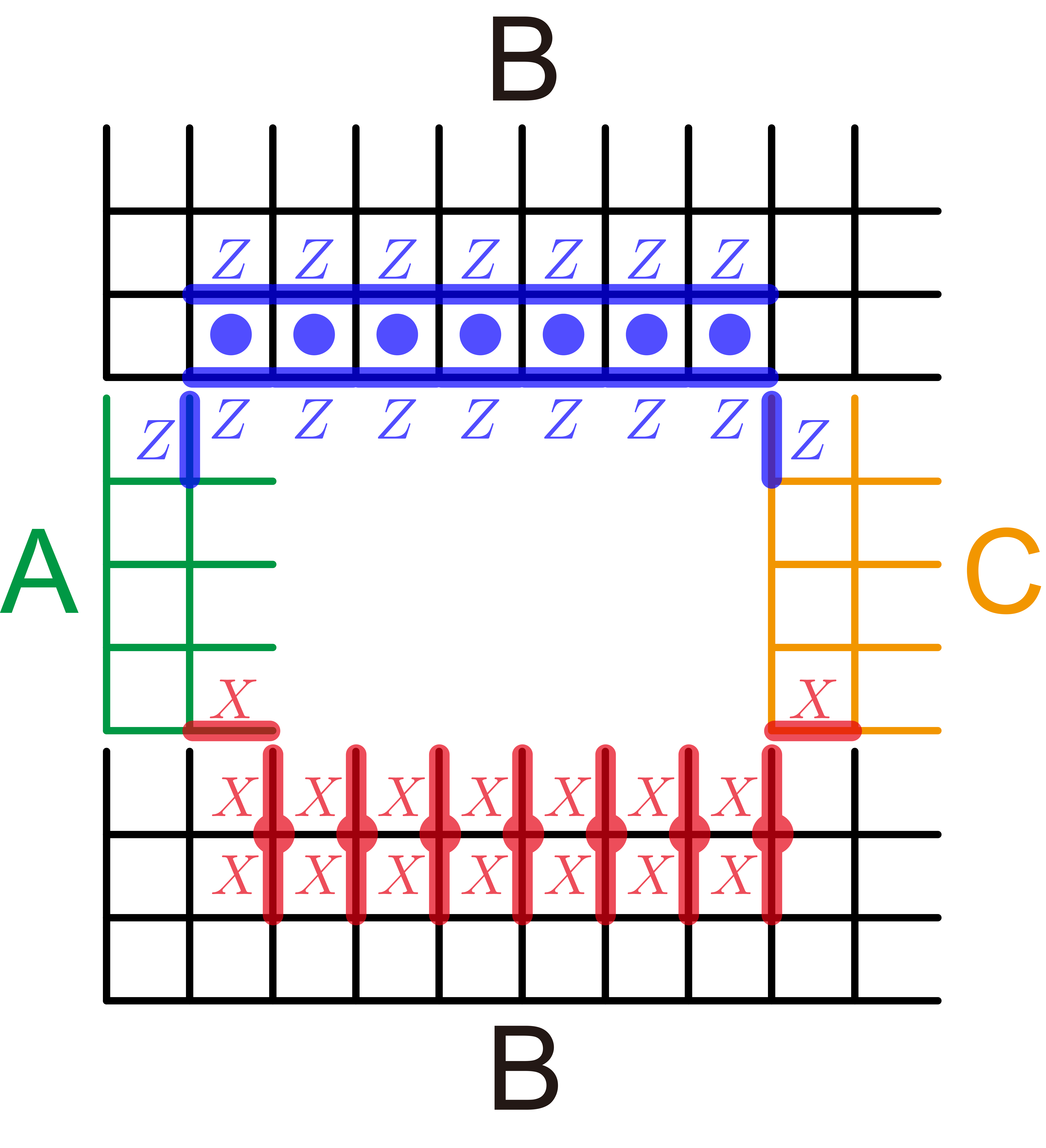}%
    \label{fig: sTEE_shifted_TC_decomposed}%
    }
    \caption{sTEE of the shifted toric code. (a) The stabilizers of the shifted toric code obtained from the standard toric code by a circuit of CNOT gates. The red and blue Pauli operators denote the $\mathcal{A}_v$ and $\mathcal{B}_p$ terms, respectively, and the corresponding dots mark their lattice positions. (b) In the rectangular partition, the green, orange, and black regions represent regions $A$, $C$, and $B$, respectively. The $X$-type and $Z$-type stabilizers, formed by taking products of the $\mathcal{A}_v$ operators and $\mathcal{B}_p$ operators marked by the corresponding dots, contribute to the spurious TEE.
    %\YG{[I adjust the size and spacing, check if you like]}
    }
    \label{fig: sTEE of toric code}
\end{figure}

More generally, consider the $h$-shifted toric code with stabilizers
\begin{eqs}
\begin{gathered}
    \mathcal{A}_v = \vcenter{\hbox{\includegraphics[scale=.175]{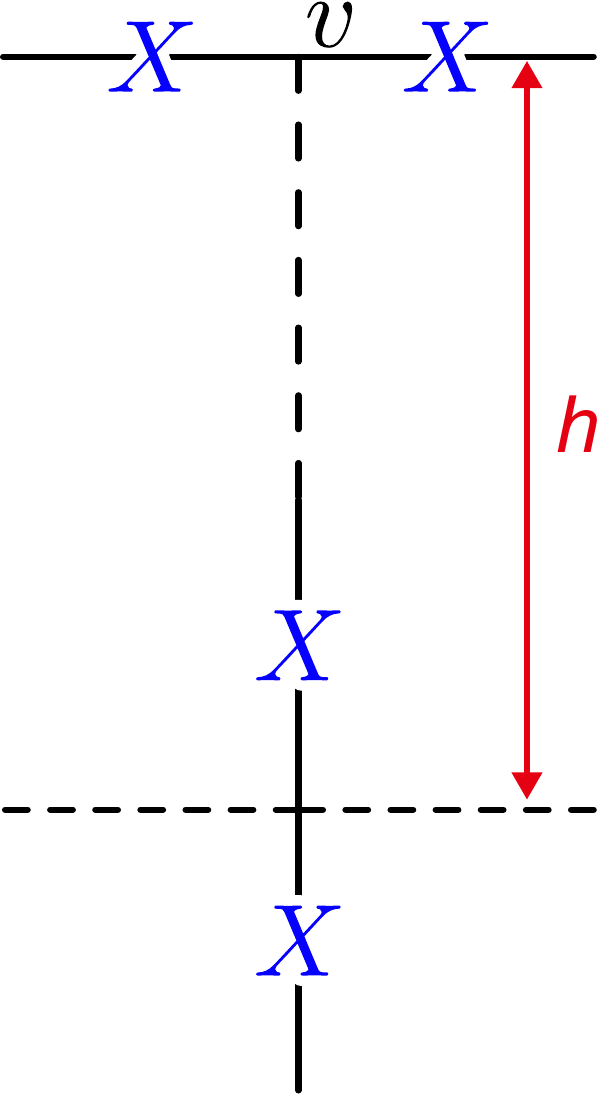}}}, \quad
    \mathcal{B}_p = \vcenter{\hbox{\includegraphics[scale=.175]{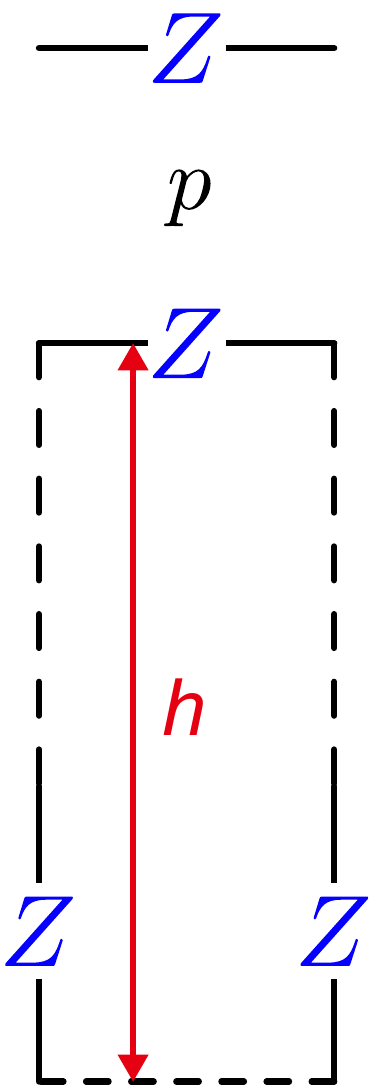}}}
\end{gathered},
\end{eqs}
where the vertical part is shifted by $h$ qubits relative to the standard toric code. For the rectangular partition, one finds $\gamma = h + 1$.
The mechanism is the same as in Fig.~\ref{fig: sTEE of toric code}: there are $h$ additional layers of $X$- and $Z$-type stabilizers $\mathcal{S}_{ABC}$ that cannot be factorized as $\mathcal{S}_{AB}\mathcal{S}_{BC}$. Thus, for the $h$-shifted toric code, the sTEE grows linearly with the stabilizer height, showing that the geometry of the stabilizers is a key ingredient in the failure of the rectangular partition. It is straightforward to verify that there are $h$ types of secondary boundary gauge operators on the upper and lower half-plane, respectively.

\section{Computing entanglement entropy of stabilizer states: Proof of Lemma~\ref{lemma: TEE contribution}}
\label{app: Computing entanglement entropy}

This subsection is dedicated to providing a self-contained derivation for the entanglement entropy of a subsystem stabilizer state,
leading to a proof of Lemma \ref{lemma: TEE contribution}.

For a general $\mathbb{Z}_p$-qudit system with $d\ge 2$, we denote the computational basis as $\{|j\rangle, j\in \mathbb{Z}_p \}$. The corresponding Pauli $X$ and $Z$ act on the basis as follows
\begin{equation}
X\ket{j} = \ket{j+1}
   ,\quad  Z \ket{j} = \omega_d^j \ket{j},\quad \text{with} \quad \omega_d = \exp(2\pi\mathrm{i}/d),
\end{equation}
Note $ZX = \omega_d XZ$ generalizes the standard anti-commutation relation for qubits ($d=2$). 
For the commutation relation among general $n$ qudit Pauli operators, it is convenient to introduce the symplectic formalism, where a general Pauli operator (up to phase) is represented by an array of $2n$ elements in $\ZZ_d$ 
\begin{equation}
\sigma(\alpha_1,\alpha_2,...,\alpha_n;\beta_1,\beta_2,...,\beta_n) := X^{\alpha_1} Z^{\beta_1} \otimes  X^{\alpha_2} Z^{\beta_2} \otimes \cdots \otimes X^{\alpha_n} Z^{\beta_n}. 
\label{eq: pauli to mod}
\end{equation}
In other words, $\sigma$ induces a map from the free $\ZZ_d$-module of rank $2n$ to the (reduced) Pauli group $\mathcal{P}_n=\tilde{\mathcal{P}}_n/\mathcal{K}$, which is the quotient group of the Heisenberg-Weyl Pauli group $\tilde{\mathcal{P}}_n$ that is closed under matrix multiplication, over phases $\mathcal{K}=\{\omega_d^j,j=0,1,...,d-1\}$. In this formalism, the commutation relation among Pauli operators 
can be expressed as
\begin{equation}
\sigma (f) \sigma(g) \sigma (f)^{-1} \sigma(g)^{-1} = \omega_d^{f^T \Lambda g} ,\quad \text{where} \quad  \Lambda = \begin{pmatrix}
        0 & I_n \\
        -I_n & 0
    \end{pmatrix}. 
\end{equation}

A stabilizer group $J$ is an abelian subgroup of the Pauli group $\mathcal{P}_n$ such that there exists a pure state $\ket{\psi}$ stabilized by it, i.e. $s|\psi\rangle=|\psi\rangle$ for all $s\in J$. 
If such a pure state is unique, it is defined as the pure stabilizer state associated with the stabilizer group $J$.
Otherwise, we define the mixed stabilizer state as the equally weighted mixture of all such states,
\begin{equation}
    \rho_{ J} = \frac{1}{d^n} \sum_{s\in  J} s = \frac{| J|}{d^n} P_{ J}, \label{eq: rhoS}
\end{equation}
where $P_{ J}= (1/| J|) \sum_{s\in  J} s$ is a projector of rank $r_{ J}= d^n/| J|$.
Indeed, the prefactor $p=| J|/d^n$ in Eq.~\eqref{eq: rhoS} represents the probability for each $|\psi\rangle$ stabilized by $ J$. 
Therefore, the von Neumann entropy of $\rho_{ J}$ has the following form
\begin{equation}
    S(\rho_{ J}) = - \Tr(\rho_{ J} \log \rho_{ J}) = n \log d - \log | J|.
\end{equation}
The problem of computing the entanglement entropy is thus equivalent to the problem of computing the size of the stabilizer group $ J$.

We are then prepared to prove Lemma \ref{lemma: TEE contribution}.

\begin{replemma}{lemma: TEE contribution}
For stabilizer codes, a nonzero Levin-Wen topological entanglement entropy arises from each independent stabilizer $\mathcal{S}$ supported in $A \cup B \cup C$ that cannot be factorized as
\begin{equation}
    \mathcal{S} = \mathcal{S}_{AB}\,\mathcal{S}_{BC},
\end{equation}
where $\mathcal{S}_{AB}$ and $\mathcal{S}_{BC}$ are stabilizers supported on $A \cup B$ and $B \cup C$, respectively. We refer to this failure of factorization as the \textbf{indivisibility condition}, and call such a stabilizer $\mathcal{S}$ \textbf{indivisible}.
\end{replemma}

\proof{
Note that in the expression of TEE, contributions from system sizes perfectly cancel each other. We have 
\begin{equation}
    \gamma = \frac{1}{2}\left(S_{AB} + S_{BC} - S_{ABC} - S_B\right) = \frac{1}{2}\log \frac{| J_{ABC}| | J_B|}{| J_{AB}| | J_{BC}|}.
\end{equation}
From isomorphic theorems of abelian groups, we have
\begin{equation}
    \frac{| J_{ABC}| | J_B|}{| J_{AB}| | J_{BC}|} = \frac{| J_{ABC}/ J_{AB}|}{| J_{BC}/ J_{B}|} = \frac{| J_{ABC}/ J_{AB}|}{| J_{BC}/( J_{BC} \cap  J_{AB})|} = \frac{| J_{ABC}/ J_{AB}|}{|( J_{BC}  J_{AB})/ J_{AB}|} = \left|\frac{ J_{ABC}/ J_{AB}}{( J_{BC}  J_{AB})/ J_{AB}}\right| = \left|\frac{ J_{ABC}}{ J_{BC}  J_{AB}}\right|.
\end{equation}
This proves the lemma.
}

Our next goal is to derive an algorithm for computing the topological entanglement entropy. The method is inspired by the stabilizer-formalism approach to entanglement entropy in Ref.~\cite{fattalEntanglementStabilizerFormalism2004}.
This reduces to the problem of computing the entanglement entropy of a bipartite pure stabilizer state.
Let us consider a stabilizer group $ J$ with $| J|=d^n$ elements, namely it has $n$ independent commuting generators and the corresponding state $\rho_{ J}$ is pure, denoted as $\rho=|\psi\rangle\langle\psi|$. 
For a bipartition $A,B$ of the full system with $n=n_A+n_B$, we obtain two reduced density matrices $\rho_A= \Tr_B(\rho)$ and $\rho_B= \Tr_A(\rho)$. Due to the tracelessness of the Pauli operators (except the identity operator), the reduced density matrices have the following explicit forms
\begin{equation}
    \rho_A = \frac{1}{d^{n_A}} \sum_{s\in  J_A} s, \quad \rho_B = \frac{1}{d^{n_B}} \sum_{s\in  J_B} s
\end{equation}
where $ J_A$ and $ J_B$ are subgroups of the full stabilizer group that only have support on $A$ and $B$ respectively. 

Now, to characterize subgroups $ J_{A}, J_B$ and compute the corresponding entanglement entropy $S_A:=S(\rho_A)=S(\rho_B)=S_B$, we 
employ the decomposition of the symplectic space $\ZZ_d^{2n} = \ZZ_d^{2n_A}\oplus \ZZ_d^{2n_B}$ as a $\ZZ_d$-module, aligned with the aforementioned bipartition of full system. For a complete set of stabilizer generators $f^{j=1,...,n} \in \ZZ_d^{2n}$, it is informative to construct the truncated symplectic inner products 
\begin{equation}
   M^A_{jk}:= (f^A_j)^{T} \Lambda f^A_k = - (f^B_j)^{T} \Lambda f^B_k:=-M^B_{jk}
\end{equation}
where $f^A\in \ZZ_d^{2n_A}$ and $f^B\in \ZZ_d^{2n_B}$ are projections/truncations of $f$ onto $\ZZ_d^{2n_A}$ and $\ZZ_d^{2n_B}$ respectively. 

Since $M^A=-M^B$, we can use the same set of row and column operators to reduce them to Smith normal form (i.e., regard them as matrices over $\ZZ$), 

\begin{equation}
   L M^A R = D^A ={\text{diag}} (\lambda_1,\lambda_2,\ldots,\lambda_r,0,...,0)=-D^B=-LM^BR
   \label{eq: snf}
\end{equation}
$L$ and $R$ are invertible integral matrices representing the row and column operations, respectively.  
In other words, $L$ and $R$ are the redefinitions of the stabilizer generators with simple commutation relations encoded in the matrix $D^{A/B}$. 
More explicitly, let us denote $g^k:=\sum_j f^j  R_{jk}$, Eq.~\eqref{eq: snf} implies that $g^1,g^2,\ldots,g^r$ have phase factor $\omega_d^{\lambda_1},\omega_d^{\lambda_2},\ldots,\omega_d^{\lambda_r}$ when truncating to subsystem $A$ (opposite phase for $B$) and commuting with other generators.\footnote{The phase factors $\omega_d^{\lambda_j}$ could be trivial, namely $\omega_d^{\lambda_j}=1$ since the Smith normal form is taken in $\ZZ$ and therefore the elementary divisors $\lambda_j$ could be multiples of $d$.}
The remaining generators, $g^{r+1}, g^{r+2}, \ldots,g^{n}$, commute with the whole stabilizer group even when restricting to the subsystems $A$ or $B$, which implies that they can be recombined in a way that they are only supported on $A$ or $B$ without mixture.  

The factorization does not hold for the first $r$ generators since they have nontrivial commutation relations when restricting to a subsystem. We can trivialize the phase factors by multiplying $d/{\rm gcd}(d,\lambda_j)$ and construct a new set of generators that commutes with the full stabilizer group $J$. 

\begin{equation}
    \frac{d}{{\rm gcd}(d,\lambda_1)} g^1, ~
    \frac{d}{{\rm gcd}(d,\lambda_2)} g^2, ~
    \ldots, ~\frac{d}{{\rm gcd}(d,\lambda_r)} g^r,~g^{r+1},~g^{r+2},\ldots,g^n
\end{equation}
Note that this set of generators factorizes and generates a product of stabilizer groups $ J_A\times S_B$.
% with total count 
% $|\mathcal{S}_A\times S_B|= d^{n-r} \prod_{j=1}^r {\rm gcd}(d,\lambda_r)$. 
Therefore, the von Neumann entropy of $\rho_A\otimes \rho_B$ is given as follows 
\begin{equation}
    S(\rho_A\otimes \rho_B) =  \sum_{j=1}^r \log \frac{d}{{\rm gcd}(d,\lambda_j)}. 
\end{equation}
An immediate corollary states that 
\begin{equation}
    S_{A}=S_B=\frac{1}{2}  S(\rho_A\otimes \rho_B) = \frac{1}{2}  \sum_{j=1}^r \log \frac{d}{{\rm gcd}(d,\lambda_j)}. 
    \label{eq: SA}
\end{equation}

A few comments on the above procedure
\begin{enumerate}
    \item In Eq.~\eqref{eq: snf}, we have regarded the $\ZZ_d$ matrix elements as an integer in $\ZZ$, which may bring ambiguity in the choice of the representative. However, the physical observables, such as the entropy formula \eqref{eq: SA}, only depend on $\lambda_j$ modulo $d$, which guarantees the effectiveness of the above procedure/algorithm. 
    \item The formula \eqref{eq: SA} agrees with the physical intuition of the subsystem entanglement entropy: the uncertainty originates from the non-commuting truncated stabilizers. For qudits with prime dimensions, each non-commuting generator generates entropy $\log d$, 
    and will be discarded when constructing the remaining stabilizer group $ J_A$, so Eq.~\eqref{eq: SA} can be simplified to:
    \begin{equation}
        S_{A}=S_B=\frac{1}{2}  S(\rho_A\otimes \rho_B) = \frac{1}{2}  r\log d. 
        \label{eq: SA2}
    \end{equation}
    While for composite $d$, there is a finer structure that depends on the elementary divisors $\lambda_{j=1,2,\ldots,r}$, e.g., multiples of non-commuting generators may survive in the remaining stabilizer group.
    \item A diagonal $D^{A/B}$ is sufficient for our purpose. To be concrete, one may use the Smith normal form where $\lambda_j|\lambda_{j+1}$ in $D^{A/B}$. The final result remains the same. 
\end{enumerate}

%%%%%%%%%%%%%%%%%%%%%%%%%%%%%%%%%%%%%%%%%%%%%%%%%%%%%%%%%%%%%%%%%%%%%%%%%%%%%%%%
\section{Characterization of spurious contributions: Proof of Lemma~\ref{lemma: sTEE contribution}}\label{proof: lemma 3}

In this section, we prove Lemma \ref{lemma: sTEE contribution} in the main text that identifies the spurious contributions to topological entanglement entropy in stabilizer states.

\begin{replemma}{lemma: sTEE contribution}
Let \(D_{\mathrm{in}}\subset D\) be an inner region large enough to support representatives of all anyon types. Assume that the geometry satisfies the following two conditions:
\begin{enumerate}
    \item every trivial anyon supported in \(D_{\mathrm{in}}\) can be created by a local Pauli operator supported in \(D\);
    \item every nontrivial anyon supported in \(D_{\mathrm{in}}\) can be created by a semi-infinite string operator supported in \(A\cup D\cup E\), as shown in Fig.~\ref{fig:ABCD_partition}.
\end{enumerate}
Then, the stabilizers $\mathcal{S}$ supported in $A \cup B \cup C$ that contribute to the spurious Levin-Wen topological entanglement entropy are precisely those that satisfy the indivisibility condition and commute with every such string operator.
Equivalently, such a stabilizer $\mathcal{S}$ does not satisfy Eq.~\eqref{eq: indivisibility condition} and can be written as
\begin{equation}
    \mathcal{S} = \prod_{j \in \mathcal{J}} \mathcal{S}_j,
\end{equation}
where $\mathcal{J}$ is a finite set of stabilizer generators such that no generator $\mathcal{S}_j$ is fully supported in $D_{\mathrm{in}}$.
\end{replemma}

We begin by recalling the physical meaning of this statement. Topological entanglement entropy~\cite{Kitaev2006Topological,Levin2006Detecting} is the subleading constant term in the entanglement area law of a two-dimensional gapped ground state, and it captures the total quantum dimension
\begin{equation}
\mathcal{D}=\sqrt{\sum_{\varsigma\in\text{anyons}} d_{\varsigma}^2},
\end{equation}
where the sum runs over all superselection sectors and $d_{\varsigma}$ is the quantum dimension of the anyon type $\varsigma$. For Abelian anyons, $d_{\varsigma}=1$, so $\mathcal{D}=\sqrt{|G_\mathcal{A}|}$, where $G_\mathcal{A}$ is the finite Abelian group of anyon types. In the stabilizer setting,
\begin{equation}
G_\mathcal{A}\cong S/\varepsilon(P).
\end{equation}

To extract $\gamma$, we use the Levin-Wen prescription~\cite{Levin2006Detecting} adopted in the main text and throughout this Supplemental Material. Namely, we consider the conditional mutual information associated with the annular geometry shown in Fig.~\ref{fig: annulus},
\begin{equation}
    \gamma = \frac{1}{2} I(A:C|B).
\end{equation}

\begin{figure}
\centering
\begin{tikzpicture}[scale=1,baseline={([yshift=-2pt]current bounding box.center)}]
\def\R{1.5}  % outer radius
\def\r{0.8}  % inner radius
  % Draw outlines
  \draw[line width=0.9pt] (0,0) circle (\R);
  \draw[line width=0.9pt] (0,0) circle (\r);
  \foreach \ang in {45,135,225,315} {
    \draw[line width=0.9pt]
      ({\r*cos(\ang)},{\r*sin(\ang)}) -- ({\R*cos(\ang)},{\R*sin(\ang)});
  };
\node at ({(\R+\r)*cos(180)/2},{(\R+\r)*sin(180)/2}) {$A$};
\node at ({(\R+\r)*cos(90)/2},{(\R+\r)*sin(90)/2}) {$B$};
\node at ({(\R+\r)*cos(270)/2},{(\R+\r)*sin(270)/2}) {$B$};
\node at ({(\R+\r)*cos(0)/2},{(\R+\r)*sin(0)/2}) {$C$};
\end{tikzpicture}
\hspace{40pt}
\begin{tikzpicture}[scale=1,baseline={([yshift=-2pt]current bounding box.center)}]
\def\R{1.5}  % outer radius
\def\r{0.8}  % inner radius
  % Draw outlines
  \draw[line width=0.9pt] (0,0) circle (\R);
  \draw[line width=0.9pt] (0,0) circle (\r);
  \draw[red, line width=1.0pt] (0,0) circle ({0.5*(\r+\R)});
  \node at ({0.5*(\r+\R)+0.15},0) {$b$};
  \draw[line width=0.9pt, decorate, decoration={zigzag, segment length=6pt, amplitude=1pt}]
    (-2,0) -- (0,0);
    \filldraw[blue, line width=1.1pt] (0,0) circle (2pt);
    \node at (7pt,0pt) {$a$};
\end{tikzpicture}
  \caption{Left: The geometry of the Levin-Wen prescription. Right: An anyon $a$ is imported inside the annulus by a Pauli string, flipping the expectation value of the Wilson loop $b$ that detects this anyon. See Lemma \ref{lemma: orthogonal anyon sector} for more details.}
    \label{fig: annulus}
\end{figure}
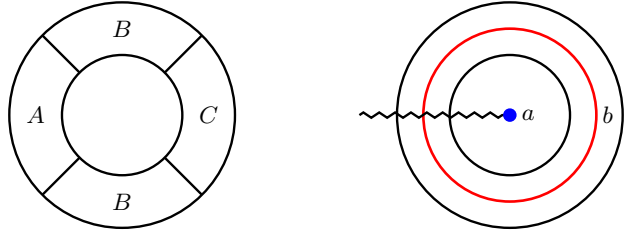

From the quantum information perspective, the conditional mutual information serves as an obstruction of the Petz recovery~\cite{Petz1986,Petz1988} $
    \mathcal{R}_{B\rightarrow BC}(\rho_{AB})  = \rho_{ABC}
$. In general, a non-zero $\gamma$ means that the recovery of the state $\rho_{ABC}$ demands certain global knowledge that might be absent in the subsystems, e.g. in $\rho_{AB}$ or $\rho_{BC}$. Indeed, the information of the superselection sector/anyon type is not accessible from subsystems, and therefore contributes to the $\gamma$. However, as first pointed out by Bravyi, there could be contributions from special short-range entangled states that are not related to topological orders therein, referred as the spurious contribution to the topological entanglement entropy, or ``the spurious topological entanglement entropy''
\begin{equation}
    \delta = \gamma-\log \mathcal{D}
\end{equation}
Recently, Refs.~\cite{Kim2023UniversalLowerBoundTEE, Levin2024TEE} show that such spurious contributions are always non-negative, i.e. $\gamma\ge \log \mathcal{D}$ (or $\delta\ge 0$). 

For abelian anyons supported on stabilizer states, the positivity of $\delta$ can be shown by constructing the following channel. Let $\rho_{ABC}$ be the reduced density matrix of the ground state on qudits supported in $A\cup B\cup C$. Then the channel can be defined as 
\begin{equation}
    \mathcal{C}(\rho_{ABC}) := \sum_{\varsigma\in G_\mathcal{A}} p_{\varsigma} \mathcal{V}_{\varsigma} \rho_{ABC} \mathcal{V}_{\varsigma}^\dagger,
\end{equation}
Here, we denote the string operator that pulls an anyon $\varsigma$ from outside to the center of the annulus along $A$ by $\mathcal{W}_{\varsigma}$ and then $\mathcal{V}_{\varsigma}=\mathcal{W}_{\varsigma}|_A$ is the truncation on qudits in $A$, which is also the truncation on qudits in $A\cup B\cup C$ since it is required that $\mathcal{W}_\varsigma$ does not overlap $B$ or $C$, which can be achieved because the width of $A$ is larger than the width of an anyon string of each type. In general, it is not unitary since $\mathcal{W}_{\varsigma}$ may consist of a finite-depth circuit, and a truncation will break the unitarity, which needs further treatment, such as slightly shrinking the region in question as shown in \cite{Kim2023UniversalLowerBoundTEE, Levin2024TEE}. In our case with stabilizer states, the string operator $\mathcal{W}_{\varsigma}$ is a Pauli string and therefore $\mathcal{V}_{\varsigma}$ is always a unitary. 

The probability is given by\begin{equation}p_{\varsigma}=1/|G_\mathcal{A}|=1/\mathcal{D}^2\end{equation} in our case.
In other words, $\mathcal{C}(\rho_{ABC})$ is an equal-weight probabilistic superposition of superselection/anyon sectors resulting in a maximal mixture of the anyon charges. 
The conditional mutual information of $\mathcal{C}(\rho_{ABC})$ is closely related to that of $\rho$, as the channel mixes distinct sectors without changing the local data.

To compute the conditional mutual information of $\mathcal{C}(\rho_{ABC})$, we first establish the following lemma.
\begin{lemma}\label{lemma: orthogonal anyon sector}
    Consider a Levin-Wen geometry over a two-dimensional topological stabilizer code.
    Define $\rho_{ABC}^{(\varsigma)}$ as the reduced density matrix on region $A\cup B\cup C$ with an anyon with the type $\varsigma \in G_\mathcal{A}$ inside $D_{in}$ and its anti-particle inside $E$.
    If there is no stabilizer generator that is supported by a region with overlap with both $D$ and $E$, that is, the width of the region $ABC$ is larger than the interaction range, then $\operatorname{tr}\left(\rho^{(\varsigma)}_{ABC}\rho^{(\beta)}_{ABC}\right) = 0$ if $\varsigma\neq \beta$.
    Additionally, define $\rho_{AB}^{(\varsigma)}$ and $\rho_{BC}^{(\varsigma)}$ as the reduced density matrix of $\rho_{ABC}^{(\varsigma)}$ on region $A\cup B$ and $B\cup C$ respectively. Then $\rho^{(\varsigma)}_{AB}=\rho^{(\beta)}_{AB}$, $\rho^{(\varsigma)}_{BC}=\rho^{(\beta)}_{BC}$ for all $\varsigma,\beta\in G_{\mathcal{A}}$.
\end{lemma}

\proof{ First, we will show that $\rho^{\varsigma}_{ABC}$ is well-defined because a quantum state on an annulus region is determined by the anyon sector in the region encircled by it.

Consider a state with an anyon $\alpha\in S$ inside $D_{in}$ and its anti-particle inside $E$, and we denote the corresponding density matrix by $\rho_{\alpha}$. For two anyon $\alpha_1,\alpha_2$, if they are the same type of anyon, then $\alpha_1-\alpha_2$ is a trivial anyon. The condition in Lemma~\ref{lemma: sTEE contribution} guarantees that there exists a Pauli operator $\mathcal{T}$ fully supported in $D$ creating such an anyon, so
\begin{equation}
    \rho_{\alpha_1}=\mathcal{T}\rho_{\alpha_2}\mathcal{T}^\dagger.
\end{equation}
Then, we reduce the two sides of the equation above on qudits fully supported in $A\cup B\cup C$,
 \begin{equation}
     \rho_{\alpha_1}|_{A\cup B\cup C}=\mathcal{T}|_{A\cup B\cup C}\rho_{\alpha_2}|_{A\cup B\cup C}\mathcal{T}^\dagger|_{A\cup B\cup C}
 \end{equation}

Since $\mathcal{T}$ is fully supported in $D$, its truncation is simply identity, so
 \begin{equation}
     \rho_{\alpha_1}|_{A\cup B\cup C}=\rho_{\alpha_2}|_{A\cup B\cup C},
 \end{equation}
 which means the reduced density matrix of $\rho_{\alpha}$ only depends on the anyon type, so $\rho_{ABC}^{(\varsigma)}$ is well-defined.

Denote an overall pure excited state with anyons $(\varsigma,\overline{\varsigma})$ in $D_{in}$ and $E$ as $\ket{\Phi^{(\varsigma)}}$. Then $\ket{\Phi^{(\varsigma)}}$ is a purification of $\rho^{(\varsigma)}_{ABC}$. Any other purification of $\rho^{(\varsigma)}_{ABC}$ can be obtained from $\ket{\Phi^{(\varsigma)}}$ by acting a unitary exclusively in ${D_{in}}$ and $E$.

All Pauli strings of the form $U_{D_{in}}\otimes U_E$ form a basis of the algebra of such unitaries, where $U_{D_{in}}, U_E$ is supported exclusively in region ${D_{in}}, E$, respectively. However, any such unitary cannot change the anyon sector inside ${D_{in}}$. (To change the sector, the inner unitary in ${D_{in}}$ must tunnel an anyon to the boundary, and the outer unitary in $E$ should ``pick up'' this boundary excitation and move it outward. But $B$ is thicker than the interaction range, so any unitary in $D_{in}$ cannot remove an excitation at the boundary of ${D_{in}}$.)
Therefore, any purification of $\rho^{(\varsigma)}_{ABC}$ is orthogonal to any purification of $\rho^{(\beta)}_{ABC}$,{
leading to fidelity $F(\rho^{(\varsigma)}_{ABC},\rho^{(\beta)}_{ABC}) = 0$,
and thus 
\begin{equation}\label{eq: orthogonality of different rho}
    \operatorname{tr}\rho^{(\varsigma)}_{ABC}\rho^{(\beta)}_{ABC} = 0.
\end{equation}

Additionally, for the Levin-Wen geometry used, it is clear that if $A\cup D\cup E$ supports a semi-infinite string operator of all types, then $C\cup D\cup E$ does too. For an anyon of type $\varsigma$, we have
\begin{equation}\label{eq: expression of rho^(varsigma)_ABC}
    \rho^{(\varsigma)}_{ABC}=\mathcal{V}_{\varsigma} \rho_{ABC} \mathcal{V}_{\varsigma}^\dagger=\mathcal{V}'_{\varsigma} \rho_{ABC} \mathcal{V}_{\varsigma}'^\dagger
\end{equation}
Here $\mathcal{V}_{\varsigma}$ is the truncation on $A\cup B\cup C$ of the string operator fully supported in $A\cup D\cup E$ that pulls an anyon $\varsigma$ from outside to the center of the annulus, and $\mathcal{V}'_{\varsigma}$ is the truncation on $A\cup B\cup C$ of the string operator fully supported in $C\cup D\cup E$ that pulls the same anyon from outside to the center of the annulus. 
By construction, we have
\begin{equation}
\begin{aligned}
    &\mathcal{V}_{\varsigma}|_{B\cup C}=\mathcal{V}_{\varsigma}|_{A\cup D\cup E}|_{B\cup C}=1,\\
    &\mathcal{V}'_{\varsigma}|_{A\cup B}=\mathcal{V}'_{\varsigma}|_{C\cup D\cup E}|_{A\cup B}=1.
\end{aligned}
\end{equation}
Hence, the reduction of Eq.~\eqref{eq: expression of rho^(varsigma)_ABC} is given by
\begin{equation}
    \begin{aligned}
        &\rho^{(\varsigma)}_{AB}=\mathcal{V}'_{\varsigma}|_{A\cup B} \rho_{AB} \mathcal{V}_{\varsigma}'|_{A\cup B}^\dagger=\rho_{AB},\\
        &\rho^{(\varsigma)}_{BC}=\mathcal{V}'_{\varsigma}|_{B\cup C} \rho_{BC} \mathcal{V}_{\varsigma}'|_{B\cup C}^\dagger=\rho_{BC}.
    \end{aligned}
\end{equation}
Therefore, the reduced density matrix $\rho^{(\varsigma)}_{AB}$ or $\rho^{(\varsigma)}_{BC}$ cannot be distinguished for different $\varsigma$. This, combining with Eq.~\eqref{eq: orthogonality of different rho}, completes the proof.
\qed
}}

Now let us consider the value of $I_{\mathcal{C}(\rho_{ABC})}(A:C\mid B)$. Write
\begin{equation}
    \mathcal{C}(\rho_{ABC})=\sum_{\varsigma\in G_\mathcal{A}} p_\varsigma \,\rho^{(\varsigma)}_{ABC},
    \qquad
    \rho^{(\varsigma)}_{ABC}:=\mathcal{V}_\varsigma \rho_{ABC} \mathcal{V}_\varsigma^\dagger.
\end{equation}
By Lemma~\ref{lemma: orthogonal anyon sector}, the states $\rho^{(\varsigma)}_{ABC}$ for different $\varsigma$ are mutually orthogonal.
Therefore, using the entropy formula for a mixture of orthogonal states,
\begin{equation}
    S\!\left(\mathcal{C}(\rho_{ABC})\right)
    =
    \sum_{\varsigma\in G_\mathcal{A}} p_\varsigma\, S\!\left(\rho^{(\varsigma)}_{ABC}\right)
    + H(\{p_\varsigma\}),
\end{equation}
where \(H(\{p_\varsigma\})=-\sum_{\varsigma} p_\varsigma \log p_\varsigma\) is the Shannon entropy of the probability distribution \(\{p_\varsigma\}\).

Moreover, Lemma~\ref{lemma: orthogonal anyon sector} implies that the reduced density matrices on \(A\cup B\) and \(B\cup C\) are independent of the anyon sector:
\begin{equation}
    \rho^{(\varsigma)}_{AB}=\rho^{(\beta)}_{AB}=:\rho_{AB},
    \qquad
    \rho^{(\varsigma)}_{BC}=\rho^{(\beta)}_{BC}=:\rho_{BC},
    \qquad
    \forall\,\varsigma,\beta\in G_\mathcal{A}.
\end{equation}
Taking a further partial trace, the same holds for the reduced density matrix on \(B\), so \(\rho^{(\varsigma)}_{B}=:\rho_{B}\) for all \(\varsigma\). Thus the reduced states on \(A\cup B\), \(B\cup C\), and \(B\) are all independent of the anyon sector. Hence,
\begin{equation}
\begin{aligned}
    \mathcal{C}(\rho_{AB})&=\sum_{\varsigma\in G_\mathcal{A}} p_\varsigma\, \rho^{(\varsigma)}_{AB}=\rho_{AB},\\
    \mathcal{C}(\rho_{BC})&=\sum_{\varsigma\in G_\mathcal{A}} p_\varsigma\, \rho^{(\varsigma)}_{BC}=\rho_{BC},\\
    \mathcal{C}(\rho_{B})&=\sum_{\varsigma\in G_\mathcal{A}} p_\varsigma\, \rho^{(\varsigma)}_{B}=\rho_{B}.
\end{aligned}
\end{equation}
Consequently,
\begin{equation}
\begin{aligned}
    S\!\left(\mathcal{C}(\rho_{AB})\right)
    &=
    \sum_{\varsigma\in G_\mathcal{A}} p_\varsigma\, S\!\left(\rho^{(\varsigma)}_{AB}\right),\\
    S\!\left(\mathcal{C}(\rho_{BC})\right)
    &=
    \sum_{\varsigma\in G_\mathcal{A}} p_\varsigma\, S\!\left(\rho^{(\varsigma)}_{BC}\right),\\
    S\!\left(\mathcal{C}(\rho_{B})\right)
    &=
    \sum_{\varsigma\in G_\mathcal{A}} p_\varsigma\, S\!\left(\rho^{(\varsigma)}_{B}\right).
\end{aligned}
\end{equation}
Substituting these expressions into the definition of conditional mutual information, we obtain
\begin{equation}
\begin{aligned}
    I_{\mathcal{C}(\rho_{ABC})}(A:C\mid B)
    ={}&
    S\!\left(\mathcal{C}(\rho_{AB})\right)
    +S\!\left(\mathcal{C}(\rho_{BC})\right)
    -S\!\left(\mathcal{C}(\rho_{ABC})\right)
    -S\!\left(\mathcal{C}(\rho_{B})\right)\\
    ={}&
    \sum_{\varsigma\in G_\mathcal{A}} p_\varsigma
    \Bigl[
    S\!\left(\rho^{(\varsigma)}_{AB}\right)
    +S\!\left(\rho^{(\varsigma)}_{BC}\right)
    -S\!\left(\rho^{(\varsigma)}_{ABC}\right)
    -S\!\left(\rho^{(\varsigma)}_{B}\right)
    \Bigr]
    -H(\{p_\varsigma\}).
\end{aligned}
\end{equation}
Recognizing the expression in brackets as the conditional mutual information of the state $\rho^{(\varsigma)}_{ABC}=\mathcal{V}_\varsigma \rho_{ABC} \mathcal{V}_\varsigma^\dagger$, we arrive at
\begin{equation}
    I_{\mathcal{C}(\rho_{ABC})}(A:C\mid B)
    =
    \sum_{\varsigma\in G_\mathcal{A}} p_\varsigma \,
    I_{\mathcal{V}_\varsigma\rho \mathcal{V}_\varsigma^\dagger}(A:C\mid B)
    -H(\{p_{\varsigma}\}).
\end{equation}
On the other hand, since $\mathcal{V}_\varsigma$ are Pauli operators and any of their truncations is unitary, it does not affect the entanglement structure. Hence
\begin{equation}
I_{\mathcal{V}_\varsigma\rho \mathcal{V}_\varsigma^\dagger} (A:C\mid B)=I_{\rho} (A:C\mid B).
\end{equation}
The first term in the R.H.S. then becomes $I_\rho(A:C\mid B) = 2\gamma$.
The second term is $2 \log \mathcal{D}$ by direct calculation.
This leads to our result:
\begin{equation}
    I_{\mathcal{C}(\rho_{ABC})} (A:C|B) = I_\rho (A:C|B) - H( \{ p_{\varsigma} \} ) =2\gamma -2\log \mathcal{D} = 2\delta.
    \label{eq: stee-from-state}
\end{equation}
The strong subadditivity immediately implies the positivity of the spurious contribution.
On the other hand, the channel $\mathcal C$ also identifies which stabilizers contribute to this spurious part.

Let $J_{ABC}$ denote the stabilizer group appearing in the reduced density matrix of the original ground state,
\begin{equation}
    \rho_{ABC}
    =
    \frac{1}{p^N}
    \sum_{g\in J_{ABC}} g .
\end{equation}
The state with an anyon of type $\varsigma$ inside the annulus is obtained by conjugating $\rho_{ABC}$ with the truncated Pauli string $\mathcal V_\varsigma$:
\begin{equation}
    \rho^{(\varsigma)}_{ABC}
    =
    \mathcal V_\varsigma \rho_{ABC}\mathcal V_\varsigma^\dagger .
\end{equation}
Since $\mathcal V_\varsigma$ and $g$ are Pauli operators, conjugation by $\mathcal V_\varsigma$ can only multiply $g$ by a phase. We therefore define $\chi_g(\varsigma)$ by
\begin{equation}
    \mathcal V_\varsigma g\mathcal V_\varsigma^\dagger
    =
    \chi_g(\varsigma) g .
\end{equation}
It follows that
\begin{equation}
    \rho^{(\varsigma)}_{ABC}
    =
    \frac{1}{p^N}
    \sum_{g\in J_{ABC}}
    \chi_g(\varsigma) g .
\end{equation}
Thus, the different anyon sectors have the same underlying Pauli stabilizer operators, but their stabilizer eigenvalues differ by phases determined by the characters $\chi_g$.
It is easy to see that $\chi_g: G_\mathcal{A}\rightarrow \mathbb{C}$ is a character of the finite abelian group $G_\mathcal{A}$, so
\begin{equation}
\chi_g(\varsigma)\chi_g(\beta)=\chi_g(\varsigma+\beta).
\end{equation}
Additionally, the image of $\chi_g$ should always be a root of unity, since
\begin{equation}
    \chi_g(\varsigma)^p=\chi_g(p\varsigma)=0,
\end{equation}
because the quantum dimension of a qudit is $p$.
The mixed state density operator can be calculated by
\begin{equation}
    \mathcal{C}(\rho_{ABC}) = \frac{1}{d^N|G_\mathcal{A}|} \sum_{g\in  J_{ABC}} g \left(\sum_{\varsigma\in G_\mathcal{A}}\chi_g(\varsigma)\right).
\end{equation}
Closure of the group under addition gives
\begin{equation}\label{eq: alpha A=A}
    \beta+G_\mathcal{A}=G_\mathcal{A}, \quad \beta\in G_\mathcal{A}.
\end{equation}
Then the phase factor $\sum_{\varsigma\in G_\mathcal{A}}\chi_g(\varsigma)$ can be written as:
\begin{equation}
\begin{aligned}
\sum_{\varsigma\in G_\mathcal{A}}\chi_g(\varsigma)&=\frac 1p \sum_{n=0}^{p-1}\sum_{\varsigma\in G_\mathcal{A}}\chi_g(\varsigma_0^n\varsigma)\\
&=\frac 1p \sum_{\varsigma\in G_\mathcal{A}}\chi_g(\varsigma)\left(\sum_{n=0}^{p-1}\chi_g(\varsigma_0)^n\right)\\
&=\frac 1 p \sum_{\varsigma\in G_\mathcal{A}}\chi_g(\varsigma) \delta _{1,\chi_g(\varsigma_0)}.
\end{aligned}
\end{equation}
Here $\varsigma_0\in G_\mathcal{A}$ is an arbitrarily chosen anyon. The first equality follows from substituting Eq.~\eqref{eq: alpha A=A} into the sum over $\varsigma\in G_\mathcal{A}$ by choosing $\beta$ to be $\varsigma_0^n$, $n=0,\dots, p-1$, respectively. The second equality follows from the fact that $\chi_g$ is an Abelian group homomorphism. The third equality follows from an identity of roots of unity. 

Therefore, $\sum_{\varsigma\in G_\mathcal{A}}\chi_g(\varsigma)$ is nonzero if and only if $\chi_g$ is the trivial representation, that is, if and only if $g$ commutes with all $\mathcal{V}_{\varsigma}$.
This, together with Lemma \ref{lemma: TEE contribution}, forms a proof of Lemma \ref{lemma: sTEE contribution}.

%%%%%%%%%%%%%%%%%%%%%%%%%%%%%%%%%%%%%%%%%%%%%%%%%%%%%%%%%%%%%%%%%%%%%%%%%%
\section{Gr\"obner bases for ideals and modules}\label{app:Grobner}

This appendix summarizes the Gr\"obner-basis results used in the proof of Theorem~\ref{thm: improved Levin-wen}. In Appendix~\ref{app: The proof of lemmas and theorems}, we repeatedly replace the original translation-invariant stabilizer generators by alternative generating sets adapted to the analysis of bulk stabilizers, half-plane restrictions, and boundary gauge operators. Gr\"obner bases make these replacements systematic, and the accompanying degree bounds ensure that the new generators still have uniformly bounded support. This algebraic input is what makes the finite-range estimates in the main text possible, in particular the bound on $r_{\mathrm{bulk}}$.

The stabilizer formalism is naturally written over the Laurent ring \(R=\mathbb Z_p[x^{\pm1},y^{\pm1}]\). Whenever a Gr\"obner-basis argument is needed, we multiply each generator by a suitable monomial \(x^ay^b\). Because monomials are units in \(R\), this only translates the corresponding stabilizer on the lattice and does not change the submodule it generates. We may therefore work over the polynomial ring
\begin{equation}
\mathcal R=K[x_1,\dots,x_n],
\end{equation}
with \(K\) a field. In our application \(K=\mathbb Z_p\) with \(p\) prime, so standard Gr\"obner theory applies. We also write \(F=\mathcal R^m\) for a free \(\mathcal R\)-module of rank \(m\). Upper indices such as \(g^{(i)}\) label generators and are not exponents.

The discussion of ideals below is included mainly to fix notation; the stabilizer application uses the module version.

\subsection{Ideals and Buchberger's criterion}

A monomial order \(\preceq\) on \(\mathcal R\) is a total order on monomials such that \(M\preceq N\) implies \(MP\preceq NP\) for every monomial \(P\), and \(1\preceq M\) for every monomial \(M\). Every monomial order is a well-order, so each nonzero \(f\in \mathcal R\) has a well-defined leading monomial, leading coefficient, and leading term:
\begin{equation}
\operatorname{lm}(f),\qquad
\operatorname{lc}(f),\qquad
\operatorname{lt}(f)=\operatorname{lc}(f)\operatorname{lm}(f).
\end{equation}

Given a finite set \(G=\{g^{(1)},\dots,g^{(s)}\}\subset R\), one reduces \(f\in R\) modulo \(G\) by repeatedly subtracting suitable multiples of the \(g^{(i)}\) so as to cancel terms whose monomials are divisible by some \(\operatorname{lm}(g^{(i)})\). The output is a remainder, or normal form. For an arbitrary generating set this remainder may depend on the sequence of reductions; for a reduced Gr\"obner basis it is unique.

For \(f,g\in \mathcal R\), the \(S\)-polynomial is
\begin{equation}
S(f,g):=
\frac{\operatorname{lcm}(\operatorname{lm}(f),\operatorname{lm}(g))}{\operatorname{lt}(f)}\,f
-
\frac{\operatorname{lcm}(\operatorname{lm}(f),\operatorname{lm}(g))}{\operatorname{lt}(g)}\,g,
\end{equation}
where $\operatorname{lcm}(a,b)$ denotes the least common multiple of the monomials $a$ and $b$.
A finite set \(G\) is a Gr\"obner basis of the ideal \(I=\langle G\rangle\) if and only if every nonzero \(f\in I\) has leading monomial divisible by \(\operatorname{lm}(g)\) for some \(g\in G\). By Buchberger's criterion~\cite{Buchberger1965algorithm}, this is equivalent to requiring that every \(S(g^{(i)},g^{(j)})\) reduce to zero modulo \(G\) via polynomial long division.
%\bowen{Define what being "reduced" mean}

Two standard consequences will be used implicitly later. First, division by a Gr\"obner basis solves ideal membership: \(f\in I\) if and only if the remainder of \(f\) upon division by \(G\) is zero. Second, for a fixed monomial order, every ideal has a unique reduced Gr\"obner basis, meaning that each basis element is monic and that no monomial in any basis element is divisible by the leading monomial of another. Buchberger's algorithm constructs a Gr\"obner basis by starting from any generating set and repeatedly adjoining nonzero remainders of \(S\)-polynomials; a final reduction step then yields the reduced basis. Termination follows from the Noetherianity of \(\mathcal R\): the ideal generated by the leading monomials can strictly increase only finitely many times.

\subsection{Gr\"obner bases for submodules of free modules}

The same ideas extend from ideals to submodules of a free module \(F=\mathcal R^m\). Let \(e_1,\dots,e_m\) be the standard basis of \(F\). The analogue of an ordinary monomial is a module monomial \(x^\alpha e_j\), and one fixes a module term order on these objects. Two standard choices are TOP (term-over-position), which compares the polynomial monomials first and uses the basis index only to break ties, and POT (position-over-term), which compares the basis index first and the polynomial monomials second.

For
\begin{equation}
v=\sum_{\alpha,j} c_{\alpha,j}\,x^\alpha e_j\in F,\qquad c_{\alpha,j}\in K,
\end{equation}
the leading module monomial \(\operatorname{LM}(v)\) is the largest \(x^\alpha e_j\) appearing with nonzero coefficient, the leading coefficient is \(\operatorname{LC}(v)\), and the leading term is
\begin{equation}
\operatorname{LT}(v)=\operatorname{LC}(v)\operatorname{LM}(v).
\end{equation}
Divisibility is componentwise:
\begin{equation}
x^\alpha e_i \mid x^\beta e_j
\qquad\Longleftrightarrow\qquad
i=j\ \text{and}\ x^\alpha\mid x^\beta.
\end{equation}
Accordingly, reduction can only cancel terms in the same basis component. If \(\operatorname{LM}(g)=x^\alpha e_i\) and \(w\) contains a term \(c\,x^\beta e_i\) with \(x^\alpha\mid x^\beta\), one performs the reduction
\begin{equation}
w\longmapsto w-\frac{c}{\operatorname{LC}(g)}x^{\beta-\alpha}g,
\end{equation}
and repeats until no further cancellation is possible.

If \(u,v\in F\) have leading terms in the same component, say
\begin{equation}
\operatorname{LM}(u)=x^\alpha e_i,\qquad
\operatorname{LM}(v)=x^\beta e_i,
\end{equation}
let \(m=\operatorname{lcm}(x^\alpha,x^\beta)\). The corresponding \(S\)-vector is
\begin{equation}
S(u,v):=
\frac{m}{\operatorname{LT}(u)}\,u
-
\frac{m}{\operatorname{LT}(v)}\,v,
\end{equation}
where \(\frac{m}{\operatorname{LT}(u)}\) means \(\frac{m/x^\alpha}{\operatorname{LC}(u)}\in \mathcal R\), and similarly for \(v\). If the leading terms lie in different components, their leading module monomials cannot cancel, so that pair does not contribute an \(S\)-vector.

A finite set \(G\subset F\) is a Gr\"obner basis of a submodule \(M\subset F\) if every nonzero \(w\in M\) has \(\operatorname{LM}(g)\mid \operatorname{LM}(w)\) for some \(g\in G\). Equivalently, every relevant \(S\)-vector of pairs in \(G\) reduces to zero modulo \(G\). As in the ideal case, division by \(G\) tests membership in \(M\), and a reduced Gr\"obner basis gives a unique normal form. Buchberger's algorithm extends verbatim to this setting by using module \(S\)-vectors instead of ordinary \(S\)-polynomials.

\subsection{Degree bounds and the stabilizer module}

For the arguments in the main text, the key issue is not merely the existence of Gr\"obner bases but control of degree growth. Dub\'e showed that if an ideal in \(K[x_1,\dots,x_n]\) is generated by polynomials of total degree at most \(D\), then it admits a Gr\"obner basis whose elements have total degree bounded by~\cite{Dube1990Structure}
\begin{equation}
\deg(G)\le
2\left(\frac{D^2}{2}+D\right)^{2^{\,n-1}}.
\end{equation}
For \(n=2\), this simplifies to
\begin{equation}
\deg(G)\le \frac12 D^2(D+2)^2.
\end{equation}

An analogous bound holds for submodules of free modules~\cite{Liang2022GroebnerModules}. Endow \(F=\mathcal R^m\) with the grading
\begin{equation}
\deg\!\left(\sum_{j=1}^m f_j e_j\right):=\max_{1\le j\le m}\deg(f_j),
\end{equation}
where \(\deg(f_j)\) is the total degree of \(f_j\in \mathcal R\). Let \(M\subset F\) be generated by elements of degree at most \(D\), and define
\begin{equation}
\rho:=\dim(F/M)
=
\dim\bigl(\mathcal R/\operatorname{Ann}_R(F/M)\bigr),
\qquad
\operatorname{Ann}_{\mathcal R}(F/M):=\{a\in \mathcal R\mid aF\subset M\},
\end{equation}
where \(\dim(F/M)\) denotes Krull dimension.

\begin{theorem}[Theorem~8.4 of Ref.~\cite{Liang2022GroebnerModules}]\label{thm:Liang_module_bound}
With the notation above, for any module term order on \(F\), the reduced Gr\"obner basis \(G\) of \(M\) satisfies
\begin{equation}
\deg(G)\le
2\left[
\frac{
\bigl((Dm-1)(n-\rho)+1\bigr)^{\,n-\rho}m + D
}{2}
\right]^{2^\rho}.
\end{equation}
\end{theorem}

We now specialize to the stabilizer modules in the main text. There \(n=2\), \(K=\mathbb Z_p\), and the free Pauli module has rank \(m=2q\), with one \(X\)- and one \(Z\)-component for each of the \(q\) qudits in a unit cell. The submodule \(M\subset \mathcal R^{2q}\) is the stabilizer module in the sense of Ref.~\cite{haah_module_13}. For the topological stabilizer modules considered in this paper, the quotient \(F/M\) has full Krull dimension, so \(\operatorname{Ann}_{ \mathcal R}(F/M)=0\) and therefore \(\rho=2\). The general bound above then becomes
\begin{equation}\label{eq:bound_deg_G}
\deg(G)
\le
2\left(\frac{m+D}{2}\right)^4
=
\frac18\,(D+2q)^4.
\end{equation}

This is the only quantitative Gr\"obner-basis estimate needed later. After a monomial translation, a total-degree bound immediately bounds each individual exponent, so Eq.~\eqref{eq:bound_deg_G} also controls the horizontal and vertical extent of every Gr\"obner-basis generator. In Appendix~\ref{app: The proof of lemmas and theorems}, we use this fact, with module orders adapted to the geometry of the proof, to show that the auxiliary bulk stabilizers and boundary gauge operators have uniformly bounded support. This is precisely the input used in the main text to bound \(r_{\mathrm{bulk}}\) and, ultimately, the linear size \(L\) required for the concave Levin-Wen partition.

\subsection{Example: a Gr\"obner basis for a stabilizer submodule}\label{app:example_module_grobner}

We now give a simple example of a Gr\"obner-basis computation for a stabilizer submodule. Unlike the ideal case, the module setting is the one directly relevant to the main text, since translation-invariant stabilizers form a submodule of the free Pauli module. This example also illustrates that the resulting Gr\"obner basis depends on the chosen module order.

Consider the polynomial ring
\(
\mathcal R=\mathbb Z_2[x,y]
\)
and the free module \(\mathcal R^4\), corresponding to two qubits per unit cell. We write an element of \(\mathcal R^4\) as a column vector
\begin{equation}
\begin{bmatrix}
a_1\\ a_2\\ b_1\\ b_2
\end{bmatrix},
\end{equation}
where \(a_i\) and \(b_i\) denote the \(X\)- and \(Z\)-components, respectively. Let
\begin{equation}
S_1=
\begin{bmatrix}
1\\ x\\ xy\\ x^2y
\end{bmatrix}
=\vcenter{\hbox{\includegraphics[scale=.2]{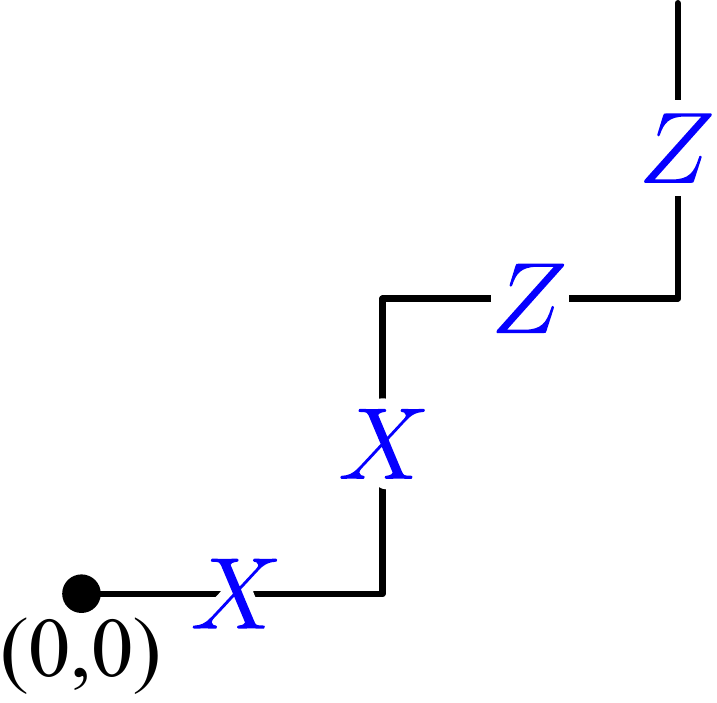}}},
\qquad
S_2=
\begin{bmatrix}
x+xy\\ x^2\\ y\\ x+xy
\end{bmatrix}
=\vcenter{\hbox{\includegraphics[scale=.2]{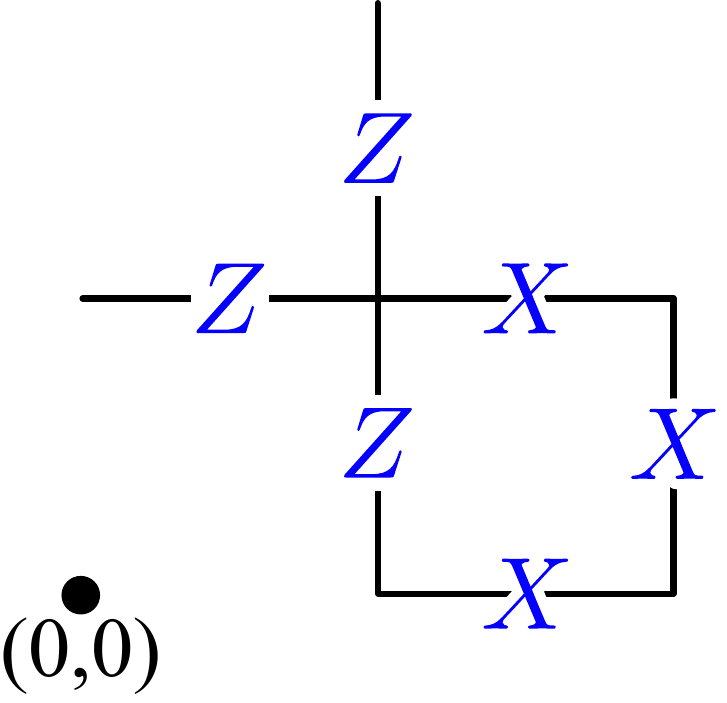}}},
\end{equation}
and let
\begin{equation}
M=\langle S_1,S_2\rangle\subset \mathcal R^4
\end{equation}
be the submodule that they generate.

We first check that these two translation-invariant Pauli operators commute. For this purpose, it is convenient to regard them as elements of the Laurent module \(R^4\), where
\(
R=\mathbb Z_2[x^{\pm1},y^{\pm1}],
\)
equipped with the standard symplectic form
\begin{equation}
\lambda_2=
\begin{bmatrix}
0&I_2\\
-I_2&0
\end{bmatrix}.
\end{equation}
The commutation condition for \(f,g\in R^4\) is
\begin{equation}
\bar f^{\,T}\lambda_2 g=0,
\end{equation}
where the bar denotes the antipode \(x^ay^b\mapsto x^{-a}y^{-b}\). Since we work over \(\mathbb Z_2\), the minus sign in \(\lambda_2\) is equivalent to a plus sign.
A direct calculation gives
\begin{align}
\bar S_1^{\,T}\lambda_2 S_1
&=
1\cdot xy+x^{-1}\cdot x^2y+(xy)^{-1}\cdot 1+(x^2y)^{-1}\cdot x = 0~,\\
\bar S_2^{\,T}\lambda_2 S_2
&=
(\overline{x+xy})\,y+x^{-2}(x+xy)+y^{-1}(x+xy)+(\overline{x+xy})\,x^2
=0~,\\
\bar S_1^{\,T}\lambda_2 S_2
&=
1\cdot y+x^{-1}(x+xy)+(xy)^{-1}(x+xy)+(x^2y)^{-1}x^2 = 0~.
\end{align}
Hence \(S_1\) and \(S_2\) form a commuting stabilizer set.

We now compute a Gr\"obner basis of \(M\). Since in the bulk-stabilizer problem the polynomial term is more important than the component label, we use a TOP (term-over-position) order. More precisely, we first compare polynomial monomials using the lexicographic order with
\begin{equation}
y\succ x,
\end{equation}
and only if the polynomial monomials are equal do we compare positions. In this example, we choose the position order
\begin{equation}
1<3<2<4,
\end{equation}
since the first and third entries correspond to the horizontal edges, while the second and fourth entries correspond to the vertical edges.
With this convention, the leading term of \(S_1\) is
\begin{equation}
\LT(S_1)=x^2y\,e_4,
\end{equation}
since \(x^2y\) is the largest monomial appearing in \(S_1\). For \(S_2\), the largest polynomial monomial is \(xy\), which occurs in the first and fourth components. Since \(4\) is the largest position in the chosen order, we obtain
\begin{equation}
\LT(S_2)=xy\,e_4.
\end{equation}
The two leading terms therefore lie in the same component, so we must form the \(S\)-vector. Because
\begin{equation}
\operatorname{lcm}(x^2y,xy)=x^2y,
\end{equation}
we find
\begin{equation}
S(S_1,S_2)=S_1+xS_2,
\end{equation}
where subtraction is the same as addition over \(\mathbb Z_2\).
Define
\begin{equation}
G_1:=S_1+xS_2=
\begin{bmatrix}
1+x^2+x^2y\\ x+x^3\\ 0\\ x^2
\end{bmatrix}
=\vcenter{\hbox{\includegraphics[scale=.2]{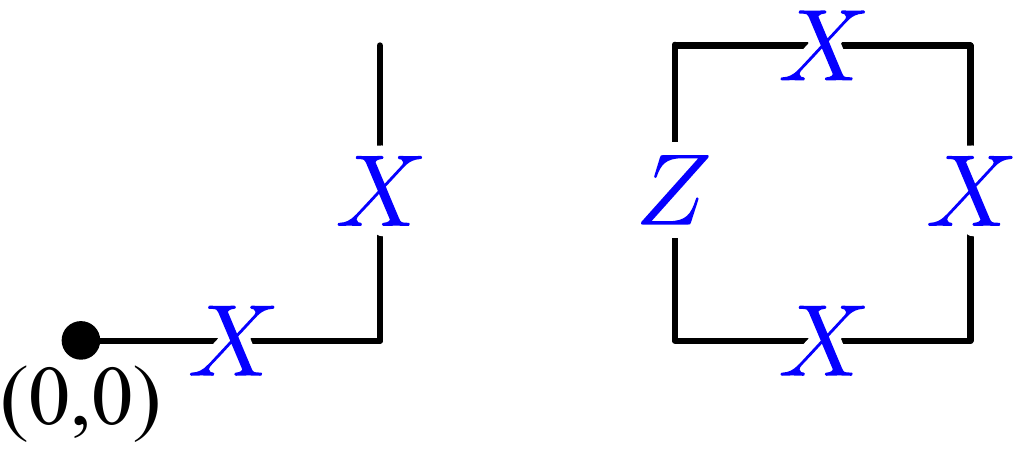}}},
\qquad
G_2:=S_2=
\begin{bmatrix}
x+xy\\ x^2\\ y\\ x+xy
\end{bmatrix}
=\vcenter{\hbox{\includegraphics[scale=.2]{submodule_S2.pdf}}}.
\end{equation}We now examine their leading terms. For \(G_1\), the largest monomial is \(x^2y\), and it appears in the first component, so
\begin{equation}
\LT(G_1)=x^2y\,e_1.
\end{equation}
For \(G_2\), the leading term remains
\begin{equation}
\LT(G_2)=xy\,e_4.
\end{equation}
These two leading terms lie in different components. Hence there is no further \(S\)-vector to consider, and Buchberger's criterion for submodules implies that
\(
\{G_1,G_2\}
\)
is a Gr\"obner basis of \(M\).
The chosen module order can change the form of the Gr\"obner basis. Here the TOP order is natural because it prioritizes the polynomial support, which is the quantity that controls the height of stabilizers. In this example, \(G_1\) and \(G_2\) therefore provide a generating set for the bulk stabilizers in the upper half-plane \(y\ge 0\).

%%%%%%%%%%%%%%%%%%%%%%%%%%%%%%%%%%%%%%%%%%%%%%%%%%%%%%%%%%%%%%%%%%%%%%%%%%%%%%%%
% \section{The proof of Lemmas~\ref{lemma: height limit} and \ref{lemma: string size bound} and Theorem~\ref{thm: improved Levin-wen}}\label{app: The proof of lemmas and theorems}
\section{Support bounds for trivial anyons, structure of boundary gauge operators, and proof of Theorem~\ref{thm: improved Levin-wen}}\label{app: The proof of lemmas and theorems}
\subsection{Notation}
In a two-dimensional translationally invariant stabilizer code, we are concerned with three kinds of objects: Pauli operators, stabilizers, and anyons, as well as the relations among them. Pauli operators are characterized by their commutation relations. Stabilizers are generated by products of translated stabilizer generators. Anyons are defined as violations of these stabilizer generators, and such violations are determined by the commutation relations between Pauli operators and stabilizers. 

The polynomial notation used below is standard in the theory of translationally invariant stabilizer codes: the variables $x$ and $y$ keep track of lattice coordinates, so that spatial translations are encoded algebraically. With this notation, the spatial structure of these objects can be represented by elements of modules over a polynomial ring $R$, and the relations among them—commutation, generation of stabilizers, and creation of syndromes—can be expressed in transparent algebraic forms. 

\subsubsection{Defining $P$ and $\hat P$, the modules of Pauli operators}
To encode configurations with finite support, we introduce the Laurent polynomial ring
\begin{equation}
    R:=\mathbb{Z}_p[x^{\pm1},y^{\pm1}].
\end{equation}
A monomial \(x^n y^m\) represents a translation by \((n,m)\). To also allow configurations with infinite support, we use
\begin{equation}
\hat R:=\mathbb{Z}_p\llbracket x^{\pm1},y^{\pm1}\rrbracket,
\end{equation}
the \(\mathbb{Z}_p\)-module of formal \(\mathbb{Z}_p\)-linear combinations of monomials \(x^n y^m\), which might be infinite. Notice that here \(\hat R\) is not a ring, because the product of two infinitely supported expressions may not be well defined.

If the unit cell contains \(q\) qudits, then each qudit contributes one \(X\)-type and one \(Z\)-type Pauli component. This motivates the definitions
\begin{equation}
    P:=R^{2q},\qquad \hat P:=\hat R^{2q}.
\end{equation}
Thus \(P\) represents Pauli operators with finite support, while \(\hat P\) allows possibly infinite support. Equivalently,
\begin{align}
    &v\in P\subset \hat P \qquad \text{if \(v\) has finite support},\\
    &v\in \hat P\setminus P \qquad \text{if \(v\) has infinite support}.
\end{align}
We use lowercase Latin letters such as \(u,v,w\) for elements of \(P\) or \(\hat P\). As a simple example, if there is one qudit on each horizontal edge and one on each vertical edge of the square lattice, then \(q=2\).

Uppercase calligraphic letters, such as \(\mathcal{S}\), denote the corresponding operators on the Hilbert space. There is a natural \(\mathbb{Z}_p\)-linear map
\begin{equation}
    \mathcal{O}:\hat P\to\mathcal{A},\qquad v\mapsto \mathcal{O}(v),
\end{equation}
from the algebraic description to the operator algebra \(\mathcal{A}\).

A central advantage of this formalism is that Pauli commutation relations are encoded by a symplectic form. We define
\begin{equation}
    *:(\hat P\times P)\cup(P\times \hat P)\to \mathbb{Z}_p
\end{equation}
by
\begin{equation}
    v*v'=\operatorname{const}(\bar v^{\,T}\Lambda v')=-v'*v,\qquad v\in P,\quad v'\in \hat P,
\end{equation}
where the bar denotes the involution
\begin{equation}
    \overline{x}=x^{-1},\qquad \overline{y}=y^{-1},
\end{equation}
extended \(\mathbb{Z}_p\)-linearly to \(R\) and \(\hat R\), and $\Lambda $ is the anti-symmetric matrix defined by
\begin{equation}
    \Lambda:=
    \begin{bmatrix}
        0_{q\times q} & I_{q\times q}\\
        -I_{q\times q} & 0_{q\times q}
    \end{bmatrix}.
\end{equation}
This pairing is defined whenever at least one of the two arguments has finite support, which guarantees that the constant term is well defined. Its physical relation to the commutation phase is given by:
\begin{align}
    \mathcal{O}(v)\mathcal{O}(v')
    &=\omega^{\phi}\mathcal{O}(v')\mathcal{O}(v),\\
    \phi&=v*v',
\end{align}
where
\begin{equation}
    \omega=\exp(2\pi i/p).
\end{equation}

For a subset \(A\subset P\), we define the subspaces orthogonal with respect to this form as
\begin{align}
    A^{\perp}&:=\{c\in \hat P\mid c*a=0,\ \forall a\in A\},\label{def:perpP}\\
    A^{\Omega}&:=\{c\in P\mid c*a=0,\ \forall a\in A\}.\label{def:omegaP}
\end{align}
For a subset \(B\subset \hat P\), we define
\begin{equation}
    B^{\perp}:=\{c\in P\mid c*b=0,\ \forall b\in B\}.\label{def:perphatP}
\end{equation}
Physically, the orthogonal of a subset $A$ of $P$ or $\hat P$, is the set of operators that commute with all operators in $A$. 

\subsubsection{Defining $S$ and $\hat S$, the module of stabilizers and anyons}
We now introduce the translation-invariant stabilizer generators. Let
\begin{equation}
    g^1,\dots,g^{n_S}\in (\mathbb{Z}_p[x,y])^{2q}
\end{equation}
be a chosen set of local generators, where the superscript labels the generator and is not an exponent. Since the multiplication of elements in a module is not defined, we will denote the indices of the module generators by superscripts and the component indices by subscripts. For example, if
\begin{equation}
    g^\mu=(g^\mu_1,\dots,g^\mu_{2q}),
\end{equation}
then \(g^\mu_i\) denotes the \(i\)th component of \(g^\mu\).

Under the topological order assumption, every bulk stabilizer can be written as an \(R\)-linear combination of these generators, and every possibly infinite bulk stabilizer as an \(\hat R\)-linear combination of the same set. It is therefore useful to introduce
\begin{equation}
    S:=R^{n_S},\qquad \hat S:=\hat R^{n_S}.
\end{equation}
An element of \(S\) or \(\hat S\) records a stabilizer as a linear combination of the generators \(g^\mu\), whereas an element of \(P\) or \(\hat P\) records the corresponding Pauli operator. We use lowercase Greek letters such as \(\theta,\eta,\xi\) for elements of \(S\) or \(\hat S\).

The map from generator coefficients to the corresponding Pauli operator is
\begin{equation}
    w:\hat S\to \hat P,
\end{equation}
whose restriction to \(S\) maps \(S\) to \(P\). Explicitly,
\begin{align}
    \theta&=(\theta_1,\dots,\theta_{n_S})\in S,\qquad \theta_\mu\in R,\\
    w(\theta)&=\sum_{\mu=1}^{n_S} \theta_\mu g^\mu\in P,\\
    \hat\theta&=(\hat\theta_1,\dots,\hat\theta_{n_S})\in \hat S,\qquad \hat\theta_\mu\in \hat R,\\
    w(\hat\theta)&=\sum_{\mu=1}^{n_S} \hat\theta_\mu g^\mu\in \hat P.
\end{align}

An anyon is represented by a finite syndrome, namely an element \(\alpha\in S\). For instance, if the first component of \(\alpha\) is \(1+x+2y^2\), then \(\alpha\) indicates violations of \(g^1\), \(xg^1\), and \(y^2g^1\) with phases \(\omega\), \(\omega\), and \(\omega^2\), respectively, and no violation of any other translate \(x^n y^m g^1\).

How a given anyon violates a stabilizer is encoded by another constant-term bilinear form,
\begin{equation}
    \langle\cdot,\cdot\rangle:(S\times \hat S)\cup(\hat S\times S)\to \mathbb{Z}_p,
\end{equation}
defined by
\begin{equation}
    \langle \alpha,\theta\rangle=\operatorname{{const}}\left(\sum_{\mu=1}^{n_S} \bar\alpha_\mu\,\theta_\mu\right)
    =\langle \theta,\alpha\rangle,
    \qquad \alpha\in S,\quad \theta\in \hat S.
\end{equation}
If \(\hat a^\dagger\) creates the syndrome \(\alpha\), then for a stabilizer \(\mathcal{S}=\mathcal{O}(w(\theta))\),
\begin{align}
    \mathcal{S}\hat a^\dagger=\omega^{\phi'}\hat a^\dagger\mathcal{S}, \quad
    \phi'=\langle \alpha,\theta\rangle.
\end{align}
Physically, \(\langle \alpha,\theta\rangle\) records the phase produced when the excitation \(\alpha\) violates the stabilizer \(\mathcal{O}(w(\theta))\).

The error syndrome created by a finitely supported Pauli operator is described by the linear map
\begin{equation}
    \varepsilon:P\to S.
\end{equation}
For \(v=(v_1,\dots,v_{2q})\in P\), we write
\begin{equation}
    \varepsilon(v)=\bigl(\varepsilon(v)_1,\dots,\varepsilon(v)_{n_S}\bigr)\in S,
\end{equation}
with components
\begin{equation}\label{eq: definition of epsilon(v)_mu}
    \varepsilon(v)_\mu=\sum_{i=1}^{2q}\bar g^\mu_i(\Lambda v)_i.
\end{equation}

\subsubsection{Characterizing spatial supports and defining projections}
To derive rigorous bounds on the system size, we need to keep careful track of the spatial support of lattice elements. We therefore introduce a set of quantities that describe the support of a polynomial in \(R\) or \(\hat R\), together with projection maps onto specified spatial regions. These definitions extend naturally to modules over \(R\), so that the same notation can be used uniformly for Pauli operators, stabilizers, and anyons.

For \(v\) in a module over \(R\), we define its support by
\begin{equation}
\supp(v):=\{(n,m)\in\mathbb Z^2 \mid \text{some component of }v\text{ has a nonzero coefficient of }x^n y^m\}.
\end{equation}

We denote by \(x_{\max}(v)\) the maximum \(x\)-coordinate of \(\supp(v)\); \(y_{\max}(v)\), \(x_{\min}(v)\), and \(y_{\min}(v)\) are defined analogously. These quantities characterize the spatial support of an object on the lattice and will be used frequently in later proofs. If \(v\) has finite support, then \(x_{\max}(v)\), \(y_{\max}(v)\), \(x_{\min}(v)\), and \(y_{\min}(v)\) are all finite, and \(v\) is contained in the rectangular region
\begin{equation}
\BBox(v):=[x_{\min}(v),x_{\max}(v)]\times [y_{\min}(v),y_{\max}(v)]\cap \mathbb Z^2,
\end{equation}
which we call the bounding box of \(v\). By definition, \(\supp(v)\subset \BBox(v)\). Since all supports considered here are subsets of \(\mathbb Z^2\), we will often suppress the explicit intersection with \(\mathbb Z^2\) and simply write
\begin{equation}
\supp(v)\subset [x_1,x_2]\times [y_1,y_2]
\end{equation}
in place of
\begin{equation}
\supp(v)\subset [x_1,x_2]\times [y_1,y_2]\cap \mathbb Z^2.
\end{equation}

For \(v\) with infinite support, however, the quantities \(x_{\max}(v)\), \(y_{\max}(v)\), \(x_{\min}(v)\), and \(y_{\min}(v)\) may themselves be infinite, and the bounding box \(\BBox(v)\) may be unbounded.
% a figure here to show a rectangular region

Since all generators are local, it is convenient to choose the origin so that every component of every generator starts at the lower-left corner of its support:
\begin{align}\label{eq: min_x min+y g^mu_i =0}
    x_{\min}(g^\mu_i)=y_{\min}(g^\mu_i)=0,\qquad i=1,\dots,2q,\ \mu=1,\dots,n_S.
\end{align}
We then define the maximal extents of the generators by
\begin{equation}\label{eq: definition of ls and r}
\begin{aligned}
    l^\mu_{xi}&:=x_{\max}(g^\mu_i),\qquad
    l^\mu_{yi}:=y_{\max}(g^\mu_i),\qquad i=1,\dots,2q,\ \mu=1,\dots,n_S,\\
    l^\mu_x&:=\max_{1\le i\le 2q} \{l^\mu_{xi}\},\quad
    l^\mu_y:=\max_{1\le i\le 2q}\{ l^\mu_{yi}\},\qquad \mu=1,\dots,n_S,\\
    l_x&:=\max_{1\le \mu\le n_S} \{l^\mu_x\},\quad
    l_y:=\max_{1\le \mu\le n_S} \{l^\mu_y\},\\
    r&:=\max\{l_x,l_y\}.
\end{aligned}
\end{equation}

These bounds imply a simple relation between the support of a coefficient vector \(\theta\in S\) and that of the Pauli operator \(w(\theta)\in P\) generated by it:
\begin{equation}
    \supp\bigl(w(\theta)\bigr)\subset
    [x_{\min}(\theta),x_{\max}(\theta)+r]\times
    [y_{\min}(\theta),y_{\max}(\theta)+r].
\end{equation}
In words, the support of \(w(\theta)\) can extend beyond that of \(\theta\) only by an amount determined by the range of the local generators. We refer to this property as the locality of the generators.

At several points we restrict attention to a particular spatial region and keep only the part of an element supported there. For the upper half-plane, we use the projection
\begin{align}
    \pi &:R\to \mathbb{Z}_p[x^{\pm1},y],\\
    \pi &:\hat R\to \mathbb{Z}_p\llbracket x^{\pm1},y\rrbracket,
\end{align}
which removes all monomials supported at lattice sites with negative \(y\)-coordinate, that is, all terms with negative powers of \(y\).

More generally, we will often project onto vertical or horizontal strips. We therefore define
\(\pi^x_{x_1:x_2}\) and \(\pi^y_{y_1:y_2}\) by retaining only those monomials whose support lies in the strip
\([x_1,x_2)\times \mathbb Z\) or \(\mathbb Z\times [y_1,y_2)\), respectively:
\begin{align}
    \pi^x_{x_1:x_2}\left(\sum_i k_i x^{n_i}y^{m_i}\right)
    &=\sum_{\{i|x_1\le n_i<x_2 \}} k_i x^{n_i}y^{m_i},\\
    \pi^y_{y_1:y_2}\left(\sum_i k_i x^{n_i}y^{m_i}\right)
    &=\sum_{\{i|y_1\le m_i<y_2\}} k_i x^{n_i}y^{m_i}.
\end{align}
Thus \(\pi\), \(\pi^x_{x_1:x_2}\), and \(\pi^y_{y_1:y_2}\) simply discard the part of an element supported outside the chosen region. The same symbols are also used for the componentwise action of these projections on direct products such as \(P\), \(\hat P\), \(S\), and \(\hat S\).

Finally, these projection maps satisfy a useful identity. The commutation phase between two Pauli operators depends only on the overlap of their supports. Therefore, projecting one operator onto a region has the same effect as projecting the other onto the same region when computing the pairing. In particular,
\begin{equation}\label{eq: pi p_1*p_2}
    \pi p_1 * p_2 = p_1 * \pi p_2,\qquad \forall\, p_1,p_2\in P.
\end{equation}
The same identity holds for any of the projection maps introduced above.

\subsection{Proof of Lemma~\ref{lemma: string size bound}}

By definition, every trivial anyon can be created by a finitely supported Pauli operator, but for a given trivial anyon there is generally no obvious bound on the size of such an operator. In this subsection, we derive an explicit support bound and thereby prove Lemma~\ref{lemma: string size bound}.

\begin{replemma}{lemma: string size bound}
In a two-dimensional translation-invariant topological stabilizer code with $q$ qudits per unit cell and stabilizer generators of range $r$, every trivial anyon fully supported in an $r' \times r'$ square can be created by a Pauli operator supported in the concentric square of linear size $2r' + 8r^3q^4 + 2r$.
\end{replemma}

The proof proceeds in several steps.
We first reformulate the problem of creating a trivial anyon as a local syndrome-decomposition problem. We then convert the resulting equation over a Laurent polynomial ring into one over an ordinary polynomial ring by introducing an additional variable \(t\). This allows us to apply homogenization and Gr\"obner-basis reduction to obtain a syndrome decomposition with explicitly bounded degrees of the coefficients. Finally, we translate this degree bound into a bound on the spatial support of the corresponding Pauli operator.

Beyond proving the lemma, this procedure also provides a constructive method for finding Pauli operators that create arbitrary trivial anyons and admits an algorithmic implementation.

\subsubsection{Reformulating the problem as a local syndrome decomposition}

The first step is to rewrite this physical statement as an algebraic problem for the syndrome.
By definition, a trivial anyon \(\alpha_0\) is a syndrome created by some Pauli operator. Thus there exists \(s\in P\) such that
\begin{equation}\label{eq: k_i epsilon(e_i)}
   \alpha_0=\varepsilon(s).
\end{equation}
Using the standard basis \(\{e_i\}_{i=1}^{2q}\) of the free module \(P=R^{2q}\), which corresponds to the $q$ Pauli Xs and $q$ Pauli $Z$s on each qudit of a site, we may write
\begin{equation}
    s=\sum_{i=1}^{2q}s_i e_i,
\end{equation}
and therefore
\begin{equation}\label{eq: alpha_0=sum_i=1^2q s_i epsilon(e_i)}
   \alpha_0=\varepsilon(s)= \sum_{i=1}^{2q} s_i \varepsilon(e_i).
\end{equation}
Hence, the lemma is already reduced to finding the polynomials \(s_i\) in Eq.~\eqref{eq: alpha_0=sum_i=1^2q s_i epsilon(e_i)} while controlling their support. 
However, the elementary syndromes \(\varepsilon(e_i)\) may contain negative powers of \(x\) and \(y\). This has no physical significance---it only reflects the arbitrary choice of coordinates---but it will lead to weaker bounds for the polynomial-ring reduction used below, where positive and negative exponents are treated asymmetrically. We therefore first apply a simple geometric transformation so that both the target anyon and the elementary syndromes are represented by polynomial objects with nonnegative exponents.

Examine the elementary syndromes \(\varepsilon(e_i)\), namely the syndromes created by the basis Pauli operators. From Eq.~\eqref{eq: definition of epsilon(v)_mu},
\begin{equation}
    \varepsilon(e_i)_\mu=\sum_{j=1}^{2q}\bar g^\mu_j(\Lambda e_i)_j.
\end{equation}
Since each stabilizer generator has range at most \(r\), these elementary syndromes are themselves local. More precisely, for every \(\mu=1,\dots,n_S\) and $i=1,\dots ,2q$,
\begin{equation}\label{eq: support of (Lambda e_i)_j}
\begin{aligned}
 & \supp(\varepsilon(e_i)_\mu)\subset \bigcup_{j=1}^{2q}\supp\left(\bar g^\mu_j(\Lambda e_i)_j)\right)\subset [-r,0]\times [-r,0]
\end{aligned}
\end{equation}
where the relation follows from the fact that $\supp(\Lambda e_i)$ is just a point and the size bound of generators. 

To place these elementary syndromes in a polynomial setting, we define
\begin{equation}
    \tilde g^i:=\overline{\varepsilon(e_i)}.
\end{equation}
Then Eq.~\eqref{eq: support of (Lambda e_i)_j} implies
\begin{equation}\label{eq: identities of tilde g^i}
\begin{aligned}
    &\tilde g^i\in (\mathbb{Z}_p[x,y])^{n_S}\subset S,\\
    &x_{\max}(\tilde g^i)\le r,\qquad y_{\max}(\tilde g^i)\le r,\qquad \deg(\tilde g^i)\le 2r.
\end{aligned}
\end{equation}
So each \(\tilde g^i\) is a local polynomial object of range at most \(r\). Geometrically, $\tilde g^i$ is just the spatial reflection of $g^i$.

By translation invariance, we may choose the origin so that the support of \(\alpha_0\) begins at the lower-left corner of the square:
\begin{equation}
    x_{\min}(\alpha_0)=y_{\min}(\alpha_0)=0.
\end{equation}
Since \(\alpha_0\) is fully supported in an \(r'\times r'\) square, we have
\begin{equation}\label{eq: support of alpha_0mu}
\begin{aligned}
        &\supp (\alpha_0)\subset [0,r']\times [0,r'].
\end{aligned}
\end{equation}

To obtain a syndrome decomposition using $\tilde g^i$, we translate $\alpha_0$ after spatial reflection so that it remains a polynomial object with nonnegative exponents. We therefore define
\begin{equation}
    \alpha_1:=x^{r'}y^{r'}\bar \alpha_0 .
\end{equation}

Equation~\eqref{eq: support of alpha_0mu} then implies
\begin{equation}\label{eq: identity of alpha_1}
\begin{aligned}
    &\alpha_1\in (\mathbb{Z}_p[x,y])^{n_S}\subset S,\\
    &x_{\max}(\alpha_1)\le r',\qquad y_{\max}(\alpha_1)\le r',\qquad \deg(\alpha_1)\le 2r'.
\end{aligned}
\end{equation}

After the change of variables
\begin{equation}
    x\mapsto x^{-1},\qquad y\mapsto y^{-1},
\end{equation}
and multiplication by \(x^{r'}y^{r'}\), Eq.~\eqref{eq: k_i epsilon(e_i)} becomes
\begin{equation}\label{eq: alpha_1 k_i}
    \alpha_1=\sum_{i=1}^{2q} k_i\,\tilde g^i,
    \qquad
    k:=x^{r'}y^{r'}\bar s\in P.
\end{equation}

Now, the anyon $\alpha_1$ and the generators ${\tilde g^i}$ only have nonnegative exponents. The coefficients, $k_i$, however, are still an unknown element in a Laurent polynomial ring, which means they can be a polynomial with negative exponents.
Geometrically, the substitution \(x\mapsto x^{-1}\), \(y\mapsto y^{-1}\) is simply a \(180^\circ\) rotation of the lattice, followed by a translation by \((r',r')\). These operations do not change the support size. Therefore, proving that \(k\) is supported within a certain neighborhood of \(\alpha_1\) is equivalent to proving that the original Pauli operator \(s\) is supported within the corresponding neighborhood of \(\alpha_0\).
%a figure here to show the geometrical meaning of the transformation

We have thus reduced the lemma to a local algebraic problem: express the polynomial syndrome \(\alpha_1\), supported in an \(r'\times r'\) square, as a linear combination of the local generators \(\tilde g^i\) which only have nonnegative exponents, and bound the support of the coefficients \(k_i\).

\subsubsection{Reduction to a polynomial-ring problem}
At this moment, the coefficients $k_i$ can be a polynomial with negative exponents.
Techniques for obtaining degree bounds over ordinary polynomial rings, such as the degree bound of Gr\"obner basis on modules, are far more developed than those for Laurent polynomial rings. Therefore, the next step is to transfer the problem from the Laurent polynomial ring to an ordinary polynomial ring.

Define maps \(f\) and \(f'\) between the Laurent polynomial ring and an ordinary polynomial ring by
\begin{equation}
\begin{aligned}
    f&:R=\mathbb{Z}_p[x^{\pm1},y^{\pm1}] \to R_0=\mathbb{Z}_p[x,y,t],\\
    f(x^ny^m)&=
    \begin{cases}
        x^ny^m, & n\ge 0,\ m\ge 0,\\
        y^{m-n}t^{-n}, & n<0,\ m\ge 0,\\
        x^{n-m}t^{-m}, & n\ge 0,\ m<0,\\
        x^{-m}y^{-n}t^{-n-m}, & n<0,\ m<0,
    \end{cases}
\end{aligned}
\end{equation}
and extend \(f\) \(\mathbb{Z}_p\)-linearly to \(R\).
We also define
\begin{equation}
\begin{aligned}
    f' : R_0 = \mathbb{Z}_p[x,y,t] 
    &\longrightarrow R = \mathbb{Z}_p[x^{\pm1},y^{\pm1}], \\
    x^n y^m t^\ell 
    &\longmapsto x^{n-\ell} y^{m-\ell}.
\end{aligned}
\end{equation}
and extend \(f'\) \(\mathbb{Z}_p\)-linearly to \(R_0\). Intuitively, $f$ is just applying a change of variable $t=x^{-1}y^{-1}$ to rewrite all monomials in a Laurent polynomial with nonnegative exponents of $x,y$ and $t$  and $f'$ is its inverse. 

\(f\) is not itself a ring homomorphism from \(R\) to \(R_0\) since it generally cannot commute with multiplication, for example,
\begin{equation}
f(x^{-1}y)=y^2t,f(xy^{-1})=x^2t,f(x^{-1}yxy^{-1})=f(1)=1\neq x^2y^2t^2.
\end{equation}
However, one can verify that it induces a ring homomorphism on a quotient ring
\begin{equation}
    \bar f:R\to R_0/\langle xyt-1\rangle,
\end{equation}
and the map \(f'\) induces the inverse homomorphism, so \(\bar f\) is an isomorphism:
\begin{equation}
    R\cong R_0/\langle xyt-1\rangle.
\end{equation}
Under this isomorphism, solving Eq.~\eqref{eq: alpha_1 k_i} is equivalent to solving
\begin{equation}\label{equation of alpha lifted}
    \sum_{i=1}^{2q} k'_i f(\tilde g^i)
    +\sum_{i=1}^{n_S} k''_i (xyt-1)\epsilon_i
    =f(\alpha_1),
    \qquad
    k'_i,k''_i\in R_0.
\end{equation}
Since \(\tilde g^i\) and \(\alpha_1\) are polynomial objects with only nonnegative exponents, their images under \(f\) are just themselves. Hence the equation simplifies to
\begin{equation}\label{eq: alpha}
    \sum_{i=1}^{2q} k'_i \tilde g^i
    +\sum_{i=1}^{n_S} k''_i (xyt-1)\epsilon_i
    =\alpha_1,
    \qquad
    k'_i,k''_i\in R_0.
\end{equation}
Here \(\{\epsilon_i\}_{i=1}^{n_S}\) is the standard basis of the module \(S=R^{n_S}\). For convenience, we extend the index of $\tilde g^i$ to $1\sim 2q+n_S$:
\begin{equation}
    \tilde g^i=(xyt-1)\epsilon_{i-2q},\qquad  i=2q+1,\dots,2q+n_S.
\end{equation}
then the two sums in Eq.~\eqref{eq: alpha} can be written into one uniform sum:
\begin{equation}\label{eq: k' falpha}
    \sum_{i=1}^{2q+n_S} k_i \tilde g^i=\alpha_1,
    \qquad
    k_i\in R_0.
\end{equation}
Thus the original support problem has been converted into a decomposition problem in the ordinary polynomial ring \(R_0\).

\subsubsection{Homogenization and Gr\"obner-basis construction}

We now solve Eq.~\eqref{eq: k' falpha} with the goal of obtaining an explicit degree bound on the coefficients \(k_i\). Gr\"obner-basis reduction provides a systematic way to construct such a solution. To make the degree estimate transparent, we first homogenize the problem and work in a graded polynomial ring. In this homogeneous setting, the degree flow through the reduction procedure can be tracked directly, which will later give the desired bound on \(k_i\). 

First, we introduce a polynomial ring
\begin{equation}
    R_0^h=\mathbb{Z}_p[z,x,y,t].
\end{equation}
For a polynomial \(u\in R_0\), its homogenization is given by
\begin{equation}
    u^h=z^{\deg(u)}u(x/z,y/z,t/z)\in R^h_0,
\end{equation}
and for a polynomial \(v\in R_0^h\), define its dehomogenization by
\begin{equation}
    v^{\backslash *}=v(z=1,x,y,t)\in R_0.
\end{equation}
Here \((\cdot)^h\) denotes homogenization and \((\cdot)^{\backslash *}\) denotes dehomogenization. In words, homogenization introduces powers of \(z\) so that every term acquires the same total degree, while dehomogenization sets \(z=1\). One checks that
\begin{equation}
    v^{h\backslash *}=v,\qquad \forall\, v\in R_0,
\end{equation}
and
\begin{equation}
    \deg(u^h)=\deg(u),\qquad u\in R_0.
\end{equation}
To obtain a Gr\"obner basis of $\langle \tilde g^i \rangle $ and control the degree of coefficients, we will calculate the Gr\"obner basis of the submodule generated by homogenized generators and then apply dehomogenization on it.
We first define a graded lexicographic order by
\begin{equation}
    x^{n_1}y^{n_2}t^{n_3}z^{n_4}\prec x^{m_1}y^{m_2}t^{m_3}z^{m_4}
\end{equation}
if and only if
\begin{equation}
\begin{cases}
n_1+n_2+n_3+n_4<m_1+m_2+m_3+m_4,\\
\text{or } n_1+n_2+n_3+n_4=m_1+m_2+m_3+m_4 \text{ and } n_1<m_1,\\
\text{or } n_1+n_2+n_3+n_4=m_1+m_2+m_3+m_4 \text{ and } n_1=m_1 \text{ and } n_2<m_2,\\
\text{or } n_1+n_2+n_3+n_4=m_1+m_2+m_3+m_4 \text{ and } n_1=m_1 \text{ and } n_2=m_2 \text{ and } n_3<m_3.
\end{cases}
\end{equation}

In words, monomials are ordered first by total degree, and ties are broken by comparing the exponents of $x$, then $y$, and then $t$.

Then, we compute the resulting Gr\"obner basis of  ${\langle (\tilde g^i)^h\rangle}$ under this order, and denote them by \(\{G^1,\dots,G^{n_G}\}\). Each \(G^l\) can be expanded in terms of the generators \(\{(\tilde g^i)^h\}\), as:
\begin{equation}\label{equation of b' G}
    \sum_{i=1}^{n_S+2q} b'_{il}(\tilde g^i)^h=G^l,
    \qquad
    b'_{il}\in R_0^h,\quad l=1,\dots,n_G.
\end{equation}
Since both \(G^l\) and \((\tilde g^i)^h\) are homogeneous, we can select a set of coefficients that each \(b'_{il}\) is also homogeneous, and
\begin{equation}\label{deg of b' G}
    \deg(b'_{il})=\deg(G^l)-\deg((\tilde g^i)^h),
    \qquad
    \text{whenever } b'_{il}\neq 0.
\end{equation}
This relation provides a transparent and tight bound on the decomposition coefficients, which is a major advantage of applying the homogenization technique.

\subsubsection{Krull dimension and Gr\"obner-basis bound}
Next step is to bound \(\deg(G^l)\). For using the degree bound for Gr\"obner bases of submodules from Ref.~\cite{Liang2022GroebnerModules} to obtain such a bound, we need to calculate the Krull dimension \(d_K\) of
\begin{equation}
    (R_0^h)^{n_S}/\langle (\tilde g^i)^h\rangle.
\end{equation}
We first determine the Krull dimension of \(R_0^{n_S}/\langle \tilde g^i\rangle\). The Krull dimension of the module \(R_0^{n_S}/\langle \tilde g^i\rangle\) equals that of the quotient ring by its annihilator:
\begin{equation}
    \dim\bigl(R_0^{n_S}/\langle \tilde g^i\rangle\bigr)
    =
    \dim\bigl(R_0/\Ann(R_0^{n_S}/\langle \tilde g^i\rangle)\bigr).
\end{equation}
The annihilator of an \(R\)-module \(M\) is defined by
\begin{equation}
    \Ann(M):=\{f\in R\mid fm=0,\ \forall m\in M\}\subset R.
\end{equation}
Intuitively, $\Ann(M)$ is the set of ring elements that kill every element of $M$, so it captures the part of the ring action that has no effect on the module. 

Now let \(L_x\) and \(L_y\) be the horizontal and vertical periods of all anyons of the code. Translation invariance and the finiteness of the number of anyon types guarantee their existence~\cite{liang2023extracting, ruba2024homological}. Physically, translating an anyon by \(L_x\) or \(L_y\) leaves its topological type unchanged, so
\begin{equation}
    (x^{L_x}-1)\alpha\sim (y^{L_y}-1)\alpha\sim \text{trivial anyon},
    \qquad \forall \alpha\in S.
\end{equation}
Remember that Eq.~\eqref{eq: k' falpha} deduces that each trivial anyon, under a geometrical transformation and a change of variables, can be expanded by $\tilde g^i$, which implies
\begin{equation}
    (x^{L_x}-1)f(\alpha),\ (y^{L_y}-1)f(\alpha)
    \in
    \langle \tilde g^1,\dots ,\tilde g^{2q+n_S}\rangle\subset R_0^{n_S}.
\end{equation}
Also, the isomorphism between $R\cong R_0/\langle (xyt-1)\epsilon_i\rangle$ implies that every \(\alpha'\in R_0^{n_S}\) can be written as
\begin{equation}
    \alpha'=f(\alpha)+\sum_{i=1}^{n_S} k_i(xyt-1)\epsilon_i,
    \qquad
    \alpha\in S,\quad k_i\in R_0,
\end{equation}
Where the second term on the right of the equation is just a combination of $\tilde g^{2q+1},\dots ,\tilde g^{2q+n_S}$ and then it lies in $\langle \tilde g^1,\dots ,\tilde g^{2q+n_S}\rangle$. Therefore,
\begin{equation}
    (x^{L_x}-1)\alpha',\ (y^{L_y}-1)\alpha'
    \in
    \langle \tilde g^i\rangle,
\end{equation}
so
\begin{equation}
    x^{L_x}-1,\ y^{L_y}-1 \in \Ann(R_0^{n_S}/\langle \tilde g^i\rangle).
\end{equation}
Moreover, multiplying any element of \(R_0^{n_S}\) by \(xyt-1\) sends it into \(\langle (xyt-1)\epsilon_i\rangle\), and hence
\begin{equation}
    xyt-1\in \Ann(R_0^{n_S}/\langle \tilde g^i\rangle).
\end{equation}
Thus,
\begin{equation}
    \langle x^{L_x}-1,\ y^{L_y}-1,\ xyt-1\rangle
    \subset
    \Ann(R_0^{n_S}/\langle \tilde g^i\rangle),
\end{equation}
and there is an induced quotient map
\begin{equation}
    R_0/\langle x^{L_x}-1,\ y^{L_y}-1,\ xyt-1\rangle
    \twoheadrightarrow
    R_0/\Ann(R_0^{n_S}/\langle \tilde g^i\rangle).
\end{equation}

The ring \(R_0/\langle x^{L_x}-1,\ y^{L_y}-1,\ xyt-1\rangle\) is finite since $\langle x^{L_x}-1,\ y^{L_y}-1,\ xyt-1\rangle$ restricts the degree of $x$ to be less than $L_x$, the degree of $y$ to be less than $L_y$ and the degree of $t$ to be less than the minimum degree among $x$ and $y$. Hence, \(R_0/\Ann(R_0^{n_S}/\langle \tilde g^i\rangle)\) is also finite.

Let
\begin{equation}
A:=R_0/\Ann(R_0^{n_S}/\langle \tilde g^i\rangle).
\end{equation}
As shown above, \(A\) is a finite ring. Now let \(\mathfrak p\subset A\) be a prime ideal. Then the quotient
\(
A/\mathfrak p
\)
is an integral domain by definition. Since \(A\) is finite, \(A/\mathfrak p\) is also finite. Hence \(A/\mathfrak p\) is a finite integral domain, and therefore a field.
It follows that every prime ideal \(\mathfrak p\subset A\) is maximal, since an ideal is maximal if and only if the corresponding quotient ring is a field. Recall that for a commutative ring \(A\), the Krull dimension \(\dim A\) is the supremum of the lengths \(n\) of chains of prime ideals
\begin{equation}
\mathfrak p_0\subsetneq \mathfrak p_1\subsetneq \cdots \subsetneq \mathfrak p_n.
\end{equation}
Since every prime ideal of \(A\) is maximal, no strict inclusion \(\mathfrak p\subsetneq \mathfrak q\) between prime ideals is possible. Therefore, no nontrivial chain of prime ideals exists, and we conclude that
\begin{equation}
\dim A=0.
\end{equation}
In other words,
\begin{equation}
\dim\bigl(R_0/\Ann(R_0^{n_S}/\langle \tilde g^i\rangle)\bigr)=0.
\end{equation}

One may verify that
\begin{align}
    &(R_0^h)^{n_S}/\langle \tilde g^i\rangle^h
    =
    \bigl(R_0^{n_S}/\langle \tilde g^i\rangle\bigr)^h,\\
    &\Ann\!\left(\bigl(R_0^{n_S}/\langle \tilde g^i\rangle\bigr)^h\right)
    =
    \Ann(R_0^{n_S}/\langle \tilde g^i\rangle)^h,
\end{align}
from the fact that homogenization is an injective homomorphism and dehomogenization is a surjective homomorphism.
Hence, 
\begin{equation}\label{eq: R^h_0/Ann((R^h_0)^nS/<tilde g^i>^h)}
\begin{aligned}
    \dim\!\left(R_0^h/\Ann\!\left((R_0^h)^{n_S}/\langle \tilde g^i\rangle^h\right)\right)
    &=
    \dim\!\left(R_0^h/\Ann\!\left(\bigl(R_0^{n_S}/\langle \tilde g^i\rangle\bigr)^h\right)\right)\\
    &=
    \dim\!\left(R_0/\Ann(R_0^{n_S}/\langle \tilde g^i\rangle)\right)+1.
\end{aligned}
\end{equation}
The last equality follows from Lemma 2.3(6) of Ref.~\cite{VARBARO2011}.

Notice that Eq.~\eqref{eq: R^h_0/Ann((R^h_0)^nS/<tilde g^i>^h)} gives the Krull dimension of $(R_0^h)^{n_S}/\langle \tilde g^i\rangle^h$, whereas our goal is to determine the Krull dimension of $(R_0^h)^{n_S}/\langle (\tilde g^i)^h\rangle$. In general, these two quotients are not the same. Nevertheless,
\begin{equation}
    \langle (\tilde g^i)^h\rangle \subset \langle \tilde g^i\rangle^h.
\end{equation}
By the definition of the annihilator, this inclusion implies
\begin{equation}
    \Ann\!\left((R_0^h)^{n_S}/\langle \tilde g^i\rangle^h\right)\subset \Ann\!\left((R_0^h)^{n_S}/\langle (\tilde g^i)^h\rangle\right).
\end{equation}
Therefore, there is a quotient map
\begin{equation}
    R_0^h/\Ann\!\left((R_0^h)^{n_S}/\langle \tilde g^i\rangle^h\right)
    \twoheadrightarrow
    R_0^h/\Ann\!\left((R_0^h)^{n_S}/\langle (\tilde g^i)^h\rangle\right).
\end{equation}
It follows that
\begin{equation}
    \dim\!\left(R_0^h/\Ann\!\left((R_0^h)^{n_S}/\langle (\tilde g^i)^h\rangle\right)\right)
    \le
    \dim\!\left(R_0^h/\Ann\!\left((R_0^h)^{n_S}/\langle \tilde g^i\rangle^h\right)\right)
    =1.
\end{equation}
Hence
\begin{equation}
    d_K\le 1.
\end{equation}

Since \(\langle \tilde g^i\rangle^h\) is a graded submodule, Theorem~37 of Ref.~\cite{Liang2022GroebnerModules} gives
\begin{equation}\label{eq: degree bound of G^l when d_k=0}
    \deg(G^l)\le DMn_V-n_V+1,\qquad \text{if }d_K=0,
\end{equation}
and
\begin{equation}\label{eq: degree bound of G^l when d_k=1}
    \deg(G^l)\le 2\left[\frac{1}{2}\left((DM)^{n_V-d_K}M+D\right)\right]^{2^{d_K-1}},
    \qquad \text{if }d_K=1.
\end{equation}
Here \(D\) is the maximum degree of the generators \((\tilde g^1)^h,\dots,(\tilde g^{2q+n_S})^h\), \(M\) is the rank of the ambient free module, and \(n_V\) is the number of variables in the polynomial ring. In our case,
\begin{align}
    &D=\max_i \deg ((\tilde g^i)^h)= \max_i \deg (\tilde g^i)= 2r,\\
    &M=n_S,\\
    &n_V=4.
\end{align}
The first line follows from Eq.~\eqref{eq: identities of tilde g^i}.

The topological order condition excludes local degeneracy. For a prime qudit dimension, this implies that the number of independent stabilizer generators per unit cell equals the number of qudits per unit cell, so \(n_S=q\). Substituting these facts into Eqs.~\eqref{eq: degree bound of G^l when d_k=0} and \eqref{eq: degree bound of G^l when d_k=1}, we obtain
\begin{equation}
\begin{aligned}
    &\deg(G^l)\le 2rq-3 \qquad \text{if }d_K=0,\\
    &\deg(G^l)\le 8r^3q^4+2r \qquad \text{if }d_K=1,
\end{aligned}
\end{equation}
and hence, in either case,
\begin{equation}
    \deg(G^l)\le 8r^3q^4+2r.
\end{equation}

\subsubsection{Degree bound on the coefficients and conclusion}

We now convert the Gr\"obner-basis degree bound into a support bound on the creating Pauli operator.
From Eq.~\eqref{deg of b' G} and the bound above, we obtain
\begin{equation}\label{eq: deg of b' 2}
    \deg(b'_{il})\le 8r^3q^4+2r.
\end{equation}
Lemma~9 of Ref.~\cite{Liang2022GroebnerModules} implies that
\begin{equation}
    \{(G^1)^{\backslash *},\dots,(G^{n_G})^{\backslash *}\}
\end{equation}
is a Gr\"obner basis of \(\langle \tilde g^i\rangle\) under graded lexicographic order. Define
\begin{equation}
    G'^l := (G^l)^{\backslash *},\qquad l=1,\dots,n_G.
\end{equation}
Dehomogenizing Eq.~\eqref{equation of b' G} gives
\begin{equation}\label{equation of b G}
    \sum_{i=1}^{n_S+2q} (b'_{il})^{\backslash *}\tilde g^i
    =
    G'^l,
    \qquad
    (b'_{il})^{\backslash *}\in R_0.
\end{equation}
Since \(\{G'^1,\dots,G'^{n_G}\}\) is a Gr\"obner basis of $\langle \tilde g_i\rangle$ under graded lexicographic order, the equation
\begin{equation}\label{equation of c alpha}
    \sum_{l=1}^{n_G} c_l G^l=\alpha_1,
    \qquad
    c_l\in R_0,
\end{equation}
has a solution, and moreover
\begin{equation}\label{eq: deg of c}
    \deg(c_l)\le \deg(\alpha_1)-\deg(G^l),
    \qquad
    \text{whenever } c_l\neq 0.
\end{equation}
Combining Eqs.~\eqref{equation of b G} and \eqref{equation of c alpha}, we obtain
\begin{equation}
    \sum_{i=1}^{n_S+2q}\sum_{l=1}^{n_G} c_l (b'_{il})^{\backslash *}\tilde g^i
    =
    \alpha_1.
\end{equation}
Hence a solution of Eq.~\eqref{eq: k' falpha} is given by
\begin{equation}
    k_i=\sum_{l=1}^{n_G} c_l (b'_{il})^{\backslash *},
    \qquad
    i=1,\dots,2q+n_S,
\end{equation}
and therefore
\begin{equation}\label{eq: degree bound of k'_i 1}
    \deg(k_i)    \le    \deg(c_l)+\deg\!\bigl((b'_{il})^{\backslash *}\bigr)    \le
    \deg(c_l)+\deg(b'_{il})    \le    \deg(\alpha_1)+8r^3q^4+2r    \le    2r'+8r^3q^4+2r.
\end{equation}
The third inequality follows from Eqs.~\eqref{eq: deg of b' 2} and \eqref{eq: deg of c}, and the last inequality follows from Eq.~\eqref{eq: identity of alpha_1}.

Applying \(f'\) then yields a solution of Eq.~\eqref{eq: alpha_1 k_i}. By the definition of \(f'\), it will only reduce the degree of $k_i$, so 
\begin{equation}
\begin{aligned}
        \supp\!\left(f'(k_i)\right)&\subset [-\deg(f'(k_i)),\deg(f'(k_i))]\times[-\deg(f'(k_i)),\deg(f'(k_i))] \\&\subset 
    [-\deg(k_i),\deg(k_i)]\times[-\deg(k_i),\deg(k_i)].
\end{aligned}
\end{equation}
Therefore, for \(i=1,\dots,n_S\), each \(f'(k_i)\) is supported inside a square of linear size
\begin{equation}
    2r'+8r^3q^4+2r
\end{equation}
centered at the origin chosen. Undoing the earlier rotation and translation, we conclude that the original Pauli operator creating \(\alpha_0\) is supported in the concentric square of linear size
\begin{equation}
    2r'+8r^3q^4+2r
\end{equation}
around the \(r'\times r'\) square supporting the trivial anyon. This completes the proof.

\subsection{Proof of Lemma~\ref{lemma: height limit}}\label{app: The proof of Lemma 4}

This subsection develops both a qualitative and a quantitative understanding of the generation of boundary gauge operators. By definition, a boundary gauge operator is a Pauli operator supported near the boundary that commutes with all bulk stabilizers, although different boundary gauge operators need not commute with one another. We first show, using the operator-algebraic structure of the boundary theory, that every boundary gauge operator can be generated by bulk stabilizers together with truncated stabilizers, possibly through an infinite product. This makes it natural to describe boundary gauge operators uniformly in the \(R\)-module language introduced earlier. We then turn to the geometric content of this representation. Using operations on \(R\)-modules such as slicing, the construction of anyon strings, and a Gr\"obner-basis reduction of stabilizer decompositions, we derive a bound on the height of the generators needed to build a boundary gauge operator. This height bound will later allow us to decompose boundary gauge operators in a quantitatively controlled way.

\subsubsection{Generation of boundary gauge operators}

Before proving Lemma~\ref{lemma: height limit}, we first establish a structural fact about boundary gauge operators, which is stated by a lemma.

\begin{lemma}\label{lemma:secondary boundary gauge}
Boundary gauge operators are generated by the double orthogonal (closure) of the primary boundary gauge operators, i.e., those obtained by truncating stabilizers in the infinite plane. Equivalently, any boundary gauge operator can be expressed as a (possibly infinite) product of primary boundary gauge operators.
\end{lemma}
Following Ref.~\cite{liang2024operator}, define
\begin{align}
    P_S&=w(S),\\
    P_{BS}&=P_S\cap \pi P,\\
    P_{GS}&=(P_{BS}^{\Omega}\cap \pi P)/P_{BS},\label{equation:PGS}\\
    P_{GS1}&=\pi P_S/P_{BS},\\
    P_{GS2}&=(P_{BS}^{\Omega}\cap \pi P)/\pi P_S
    =P_{GS}/P_{GS1}.
\end{align}

Here \(P_{BS}\) is the group of bulk stabilizers, \(P_{GS}\) is the full group of boundary gauge operators modulo bulk stabilizers, and \(P_{GS1}\) is the subgroup of primary boundary gauge operators generated by truncated bulk stabilizers. The quotient \(P_{GS2}\) describes the remaining part of \(P_{GS}\), namely the group of secondary boundary gauge operators.

We may also regard \(\pi P_S\) as the same subset of \(\hat P\); in that case we denote it by
\begin{equation}
    \pi P_S' \coloneqq \pi P_S \subset \hat P.
\end{equation}
Then
\begin{equation}
    (\pi P_S')^{\perp}=(\pi P_S)^{\Omega}=P_{BS}.
\end{equation}
The first equality is simply the same orthogonality condition written in two different ambient spaces: on the left, \(\pi P_S\) is viewed inside \(\hat P\), while on the right it is viewed inside \(P\). The second equality is Theorem~14 of Ref.~\cite{liang2024operator}.

Lemma 12 of Ref.~\cite{liang2024operator} states that
\begin{equation}
    (\pi P_S')^{\perp\perp}=\widehat{(\pi P_S')}.
\end{equation}
By definition,
\begin{equation}
    P_{BS}^{\Omega}=P_{BS}^{\perp}\cap P.
\end{equation}
Hence,
\begin{equation}
    P_{GS2}
    =(P_{BS}^{\perp}\cap \pi P)/\pi P_S
    =\left((\pi P_S')^{\perp\perp}\cap \pi P\right)/\pi P_S
    =\bigl(\widehat{(\pi P_S')}\cap \pi P\bigr)/\pi P_S.
\end{equation}
The meaning of this identity is simple. The closure \(\widehat{(\pi P_S')}\) consists of possibly infinite \(\mathbb Z_p\)-linear combinations of primary boundary gauge operators. Intersecting with \(\pi P\) then picks out those combinations that still define finitely supported Pauli operators near the boundary. Therefore every element of \(P_{GS2}\) has a representative of this form. Since every element of \(P_{GS1}\) is already represented by a primary boundary gauge operator, it follows that every boundary gauge operator can be represented as a possibly infinite product of primary boundary gauge operators. This proves the lemma.

\subsubsection{Controlling the height of boundary gauge operators}

The lemma above identifies the algebraic origin of boundary gauge operators but does not yet control their geometric representation. In particular, such a representation may involve infinitely many bulk stabilizers extending arbitrarily deep into the bulk. These infinite bulk contributions are not removed modulo \(P_{BS}\), since \(P_{BS}\) contains only finite products of bulk stabilizers. This is precisely the obstruction that must be resolved in order to obtain quantitative decompositions later.

Indeed, suppose
\begin{equation}
    b_G=\pi w(\theta_G)
\end{equation}
is a boundary gauge operator. If we try to analyze it by slicing its representative,
\begin{equation}
    \pi w(\pi^x_{x_1:x_2}\theta_G),
\end{equation}
the height of the slice may depend on the choice of representative and need not be uniformly bounded.

Lemma~\ref{lemma: height limit} resolves this issue. It shows that although a boundary gauge operator may admit representations with arbitrarily large height, there is always another representation built from bulk and truncated stabilizer generators whose height is uniformly bounded.

\begin{replemma}{lemma: height limit}
Consider a two-dimensional translation-invariant topological stabilizer code whose stabilizer generators have range $r$, with bulk occupying the upper half-plane $y \ge 0$. Then every boundary gauge operator supported in the region $0 \le y \le h$, where $h\ge r$, admits a decomposition as an (infinite) product of
\begin{enumerate}
    \item bulk stabilizer generators in the upper half-plane;
    \item stabilizer generators of the infinite plane, truncated to their support in $y\ge 0$.
\end{enumerate}
Moreover, every factor is supported in
\begin{equation}
y \le 8r^3q^4 + 5r + 3h~.
\end{equation}
\end{replemma}
Let \(b_G\) be a finitely supported boundary gauge operator with height at most \(h\). Then there exists \(\theta\in \hat S\) such that
\begin{equation}
    b_G=\pi w(\theta).
\end{equation}
We split \(\theta\) into the part already lying in the upper half-plane and the part extending below the boundary:
\begin{align}
    &\theta_{B}=\pi \theta,\qquad \theta_T=\theta-\pi \theta,\\
    &b_B=\pi w(\theta_{B}),\qquad b_T=\pi w(\theta_{T}),\\
    &b_G=b_B+b_T.
\end{align}
Here, these two terms correspond to the product of bulk stabilizers and the product of truncated stabilizers, respectively.

The term \(\theta_T\) lies entirely below the boundary, so
\begin{equation}
    y_{\max}(\theta_T)<0.
\end{equation}
By locality of generators, this implies every factor in $b_T$ is supported in
\begin{equation}
y\le r.
\end{equation}

Thus the height of factors in \(b_T\) is bounded. The nontrivial part is to control \(b_B\), which may still be represented by an infinite product of bulk stabilizers. However, the height of $b_B$ itself is bounded since $b_B$ is the difference between \(b_G\)  and $b_T$, where the previous term has height at most \(h\) and the latter stays within distance \(r\) of the boundary:
\begin{align}
    &y_{\max}( b_{B})\le \max\{y_{\max}( b_{G}),y_{\max}( b_{T})\}\le h.
\end{align}

\subsubsection{Slicing \(b_B\) and the construction of \(b'_{Bn}\)}

The goal of this subsection is to replace the infinite product \(b_B\) by an infinite collection of local stabilizers whose projections reconstruct it, while keeping uniform control of their height. We begin by slicing \(b_B\) into strip-supported pieces. Each slice fails to be a stabilizer only at its two endpoints, where it creates a pair of opposite anyons. We then attach auxiliary Pauli strings that move these endpoint anyons downward and fuse them to the vacuum far below the boundary. In this way, each slice is converted into a finitely supported stabilizer \(b'_{Bn}\). The extra decorations cancel telescopically after projection to the upper half-plane, so the original operator \(b_B\) can be recovered from the family \(\{b'_{Bn}\}\). The problem is therefore reduced to obtaining a uniform height bound on these finitely supported stabilizers, which we do using a Gr\"obner-basis decomposition.

We slice $b_B$ as:
\begin{equation}
    b_{B}=\sum _n \pi^x_{nL:(n+1)L}b_{B},
\end{equation}
where \(L\) is chosen to satisfy that
\begin{equation}
L>2r+l_{xG},
\end{equation}
where $l_{xG}$ is the width of $b_G$. 

The choice ensures that different endpoint regions are well separated. Since  \(\pi^x_{nL:(n+1)L}b_B\) is a projection of $b_B$ on $P$, they have uniformly bounded height:
\begin{align}
    &y_{\max}( \pi^x_{nL:(n+1)L}b_{B})\le y_{\max}(b_B)\le h. \label{ineq: pi b_B}
\end{align}
Then, we will study the anyons created by the slices. To avoid handling boundary anyons on a semi-infinite plane, from now on, \(\pi^x_{nL:(n+1)L}b_B\) is temporarily regarded as an operator in the full plane \(P\), rather than in the half-plane \(\pi P\). We will project back to \(\pi P\) at the end.

A bulk stabilizer generator \(x^ny^mg^\mu\) can be violated by \(\pi^x_{nL:(n+1)L}b_B\) only if it overlaps it. Therefore,
\begin{equation}\label{eq: maxy miny x^ny^mg^mu}
 y_{\max}( x^ny^mg^\mu)\le h,\qquad y_{\min}( x^ny^mg^\mu)\ge 0.
\end{equation}
The construction of $b_B$ and \(\pi^x_{nL:(n+1)L}b_B\)s is shown in Fig.~\ref{fig: bB_slice}.
\begin{figure}[htb]
    \centering
    \includegraphics[width=0.7\linewidth]{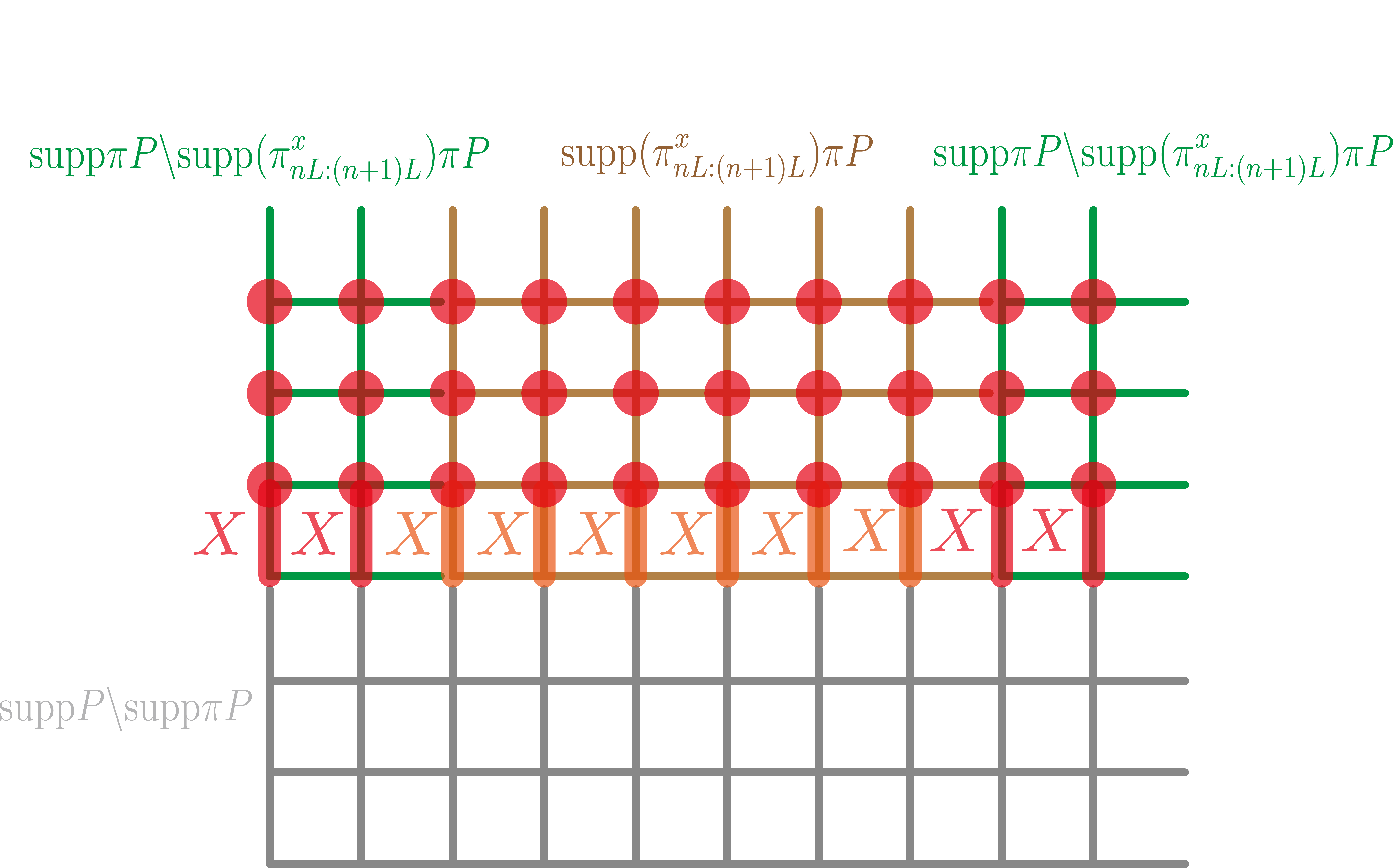}
    \caption{An example of slicing $b_B$, an infinite product of bulk stabilizer generators. The generators used here are those in shifted toric code defined in Fig.~\ref{fig: sTEE of toric code}. The lattice is divided into three regions: the gray region, the green region and the brown region, which correspond to $\supp P\backslash\supp \pi P$, $\supp \pi P\backslash \supp(\pi^x_{nL:(n+1)L}\pi P)$ and $\supp \pi^x_{nL:(n+1)L}\pi P$, respectively. Red dots indicate the position of $A_v$ terms in the lattice and the edge of $b_B$ is represented by red or orange edges. Red edges correspond to Paulis in $b_B-\pi^x_{nL:(n+1)L}b_{B}$, and orange edges correspond to Paulis in $\pi^x_{nL:(n+1)L}b_{B}$}.
    \label{fig: bB_slice}
\end{figure}

Moreover, such a violated generator must be truncated by the strip projection. Indeed, if
\begin{equation}
    \pi^x_{nL:(n+1)L}x^ny^mg^\mu=x^ny^mg^\mu,
\end{equation}
then
\begin{equation}\label{eq: reason why x^ny^m have to be truncated}
    x^ny^mg^\mu*\pi^x_{nL:(n+1)L}b_{B}
    =\pi^x_{nL:(n+1)L}x^ny^mg^\mu*b_{B}
    =x^ny^mg^\mu*b_{B}
    =0.
\end{equation}
The first equality uses Eq.~\eqref{eq: pi p_1*p_2}, and the last uses that both \(x^ny^mg^\mu\) and \(b_B\) are stabilizers.

Hence a violated stabilizer must cross one of the two strip boundaries. If it is truncated by \(\pi^x_{-\infty:(n+1)L}\), then
\begin{equation}\label{eq: maxx minx x^ny^mg^mu}
    x_{\max}(x^ny^mg^\mu)\ge (n+1)L,\qquad
    x_{\min}(x^ny^mg^\mu)< (n+1)L.
\end{equation}
Combining this with Eq.~\eqref{eq: maxy miny x^ny^mg^mu} and locality gives
\begin{equation}
        (n,m)\in \big[(n+1)L-l_x,(n+1)L\big)\times [-l_y,h].
\end{equation}
Combining it with a completely analogous argument for truncation by \(\pi^x_{nL:+\infty}\) yields
\begin{equation}\label{equation: locality of anyon}
(n,m)\in \Bigl(\bigl[nL-l_x,nL\bigr)\cup\bigl[(n+1)L-l_x,(n+1)L\big)\Bigr)\times [-l_y,h].
\end{equation}
Physically, this means that a strip can violate stabilizers only near its two endpoints. Therefore \(\pi^x_{nL:(n+1)L}b_B\) creates an anyon \(\alpha_n\) near its left endpoint and an anyon \(\alpha'_{n+1}\) near its right endpoint, as shown in Fig.~\ref{fig: bB anyons}. The size of each endpoint syndrome is bounded by
\begin{equation}
    r'\le \max \bigl\{l_x,l_y+h\bigr\}\le r+h,
\end{equation}
and their heights satisfy
\begin{equation}
    y_{\max}(\alpha_n),\ y_{\max}(\alpha'_n)\le h.
\end{equation}
Here, the inequality $l_y\le r\le h$ has already been applied. Because \(\pi^x_{nL:(n+1)L}b_B\) is itself a finite Pauli string, the two endpoint anyons must have opposite topological charge. More precisely, one finds
\begin{equation}
    \alpha_n\sim-\alpha'_n.
\end{equation}
Indeed, a stabilizer violated at the left endpoint can overlap only the two neighboring strips \(\pi^x_{(n-1)L:nL}b_B\) and \(\pi^x_{nL:(n+1)L}b_B\), and the contributions from these two strips cancel, as
\begin{equation}
\alpha_n=-\alpha'_{n-1}
\end{equation}
\begin{figure}[htb]
    \centering
    \includegraphics[width=0.7\linewidth]{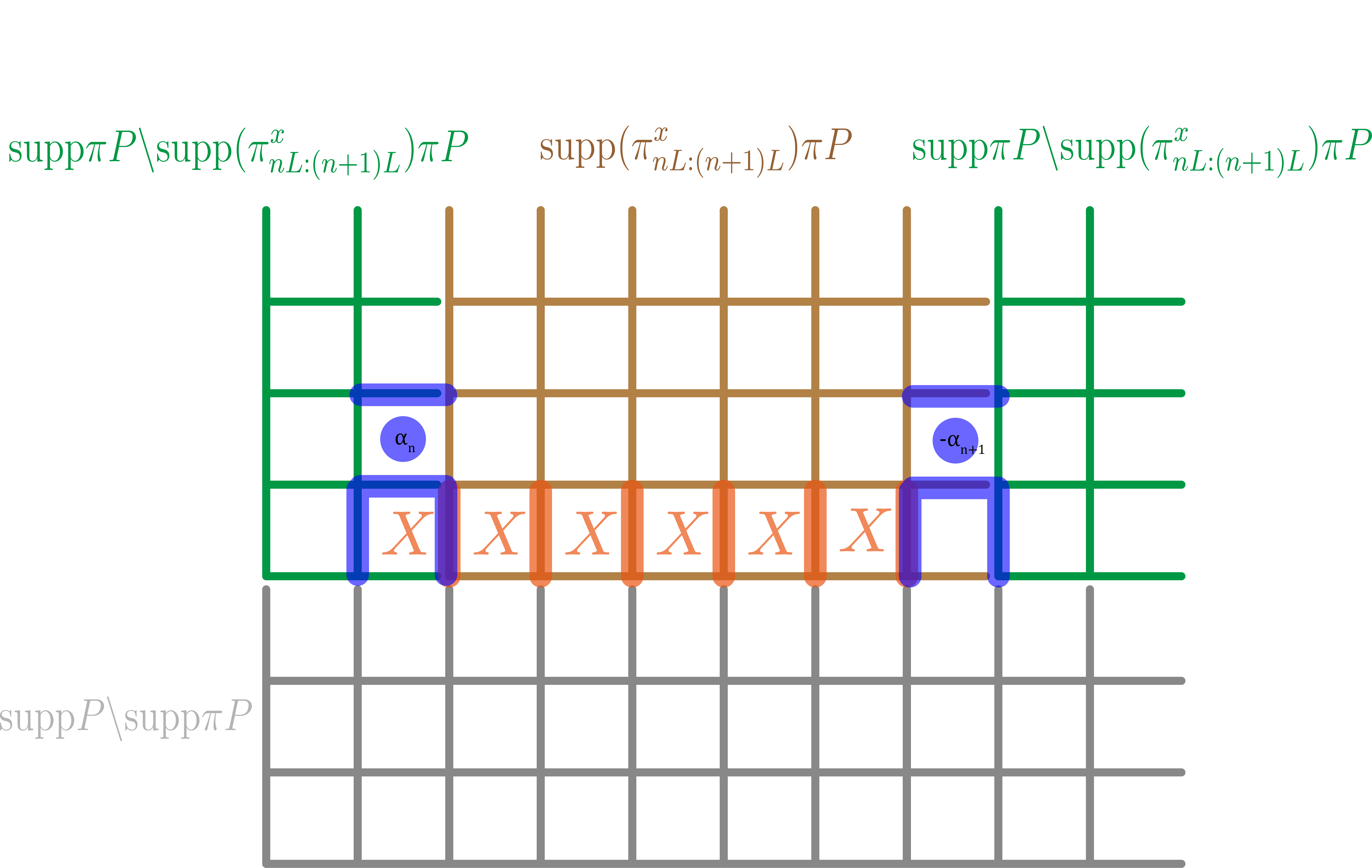}
    \caption{$\pi^x_{nL:(n+1)L}b_{B}$ creates an anyon on each of its endpoints, $\alpha_n$ and $\alpha_{n+1}'=-\alpha_{n+1}$. The $B_p$ terms it violates, as well as the anyons it creates, are represented by blue elements.}
    \label{fig: bB anyons}
\end{figure}

The anyons \(\alpha_n\) have uniformly bounded support. We introduce three auxiliary Pauli operators connecting neighboring endpoint anyons. First, let \(t_n\) be a Pauli string of width at most \(r\) that moves an anyon \(\beta_n\), of the same type as \(\alpha_n\) and supported in the same \(r'\times r'\) square, downward by \(l_{\mathrm{down}}\) lattice spacings to an anyon \(\beta'_n\). Next, let \(t''_n\) be a Pauli operator that converts \(\alpha_n\) into \(\beta_n\). Finally, let \(t'_n\) be a Pauli operator that fuses \(\beta'_n\) and \(-\beta'_{n+1}\) to vacuum. 

Here, the construction of $t_n$ is given in the argument immediately following Lemma~\ref{lemma: 2N_a logical}, and the constructions of $t_n'$ and $t''_n$ are given in the proof of Lemma~\ref{lemma: string size bound}.  
The syndromes they create can be expressed as 
\begin{align}
    &\varepsilon(\pi^x_{nL:(n+1)L}b_{B})=\alpha_n-\alpha_{n+1},\\
    &\varepsilon(t_n)=-\beta_n+\beta'_n,\\
     &\varepsilon(t'_n)=-\beta'_n+\beta'_{n+1},\\
   &\varepsilon(t''_n)=-\alpha_n+\beta_n.
\end{align}
The construction of Pauli strings is shown in Fig.~\ref{fig: bB strings}.
\begin{figure}[htb]
    \centering
    \includegraphics[width=0.7\linewidth]{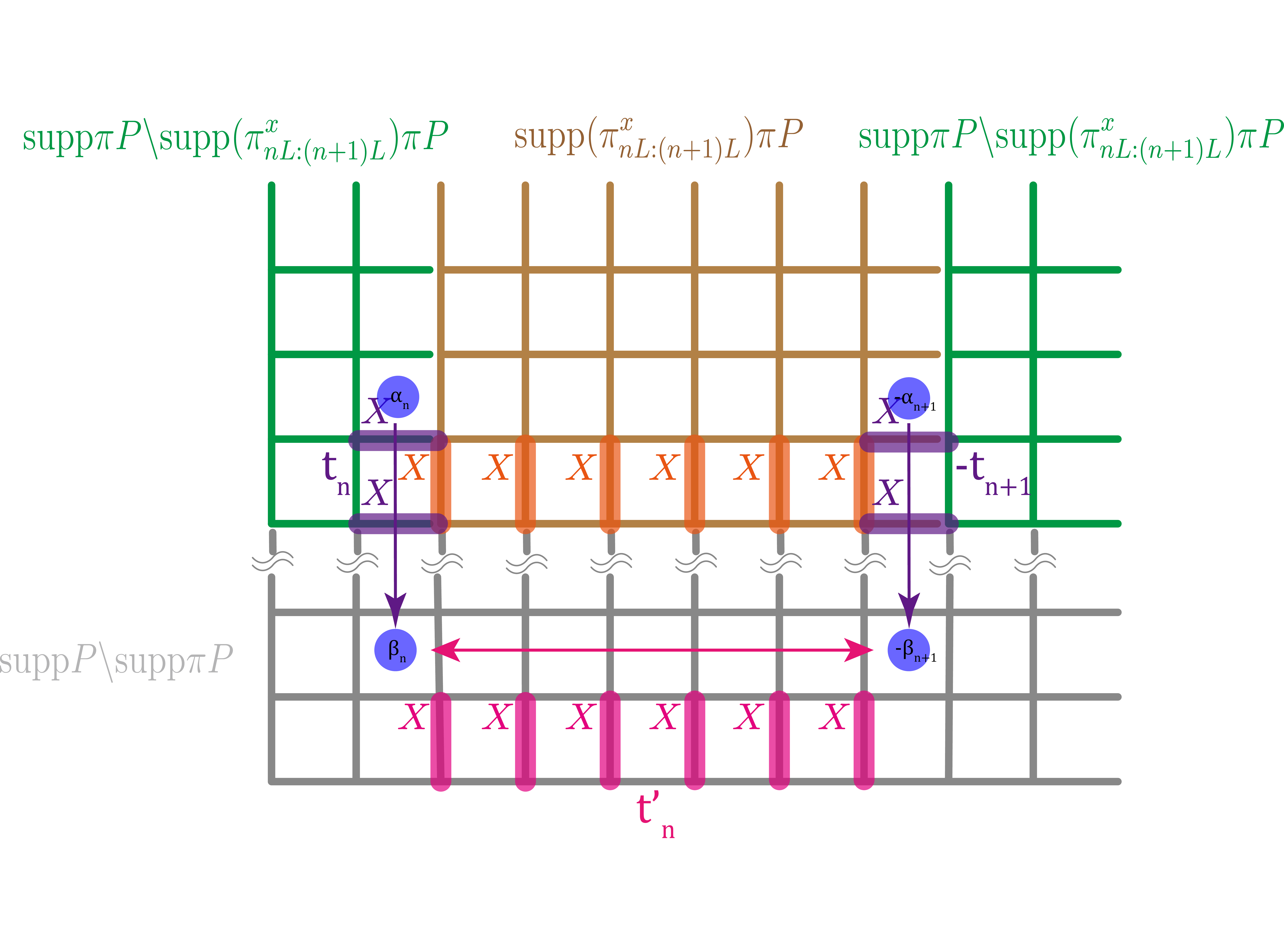}
    \caption{Using three anyon string operators: $t_n$, $-t_{n+1}$ and $t_n'$ to move $\alpha_n$ and $-\alpha_{n+1}$ downward by $l_{\text{down}}$ and then fuse into a trivial anyon. $t_n$ and $-t_{n+1}$ correspond to the purple elements and $t'_n$ correspond to the magenta elements. Thick edges on lattice represent the Paulis of anyon string operator, and arrows represent the ``path" along which anyons are moved by anyon strings. To show the ``$l_{\text{down}} $ sites down" where $\beta_n$ and $-\beta_{n+1}$ locate tightly in the figure, some sites in $P/\pi P$ are omitted, as represented by wavy lines. Notice that $t''_n=t''_{n+1}=0$ under this case.}
    \label{fig: bB strings}
\end{figure}

We choose
\begin{equation}
l_{\text{down}}\gg 8r^3q^4+2r+2d_{t''_n},
\end{equation}
where $d_{t''_n}$ is the size of the anyons created by $t''_n$, so that the \(\beta_n\) are pushed far below the boundary:
\begin{equation}
    y_{\max}(\beta_n)\le -l_{\text{down}}.
\end{equation}

Now define
\begin{equation}
    b'_{Bn}:=\pi^x_{nL:(n+1)L}b_{B}+t_n-t_{n+1}+t''_n-t''_{n+1}+t'_n.
\end{equation}
Its syndrome vanishes:
\begin{equation}
\begin{aligned}
    \varepsilon(b'_{Bn})
    =&~\varepsilon(\pi^x_{nL:(n+1)L}b_{B})+\varepsilon(t_n)-\varepsilon(t_{n+1})+\varepsilon(t''_n)-\varepsilon(t''_{n+1})+\varepsilon(t'_n)\\
    =&~ (\alpha_n-\alpha_{n+1})+(-\alpha_n+\beta_n)-(-\alpha_{n+1}+\beta_{n+1})
    +(-\beta_n+\beta'_n)-(-\beta_{n+1}+\beta'_{n+1})+(-\beta'_n+\beta'_{n+1})\\
    =&~0,
\end{aligned}
\end{equation}
which means that \(b'_{Bn}\) is a stabilizer in the full plane.

\subsubsection{The height limit of stabilizer generators in \(b'_{Bn}\)}
\label{subsubsec: The height limit of stabilizer generators in b'_Bn}
Lemma~\ref{lemma: string size bound} then gives bounds on the sizes of $t'_n$s and $t''_n$s, and which induces a height bound of them:
\begin{equation}\label{eq: y_max t'_n t''_n}
    \begin{aligned}
    &y_{\max}(t'_n)\le 8r^3q^4+2r+2r'+y_{\max}(\alpha_n)\le 8r^3q^4+4r+3h\\
    &y_{\max}(t''_n)\le 8r^3q^4+2r+2d_{t''_n}+y_{\max}(\beta_n)<0,
\end{aligned}
\end{equation}
while the construction of $t_n$ guarantees that its height will not exceed $h+r$.

These inequalities imply that the height of \(b'_{Bn}\) is bounded by
\begin{equation}
\begin{aligned}
       y_{\max}(b'_{Bn})
       &\le \max\{y_{\max}(t_n),y_{\max}(t_{n+1}),y_{\max}(t'_n),y_{\max}(t'_{n+1}),y_{\max}(t''_{n+1}),y_{\max}(\pi^x_{nL:(n+1)L}b_B)\}\\
       &\le 8r^3q^4+4r+3h.
\end{aligned}
\end{equation}

Because \(b'_{Bn}\) has finite support, the topological order condition implies that it can be expressed as a finite product of stabilizer generators. To obtain a controlled height bound on this decomposition from the height bound of $b'_{Bn}$, we now apply a Gr\"obner basis reduction technique.

Let \(g'^1,\dots,g'^{n'_S}\) be the Gr\"obner basis of \(g^1,\dots,g^{n_S}\) with respect to the anti-lexicographic TOP order
\begin{equation}\label{def: Ord'}
    x^{n_1}y^{m_1}\prec x^{n_2}y^{m_2}
    \quad\text{if}\quad
    \begin{cases}
         m_1<m_2,\\
         \text{or}\\
         m_1=m_2,\ n_1<n_2.
    \end{cases}
\end{equation}

The corresponding module \(S'\) and maps \(w'\), \(\varepsilon'\) are defined exactly as before:
\begin{align}
    &w':\hat S'\rightarrow \hat P \text{ or }S'\rightarrow P,\\
    &\varepsilon':P\rightarrow S'.
\end{align}

We also define
\begin{align}
&l'^\mu_{xi}:=x_{\max}(g'^\mu_i),\qquad l'^\mu_{yi}:=y_{\max}(g'^\mu_i),\\
&l'^\mu_x:=\max_i\{l'^\mu_{xi}\},\qquad l'^\mu_y:=\max_i\{l'^\mu_{yi}\},\\
&l'_x:=\max_\mu\{l'^\mu_x\},\qquad l'_y:=\max_\mu\{l'^\mu_y\}.
\end{align}

The key point is that passing from \(g^\nu\) to the Gr\"obner basis \(g'^\mu\) does not increase the height of the generators beyond the original local bound. The Gr\"obner basis is produced by Buchberger's algorithm through repeated formation and reduction of \(S\)-vectors. For two elements \(u\) and \(v\),
\begin{equation}
   S(u,v):=\frac{m}{\LT(u)}\,u-\frac{m}{\LT(v)}\,v,
\end{equation}
where \(\LM(u)=x^{n_1}y^{m_1} e_i\), \(\LM(v)=x^{n_2}y^{m_2} e_i\), and \(m=\lcm(x^{n_1}y^{m_1},x^{n_2}y^{m_2})\).

Using
\begin{equation}
y_{\max}(\lcm(m_1,m_2))=\max\{y_{\max}(m_1),y_{\max}(m_2)\},
\end{equation}
and the fact that the leading term of an element is one of the terms with maximal height, we see that each summand
\begin{equation}
\frac{m}{\LT(u)}\,u
\end{equation}
appearing in Buchberger's algorithm has height no larger than the largest height among the current generators. Under the chosen anti-lexicographic order, the reduction step does not increase this height either. Therefore, in the decomposition
\begin{equation}\label{equation: decomposition of g'}
    g'^\mu=\sum_{\nu=1}^{n_S} c_{\mu\nu} g^\nu ,
\end{equation}
every nonzero term satisfies
\begin{equation}\label{inequality: c limit}
    y_{\max}(c_{\mu\nu} g^\nu)\le l_y\le r.
\end{equation}
Now reduce \(b'_{Bn}\) by the Gr\"obner basis \(\{g'^\mu\}\). Since \(b'_{Bn}\) has finite support, after a translation by \(x^{n'}y^{m'}\) it becomes an element of \(\mathbb Z_p[x,y]^{2q}\), and Gr\"obner reduction gives
\begin{equation}
    x^{n'}y^{m'}b'_{Bn}=\sum_{\mu=1}^{n'_S} k_\mu g'^\mu,\qquad k_\mu\in \mathbb{Z}_p[x,y].
\end{equation}
The reduction process guarantees
\begin{equation}
    \lm(k_\mu g'^\mu)\preceq \lm(x^{n'}y^{m'}b'_{Bn}),\qquad \forall \mu,
\end{equation}
so after shifting back,
\begin{equation}\label{eq: b'_Bn}
\begin{aligned}
    b'_{Bn}&=\sum_{\mu=1}^{n'_S} \tilde k_\mu g'^\mu,\qquad
    \tilde k_\mu=x^{-n'}y^{-m'}k_\mu\in R,\\
    y_{\max}(\tilde k_\mu g'^\mu)&\le y_{\max}(b'_{Bn}),\qquad \forall \mu.
\end{aligned}
\end{equation}
Combining this with the height bound on \(b'_{Bn}\), we get
\begin{equation}
    y_{\max}(\tilde k_\mu g'^\mu)\le 8r^3q^4+4r+3h,\qquad \forall \mu.
\end{equation}
Since every nonzero component \(g'^\mu_i\) satisfies \(y_{\min}(g'^\mu_i)\ge 0\), multiplication by \(g'^\mu\) cannot decrease the maximal \(y\)-coordinate. Hence
\begin{equation}\label{inequality: k bar limit}
     y_{\max}(\tilde k_\mu)\le y_{\max}(\tilde k_\mu g'^\mu)\le 8r^3q^4+4r+3h.
\end{equation}
Finally, we rewrite the decomposition back in terms of the original generators \(g^1,\dots,g^{n_S}\):
\begin{align}
    &b'_{Bn}=\sum_{\mu=1}^{n'_S} \tilde k_\mu g'^\mu
    =\sum_{\nu=1}^{n_S}\sum_{\mu=1}^{n'_S}\tilde k_\mu c^\mu_\nu g^\nu,\\
    &y_{\max}(\tilde k_\mu c^\mu_\nu g^\nu)
    \le y_{\max}(\tilde k_\mu )+y_{\max}(c^\mu_\nu g^\nu)
    \le 8r^3q^4+5r+3h.
\end{align}
The first term is bounded by Eq.~\eqref{inequality: k bar limit}, and the second by Eq.~\eqref{inequality: c limit}. Thus, $b'_{Bn}$ can be decomposed as stabilizer generators in infinite plane with height at most 
\begin{equation}8r^3q^4+5r+3h.\end{equation}

\subsubsection{Conclusion}

The family \(\{b'_{Bn}\}\) is locally finite, so the sum \(\sum_n b'_{Bn}\) is well defined in \(\hat P\). Using
\begin{equation}
b'_{Bn}=\pi^x_{nL:(n+1)L}b_B+t_n-t_{n+1}+t'_n-t_{n+1}'+t''_{n},
\end{equation}
we obtain
\begin{equation}
\begin{aligned}
\pi\sum_n b'_{Bn}
&=\pi\sum_n \pi^x_{nL:(n+1)L}b_B
  +\pi\sum_n (t_n-t_{n+1}+t'_n-t_{n+1}')
  +\pi\sum_n t''_n .
\end{aligned}
\end{equation}
The first term equals \(b_B\), since the strip projections decompose \(b_B\). The second term cancels telescopically after projection, because \(\pi t_n\) and \(-\pi t_n\) appear with opposite signs.  The third term vanishes, because the second line of Eq.~\eqref{eq: y_max t'_n t''_n}  shows that \(t''_n\) lies entirely below the boundary and therefore disappears after projection:
\begin{equation}\label{equation: pi t}
    \pi t''_n=0.
\end{equation}
Therefore
\begin{equation}
\pi\sum_n b'_{Bn}=b_B.
\end{equation}
Each \(b'_{Bn}\) is a product of bulk stabilizer generators in the full plane with height at most
\begin{equation}
8r^3q^4+5r+3h.
\end{equation}
After projection, this gives a decomposition of \(b_B\) into bulk and truncated stabilizer generators with the same height bound. Combining this with the already controlled term \(b_T\), we conclude that \(b_G\) itself can be generated by bulk and truncated stabilizer generators whose supports lie within
\begin{equation}
y\le 8r^3q^4+5r+3h,
\end{equation}
which proves Lemma~\ref{lemma: height limit}.

\subsection{Proof of Theorem~\ref{thm: improved Levin-wen}}\label{app: proof of theorem 1}
With the preparations above, we can now finally prove Theorem~\ref{thm: improved Levin-wen}.
\begin{reptheorem}{thm: improved Levin-wen}
Consider a translation-invariant topological stabilizer code on a square lattice of $\mathbb{Z}_p$ (with prime $p$) qudits with $q$ qudits per site.
Assume that the stabilizer generators have range $r$, i.e., each generator is supported within an $r\times r$ square.
Then, for the concave partition shown in Fig.~\ref{fig: Levin-Wen(b)}, if the linear size $L$ satisfies
\begin{equation}
        L \ge 32r^3q^4+2r^4+27r + 1~,
\end{equation}
the 
conditional mutual information $\gamma$ in Eq.~\eqref{eq: TEE}
is guaranteed to be free of spurious contributions. 
\end{reptheorem}

The proof strategy is to show that every stabilizer contributing to the spurious topological entanglement entropy (sTEE) is in fact divisible, and therefore cannot satisfy the indivisibility condition. The key point is that the potentially nontrivial part of a global stabilizer decomposition comes from generators that are not fully supported in \(A\cup B\cup C\): although their contributions outside \(A\cup B\cup C\) cancel in the full product, an arbitrary partial decomposition need not remain supported inside \(A\cup B\cup C\). Lemma~\ref{lemma: sTEE contribution} imposes a strong restriction on such generators for an sTEE contributor: no generator can be fully supported inside \(D_{\mathrm{in}}\). As a result, the generators entering the decomposition can only live near the boundary between \(D\) and \(A\cup B\cup C\), penetrating into \(D\) only to a finite depth. By the topological order condition, the same is true for the part extending into \(E\): it can penetrate only a finite distance. We first decompose the stabilizer into three parts: the contribution in the bulk of $A\cup B\cup C$, the contributions near \(D\) and the contributions near \(E\) and then treat them separately. 

The concave partition shown in Fig.~\ref{fig: Levin-Wen(b)} is specified by the following geometric regions as subsets of $\mathbb{Z}_2$:
\begin{equation}
\begin{aligned}
A&:=\left\{(n,m)\in\Lambda \,\middle|\, 0\le n<L,\ 2L\le m<3L \right\},\\
B&:=\left\{(n,m)\in\Lambda \,\middle|\,
\bigl(0\le m<L \text{ or } 4L\le m<5L\bigr)\right\}\\
&\qquad\cup
\left\{(n,m)\in\Lambda \,\middle|\,
\bigl(0\le n<L \text{ or } 2L\le n<3L\bigr)
\text{ and }
\bigl(L\le m<2L \text{ or } 3L\le m<4L\bigr)
\right\},\\
C&:=\left\{(n,m)\in\Lambda \,\middle|\, 2L\le n<3L,\ 2L\le m<3L \right\},\\
D&:=\left\{(n,m)\in\Lambda \,\middle|\, L\le n<2L,\ L\le m<4L \right\},\\
E&:=\mathbb{Z}_2\setminus (A\cup B\cup C\cup D).
\end{aligned}
\end{equation}
When 
\begin{equation}\label{eq: r<L}
r<L,
\end{equation}
a generator cannot have support in $D$ and $E$ simultaneously.
Then, for any stabilizer
\begin{equation}
s_S=w(\theta_{S}),\qquad \theta_{S}\in S,
\end{equation}
satisfying
\begin{equation}
\supp s_S\subset A\cup B\cup C,
\end{equation}
each nonzero monomial term $k_i x^{n_i}y^{m_i} e_\mu$ with $\mu =1,\dots,n_S$ appearing in $\theta_{S}$, which corresponds to a single stabilizer generator factor in $s_S$, falls into exactly one of the following three classes:
\begin{enumerate}
    \item $\supp\ w(k_i x^{n_i}y^{m_i}e_j)\subset A\cup B\cup C$;
    \item $\supp\ w(k_i x^{n_i}y^{m_i}e_j)\not\subset A\cup B\cup C$, but $\supp\ w(k_i x^{n_i}y^{m_i}e_j)\subset A\cup B\cup C\cup D$;
    \item $\supp\ w(k_i x^{n_i}y^{m_i}e_j)\not\subset A\cup B\cup C$, but $\supp\ w(k_i x^{n_i}y^{m_i}e_j)\subset A\cup B\cup C\cup E$.
\end{enumerate}
We denote by $\theta_{Si}$ the sum of all monomial terms in the $i$th class, and define
\begin{equation}
s_{Si}:=w(\theta_{Si}),\qquad i=1,2,3.
\end{equation}
By construction, only $s_{S2}$ can have support intersecting $D$, but this support cannot be canceled by $s_{S1}$ or $s_{S3}$. Hence $s_{S2}$ cannot have supports in $D$. By the same argument, $s_{S3}$ cannot have support in $E$. Hence,
\begin{equation}
\supp(s_{S1}),\ \supp(s_{S2}),\ \supp(s_{S3})\subset A\cup B\cup C.
\end{equation}
If $s_{S}$ is a sTEE contributor. The fact that it satisfies the condition in Lemma~\ref{lemma: sTEE contribution} induces that $s_{S1}$, $s_{S2}$ and $s_{S3}$ satisfy the condition in Lemma~\ref{lemma: sTEE contribution} respectively. It is therefore sufficient to show that each $s_{Si}$ is divisible to prove that such an sTEE contributor cannot exist.

\subsubsection{Determining the separation between \(D_{\mathrm{in}}\) and \(D\)}

We now determine how far the inner region \(D_{\mathrm{in}}\) must be separated from \(D\). Physically, what we need is a buffer large enough so that any trivial anyon localized inside \(D_{\mathrm{in}}\) can be created by a Pauli operator whose support stays entirely inside \(D\). This is precisely the input needed later in Lemma~\ref{lemma: sTEE contribution}. We therefore introduce an undetermined buffer length \(L_S\) and define
\begin{equation}\label{eq: def of D_in}
\begin{aligned}
D_{\mathrm{in}}
:=\{(n,m)\in\Lambda \mid\;&L+L_S\le n<2L-L_S,\\
&L+L_S\le m<4L-L_S\}.
\end{aligned}
\end{equation}
The goal of this subsection is to find a sufficient value of \(L_S\).

The first step is to show that, for every anyon type, one can choose narrow string operators in both horizontal and vertical directions. This follows from the structure of logical operators on a large torus.
Suppose the horizontal and vertical anyon periods of the code are \(L_x\) and \(L_y\), respectively. On an \(nL_x\times mL_y\) torus, we will use the following lemma.

\begin{lemma}\label{lemma: 2N_a logical}
    Let \(L_x,L_y\) be multiples of horizontal and vertical periods of all \(N_a\) anyons in a topological stabilizer code.
    Defined on an \(L_x\times L_y\) torus, the code has \(2N_a\) logical qubits.
    Furthermore, for sufficiently large \(L_x,L-y\), there exists a constant \(w\) (independent of \(m,n\)) such that the logical Pauli operators can be written in a basis where \(l_{X,1},\dots,l_{X,N_a},l_{Z,N_a+1},\dots,l_{Z,2N_a}\) are supported in a vertical band of width \(w\), and \(l_{Z,1},\dots,l_{Z,N_a},l_{X,N_a+1},\dots,l_{X,2N_a}\) are supported in a horizontal band of width \(w\).
\end{lemma}

\begin{proof}
    For the counting of logical qubits, using the algebraic formalism of Ref.~\cite{haah_module_13}, the logical space is
    \begin{equation}
    \dim_{\mathbb{F}_2}\operatorname{Tor}_1^R(\mathbb{F}_2^{N_a}, R/(x^m-1, y^n -1)) =2N_a,
    \end{equation}
    where \(R\) is the Laurent polynomial ring in two variables.

    For the second statement, the code can be mapped by a constant-depth circuit to \(N_a\) copies of the toric code \cite{haah_classification_21}. Pulling back the standard toric code logical operators then gives the claimed strip-supported basis.
\end{proof}

This lemma implies that each anyon type admits string operators of bounded width. To make this precise, let \(T_1\) be a vertical strip of width \(r\), and let \(T_2\) be its complement in the torus \(T\). Denote by \(N_{l1}\) and \(N_{l2}\) the numbers of independent logical operators supported in \(T_1\) and \(T_2\), respectively.

Consider a logical operator \(l_o\in P\) fully supported in \(T_2\). Since the stabilizer generators have range \(r\), no generator can simultaneously overlap both boundaries of \(T_1\). Therefore, if we cut the torus open along those two boundaries, \(l_o\) still commutes with all stabilizer generators in the resulting flat ring. If \(l_o\) is horizontal, then it does not wind around the vertical cycle since it must be supported in a region with width $w$ when $L_y>w$, so we may also open the upper and lower boundaries. In this way, \(l_o\) becomes a finitely supported operator on the infinite plane that still commutes with all stabilizers; see Fig.~\ref{fig:expanding torus}.

\begin{figure}
    \centering
    \includegraphics[width=0.8\linewidth]{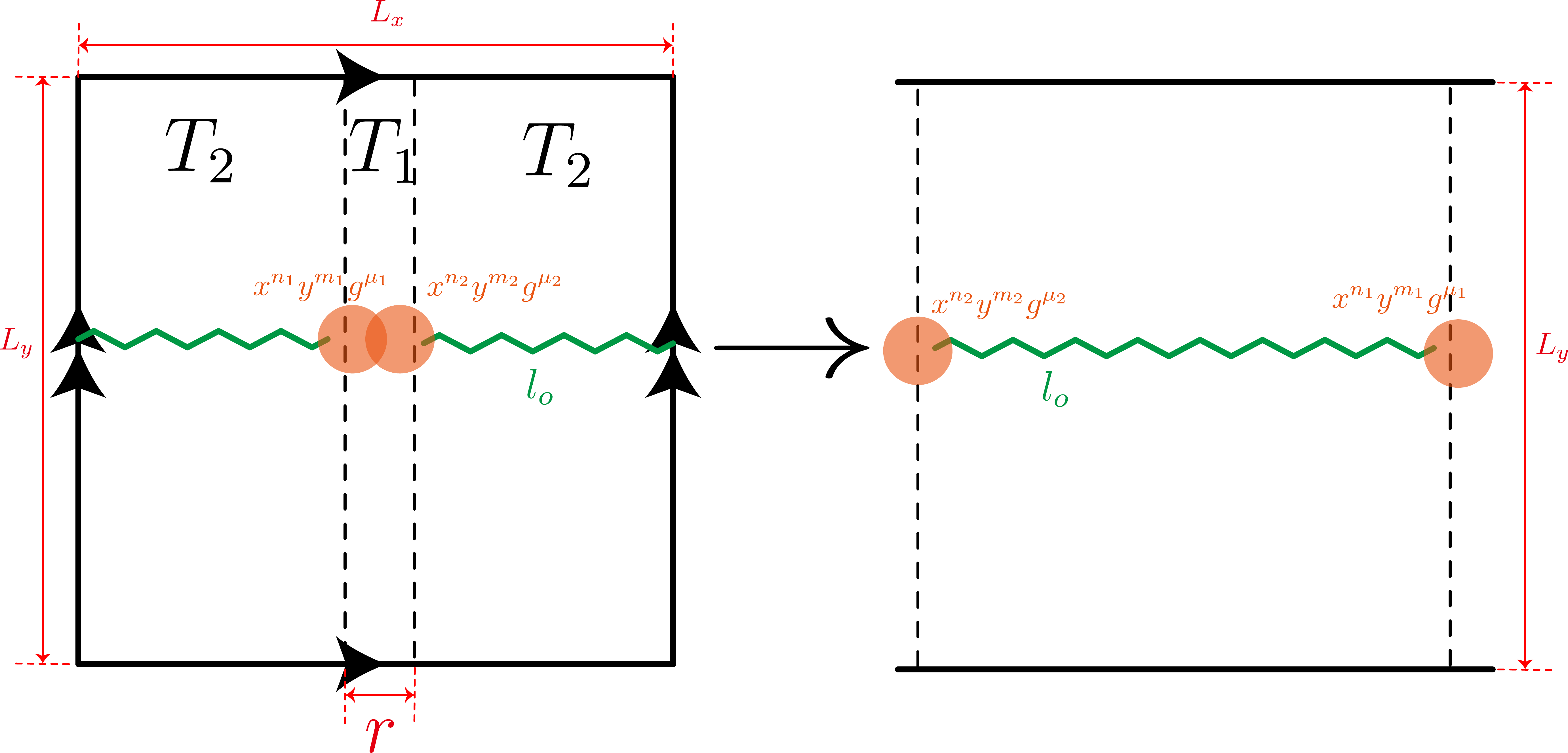}
    \caption{In a $L_x\times L_y$ torus, $l_o$ is a horizontal logical operator fully supported in $T_2$. When the width of the vertical stripe $T_1$ is $r$, a single stabilizer generator satisfies at most one of the following conditions: 1. overlaps the left boundary of $T_1$ as $x^{n_1}y^{m_1}g^{\mu_1}$; 2. overlaps the right boundary of $T_1$ as $x^{n_2}y^{m_2}g^{\mu_2}$. Therefore, cutting open the boundary of $T_1$ will not change the commutation relationship between these stabilizer generators and $l_o$. Furthermore, $l_o$ does not wind around a vertical cycle, so it can be placed in an infinite plane while still commuting with all stabilizers. Hence, it is a stabilizer.}
    \label{fig:expanding torus}
\end{figure}

By our topological condition, such an operator must be a finite product of stabilizer generators, by the same argument as in Sec.~\ref{subsubsec: The height limit of stabilizer generators in b'_Bn}. Hence no nontrivial horizontal logical operator can be fully supported in \(T_2\), and therefore
\begin{equation}
    N_{l2}\le N_a.
\end{equation}
If \(2r<L_x\), the same argument applied to \(T_1\) shows that no nontrivial vertical logical operator can be fully supported in \(T_1\), so
\begin{equation}
    N_{l1}\le N_a.
\end{equation}
On the other hand, by the cleaning lemma, the logical operators supported in \(T_1\) and \(T_2\) together account for all \(2N_a\) independent logical operators on the torus. Hence
\begin{equation}
    N_{l1}+N_{l2}=2N_a,
\end{equation}
and therefore
\begin{equation}
    N_{l1}=N_{l2}=N_a.
\end{equation}
So there are exactly \(N_a\) independent vertical logical operators that can be supported in \(T_1\). Equivalently, for each anyon type \(\varsigma\), one may choose a vertical string operator \(t_v^{(\varsigma)}\) of width at most \(r\). By the same argument, one may also choose a horizontal string operator \(t_h^{(\varsigma)}\) of width at most \(r\).

Any finite truncation of \(t^{(\varsigma)}_v\) or \(t^{(\varsigma)}_h\) creates an anyon of type \(\varsigma\) at one endpoint and the opposite anyon at the other endpoint. Moreover, the width bound implies
\begin{equation}\label{eq: width bound of string operators}
\begin{aligned}
    &\supp ( t^{(\varsigma)}_h)\subset \mathbb{Z}\times [0,r),\\
    &\supp ( t^{(\varsigma)}_v)\subset  [0,r)\times\mathbb{Z}.
\end{aligned}
\end{equation}
We next estimate the size of the endpoint anyon created by truncating such a string. Consider the left truncation of a horizontal string,
\begin{equation}
\pi ^x_{-\infty:x_1}t^{(\varsigma)}_h.
\end{equation}
A translated stabilizer generator \(x^ny^m g^\mu\) can be violated by this operator only if two conditions are satisfied simultaneously:
\begin{enumerate}
    \item it is cut by the projection \(\pi^x_{-\infty:x_1}\), namely,
    \begin{equation}
    x_{\min}(x^ny^mg^\mu)<x_1\le x_{\max}(x^ny^mg^\mu);
    \end{equation}
otherwise, the identity in Eq.~\eqref{eq: pi p_1*p_2}, combining with the fact that $t^{(\varsigma)}_h$ itself commutes with all stabilizers gives that:
\begin{equation}
x^ny^mg^\mu*\pi^x_{-\infty:x_1}t_h^{(\varsigma)}=\pi^x_{-\infty:x_1}x^ny^mg^\mu*t_h^{(\varsigma)}=x^ny^mg^\mu*t_h^{(\varsigma)}=0.
\end{equation}
    \item it overlaps the support of \(t^{(\varsigma)}_h\), which by Eq.~\eqref{eq: width bound of string operators} requires
    \begin{equation}
    y_{\min}(x^ny^mg^\mu)<r,\qquad y_{\max}(x^ny^mg^\mu)\ge 0.
    \end{equation}
\end{enumerate}
Using locality of the generators, these conditions imply
\begin{equation}\label{eq: bound on n and m near left endpoint}
(n,m)\in [x_1-r,x_1)\times[-r,r).
\end{equation}
Thus, every violated generator lies inside a \(r\times 2r\) box around the endpoint \(x=x_1\). In later proofs, using a larger box with size $2r\times 2r$ is sufficient. Therefore, to treat the anyons created by vertical and horizontal semi-infinite strings in the same way, we will use the $2r\times 2r$ box, rather than $r\times 2r$. In other words, the endpoint anyon \(\varsigma^{(x_1)}_1\) created by \(\pi ^x_{-\infty:x_1}t^{(\varsigma)}_h\) has range at most \(2r\). The same argument applies to the right truncation \(\pi ^x_{x_2:+\infty}t^{(\varsigma)}_h\), so the anyon \(\varsigma^{(x_2)}_2\) created there also has range at most \(2r\).

For a finite horizontal segment \(\pi^x_{x_1:x_2}t^{(\varsigma)}_h\), linearity of \(\varepsilon\) gives
\begin{equation}\label{eq: finite string creates endpoint anyons}
\begin{aligned}
\varepsilon(\pi ^x_{x_1:x_2}t^{(\varsigma)}_h)
&=\varepsilon(t^{(\varsigma)}_h)-\varepsilon(\pi ^x_{-\infty:x_1}t^{(\varsigma)}_h)-\varepsilon(\pi ^x_{x_2:+\infty}t^{(\varsigma)}_h)\\
&= -\varsigma^{(x_1)}_1-\varsigma^{(x_2)}_2,
\end{aligned}
\end{equation}
since the full semi-infinite string \(t^{(\varsigma)}_h\) has trivial syndrome.
Thus, even when the two endpoint regions overlap, a finite segment of the string creates only endpoint anyons, each supported within a region of linear size at most \(2r\). The same argument applies to finite segments of the vertical string \(t^{(\varsigma)}_v\).

We now use these bounded-width strings to transport anyons inside \(D_{\mathrm{in}}\). Let \(\eta\) be a single-generator violation fully supported in \(D_{\mathrm{in}}\), and suppose its anyon type is \(\varsigma\). Fix a \(2r\times 2r\) square \(Q\subset D_{\mathrm{in}}\). 
Let \(\eta_0\) be an anyon of type \(-\varsigma\) fully supported in \(Q\). Such an anyon can always be obtained by appropriately truncating one of the thin string operators constructed above.

It is enough to show that \(\eta+\eta_0\) can be created by a Pauli operator fully supported in \(D\). Indeed, this means that any single-generator violation in \(D_{\mathrm{in}}\) can be transported to the fixed square \(Q\), up to a reference anyon, by an operator supported in \(D\).
To see that this implies the conditions required in Lemma~\ref{lemma: sTEE contribution}, choose, for each anyon type \(\varsigma\), a reference anyon \(\eta^{(\varsigma)}_0\sim -\varsigma\) supported in \(Q\), which can be obtained by truncating a suitable thin string operator. Now let
\[
\alpha=\sum_i \eta_i
\]
be an arbitrary anyon configuration supported in \(D_{\mathrm{in}}\), where each \(\eta_i\) is a single-generator violation of type \(\varsigma_i\). Applying the above construction to each \(\eta_i\), we obtain a Pauli operator supported in \(D\) that creates
\[
\alpha+\sum_i \eta^{(\varsigma_i)}_0 .
\]
The second term is supported entirely in \(Q\) and has anyon type opposite to that of \(\alpha\). Consider two cases:
\begin{enumerate}
    \item If \(\alpha\) is trivial, then \(\sum_i \eta^{(\varsigma_i)}_0\) is also trivial. Since it is supported near \(Q\), it can be created by a local Pauli operator supported in \(D\). Hence \(\alpha\) itself can be created by a Pauli operator supported in \(D\).
    \item If \(\alpha\) is nontrivial, then \(\sum_i \eta^{(\varsigma_i)}_0\) has the opposite anyon type to \(\alpha\). By attaching a semi-infinite thin string operator whose endpoint cancels this reference anyon in \(Q\), and multiplying it by the Pauli operator supported in \(D\) constructed above, we obtain a semi-infinite string operator supported in \(A\cup D\cup E\) whose endpoint creates \(\alpha\).
\end{enumerate}
Therefore, the two conditions in Lemma~\ref{lemma: sTEE contribution} are satisfied.

To show that \(\eta+\eta_0\) can be created by a Pauli operator in \(D\), we construct the following operators:
\begin{itemize}
    \item a vertical string \(t_{v}\), obtained by translating and truncating \(t^{(\varsigma)}_v\), such that one endpoint creates an anyon \(\eta_1\) supported in a \(2r\times 2r\) square \(Q_0\subset D_{\mathrm{in}}\) containing \(\supp \eta\), while the other endpoint creates an anyon \(\eta_2\) supported in a \(2r\times 2r\) square $Q_1$ at the same height as \(Q\);
    \item a horizontal string \(t_{h}\), obtained by translating and truncating \(t^{(\varsigma)}_h\), such that one endpoint creates an anyon \(\eta_3\) supported in the same square as \(\eta_2\), while the other endpoint creates an anyon \(\eta_4\) supported in \(Q\).
\end{itemize}
We further require that
\begin{equation}
 \eta\sim -\eta_1,\qquad \eta_2\sim -\eta_3.
\end{equation}
Since
\begin{equation}
\begin{aligned}
        &\varepsilon(t_{v})=\eta_1+\eta_2\sim 0,\\
        &\varepsilon(t_{h})=\eta_3+\eta_4\sim 0,
\end{aligned}
\end{equation}
the endpoint anyons satisfy
\begin{equation}\label{eq: equivalence chain of transported anyons}
    \eta\sim -\eta_1\sim \eta_2\sim -\eta_3\sim \eta_4\sim -\eta_0,
\end{equation}
so \(t_{v}\) and \(t_{h}\) move the anyon \(\eta\) into the common square \(Q\), up to local neutral pairs created along the way.

Generally, the anyons in the same $2r\times 2r$ square  do not cancel with each other. Nevertheless, Eq.~\eqref{eq: equivalence chain of transported anyons} states that  \(\eta+\eta_1\), \(\eta_2+\eta_3\) and $\eta_4+\eta_0$ are trivial anyons, Lemma~\ref{lemma: string size bound} gives local Pauli operators \(u\), \(u'\) and $v$ such that
\begin{equation}\label{eq: ui ui'}
    \varepsilon(u)=\eta+\eta_1,\qquad
    \varepsilon(u')=\eta_2+\eta_3,\qquad
    \varepsilon(v)=\eta_4+\eta_0.
\end{equation}
Thus \(u\) creates the pair \(\eta,\eta_1\), while \(u'\) creates the pair \(\eta_2,\eta_3\) and  \(v\) creates the pair \(\eta_4,\eta_0\).
Using Eq.~\eqref{eq: ui ui'}, we find
\begin{equation}\label{eq: move trivial anyon to Q}
\begin{aligned}
\varepsilon(-t_{v}-t_{h}+u+u'+v)
&= \bigl(-\eta_1-\eta_2-\eta_3-\eta_4 +\eta+\eta_1+\eta_2+\eta_3+\eta_4+\eta_0\bigr)\\
&= \eta+\eta_0.
\end{aligned}
\end{equation}
Thus \(\eta+\eta_0\) is indeed created by a Pauli operator built from the strings \(t_{v},t_{h}\), the local fusion operators \(u,u'\), and \(v\). The constructions of these Pauli operators are illustrated in Fig.~\ref{fig: geometric construction of tv th eta Q u u' v}.

\begin{figure}[t]
\centering
\subfigure[Construction of $\eta$, $Q$ and $t$.]{\includegraphics[width=0.25\linewidth]{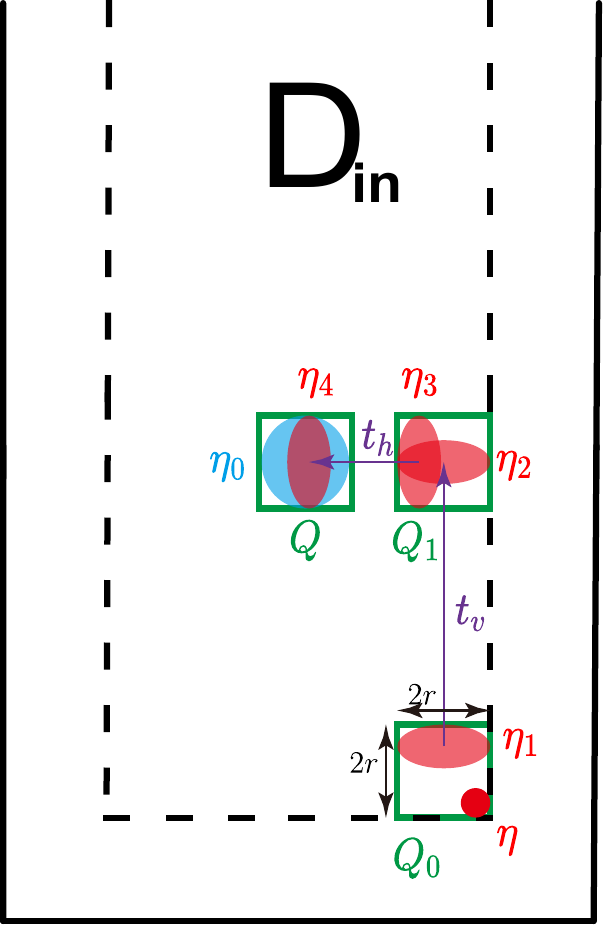}\label{fig: construction of tv th eta}
}
\hspace{8em}
\subfigure[Construction of $u$, $u'$ and $v$.]{\includegraphics[width=0.25\linewidth]{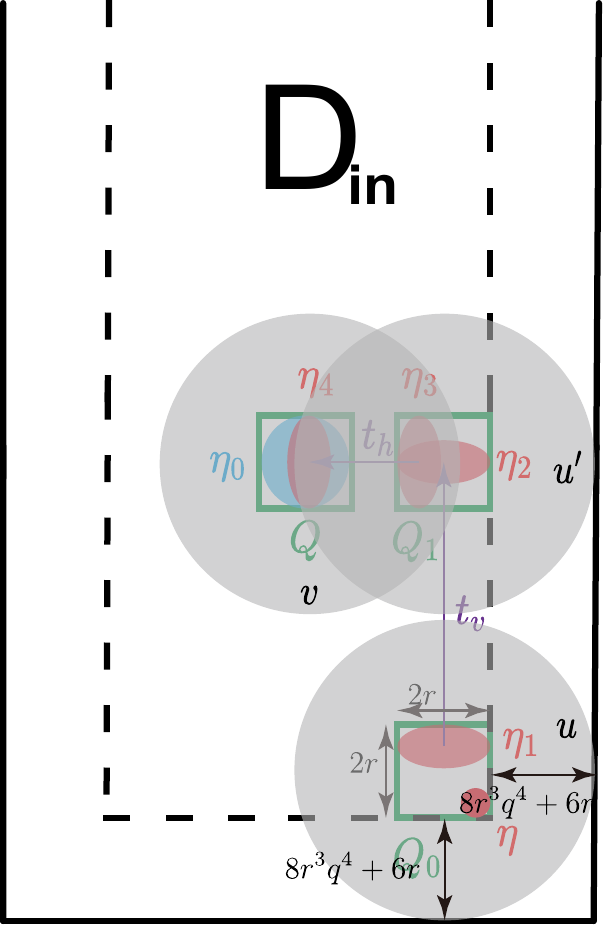}\label{fig: construction_of_u_u'_v}}
\caption{Construction of a Pauli operator, fully supported in \(D\), that creates the trivial anyon \(\eta+\eta_0\). (a) Construction of \(\eta\), \(Q\), and the string operators. Each green square represents a \(2r\times 2r\) region. The square \(Q\) is chosen arbitrarily, and the $Q_1$ lies at the same height as \(Q\). The red dot represents the single-generator violation \(\eta\), supported in \(Q_0\). The purple arrows, \(t_v\) and \(t_h\), are vertical and horizontal anyon strings of width at most \(r\); each creates endpoint anyons of range at most \(2r\). The ellipses denote the anyons created by \(t_v\) and \(t_h\), all of which have the same type as \(\eta\) or its opposite. In each green square, the total anyon charge is trivial. (b) Construction of \(u\), \(u'\), and \(v\). The gray circles denote the Pauli operators \(u\), \(u'\), and \(v\), which create the trivial anyons supported in the corresponding squares. By Lemma~\ref{lemma: string size bound}, each of these operators is supported within the \(8r^3q^4+6r\) neighborhood of its square. Therefore, if the separation between \(D\) and \(D_{\mathrm{in}}\) is \(8r^3q^4+6r\), all of them are fully supported in \(D\). The sum of \(t_v\), \(t_h\), \(u\), \(u'\), and \(v\) then gives a Pauli operator that creates \(\eta+\eta_0\). The upper parts of \(D_{\mathrm{in}}\) and \(D\) are omitted for clarity.}
\label{fig: geometric construction of tv th eta Q u u' v}
\end{figure}

It remains to bound the support of the resulting Pauli operator. The string operators \(t_{v}\) and \(t_{h}\) have width \(r\), while the local operators \(u\), \(u'\), and \(v\) are each supported in the \((8r^3q^4+2r+2r')\)-neighborhood of a \(2r\times 2r\) square contained in \(D_{\mathrm{in}}\). Since \(r'=2r\), this neighborhood radius is
\begin{equation}
    8r^3q^4+6r.
\end{equation}
Therefore,
\begin{equation}
    L_S=8r^3q^4+6r
\end{equation}
is sufficient to ensure that every operator appearing in Eq.~\eqref{eq: move trivial anyon to Q} is fully supported in \(D\). We conclude that every trivial anyon fully supported in \(D_{\mathrm{in}}\) can be created by a Pauli operator supported in \(D\).

With this separation fixed, we now proceed to show, in turn, that \(s_{S1}\), \(s_{S2}\), and \(s_{S3}\) are divisible.

\subsubsection{The divisibility of the product of stabilizer generators in the bulk of $A\cup B\cup C$}
$r<L$ implies that no generator can be supported on both $A$ and $C$ simultaneously. Since every generator appearing in $\theta_{S1}$ is fully supported in $A\cup B\cup C$, each such generator must in fact be supported either in $A\cup B$ or in $B\cup C$. Therefore, $s_{S1}$ can be directly factorized into a product of generators supported in $A\cup B$ and a product of generators supported in $B\cup C$.

\subsubsection{The divisibility of the product of stabilizer generators near the boundary of $D$}

The basic idea of decomposing $s_{S2}$ is to first split the coefficient vector \(\theta_{S2}\), which encodes the pattern of the generator product into pieces associated with the left, middle, and right parts of the geometry, and then further separate the middle piece into lower and upper strips. This produces a natural candidate decomposition of the stabilizer \(w(\theta_{S2})\). The only obstruction is that the resulting partial products may still leave residual support in region \(D\), but locality implies that such residual terms can only appear near the four corners of \(D\). We will remove these corner contributions by adding suitable boundary gauge operators, thereby turning the decomposition into one supported entirely in \(A\cup B\) and \(B\cup C\).

From Lemma~\ref{lemma: sTEE contribution} we know that an sTEE-contributing stabilizer cannot contain generator components fully supported in \(D_{\mathrm{in}}\). Using locality of generators, this implies that \(\theta_{S2}\) must satisfy
\begin{equation}\label{eq: pi' theta_S}
\pi'_{\mathrm{in}}\theta_{S2}=0.
\end{equation}
Here \(\pi'_{\mathrm{in}}:R\to R\) is the projection onto monomials fully supported in \(D'_{\mathrm{in}}\), where
\begin{equation}\label{eq: def of D' in}
\begin{aligned}
D'_{\mathrm{in}}
:=\{(n,m)\in\Lambda \mid\;&L+8r^3q^4+6r\le n<2L-8r^3q^4-7r,\\
&L+8r^3q^4+6r\le m<4L-8r^3q^4-7r\}.
\end{aligned}
\end{equation}

Equation~\eqref{eq: pi' theta_S} is a necessary but not sufficient condition for excluding generators fully supported in \(D_{\mathrm{in}}\). Geometrically, each generator is supported in an \(r\times r\) square anchored at its lower-left corner, so a generator is certainly fully supported in \(D_{\mathrm{in}}\) if its anchor point lies inside \(D_{\mathrm{in}}\) and is at least a distance \(r\) from the upper and right boundaries of \(D_{\mathrm{in}}\).

Since each generator overlapping $D$ cannot penetrate into $A\cup B\cup C$ to a depth more than $r$, the definition of \(\theta_{S2}\), together with locality of the generators, implies
\begin{equation}\label{eq: pi theta_S2}
    \pi^x_{L-r:2L}\theta_{S2}=\theta_{S2},
    \qquad
    \pi^y_{L-r:4L}\theta_{S2}=\theta_{S2}.
\end{equation}
We first split \(\theta_{S2}\) into lower, middle, and upper parts:
\begin{equation}\label{eq: partition of theta}
\begin{aligned}
    &\theta_{S2} = \theta_{d2}+\theta_{m2}+\theta_{u2},\\
    &\theta_{d2}=\pi^y_{L-r:L+8r^3 q^4+6r}\theta_{S2},\\
    &\theta_{m2}=\pi^y_{L+8r^3 q^4 +6r:4L- 8r^3 q^4 - 7r}\theta_{S2},\\
    &\theta_{u2}=\pi^y_{4L-8r^3 q^4-7r:4L}\theta_{S2}.
\end{aligned}
\end{equation}

Substituting the definition of \(\theta_{m2}\) into Eq.~\eqref{eq: pi' theta_S} and using the commutativity of the two projections gives
\begin{equation}
    \begin{aligned}
       0=& \pi'_{\mathrm{in}}\theta_{m2}\\
=&\pi'_{\mathrm{in}}\pi^y_{L+8r^3 q^4 +6r:4L- 8r^3 q^4 - 7r}\theta_{S2}\\
=&\pi^x_{L+8r^3q^4+6r: 2L-8r^3q^4-7r}\pi^y_{L+8r^3 q^4 +6r:4L- 8r^3 q^4 - 7r}\theta_{S2}\\
=&\pi^x_{L+8r^3q^4+6r: 2L-8r^3q^4-7r}\theta_{m2}.
    \end{aligned}
\end{equation}
Hence every nonzero monomial \(k_ix^{n_i}y^{m_i}e_\mu\) appearing in \(\theta_{m2}\) must satisfy
\begin{equation}
    \begin{aligned}
            &n_i<L+8r^3q^4+6r
            \quad\text{or}\quad
            2L-8r^3q^4-7r\le n_i.
    \end{aligned}
\end{equation}
This condition, combining with Eq.~\eqref{eq: pi theta_S2}, implies that \(\theta_{m2}\) is separated by $D'_{in}$ into two non-overlapping pieces,
\begin{equation}\label{eq: theta_S2}
    \theta_{m2}
    =\pi^x_{L-r:L+8r^3q^4+6r}\theta_{m2}
    +\pi^x_{2L-8r^3q^4-7r:2L}\theta_{m2}.
\end{equation}

We, therefore, referring to the partition, denote the two disjoint parts of $\theta_{S2}$ with $\theta_{A2}$ and $\theta_{C2}$, and call the sum of $\theta_{d2}$ and $\theta_{u2}$ as $\theta_{B2}$:
\begin{equation}\label{eq: theta_A theta_B theta_C}
\begin{aligned}
    &\theta_{A2}
    =\pi^x_{L-r:L+8r^3q^4+6r}\theta_{m2}
    =\pi^y_{L+8r^3 q^4 +6r:4L- 8r^3 q^4 - 7r}\pi^x_{L-r:L+8r^3q^4+6r}\theta_{S2},\\
    &\theta_{B2}
    =\theta_{d2}+\theta_{u2}
    =\pi^y_{L-r:L+8r^3 q^4+6r}\theta_{S2}
    +\pi^y_{4L-8r^3 q^4-7r:4L}\theta_{S2},\\
    &\theta_{C2}
    =\pi^x_{2L-8r^3q^4-7r:2L}\theta_{m2}
    =\pi^y_{L+8r^3 q^4 +6r:4L- 8r^3 q^4 - 7r}\pi^x_{2L-8r^3q^4-7r:2L}\theta_{S2}.
\end{aligned}
\end{equation}
Then \(\theta_{S2}=\theta_{A2}+\theta_{B2}+\theta_{C2}\), as illustrated in Fig.~\ref{fig: G_decomposition_ABC1}.

\begin{figure*}[thb]
 \centering
 \includegraphics[width=1\linewidth]{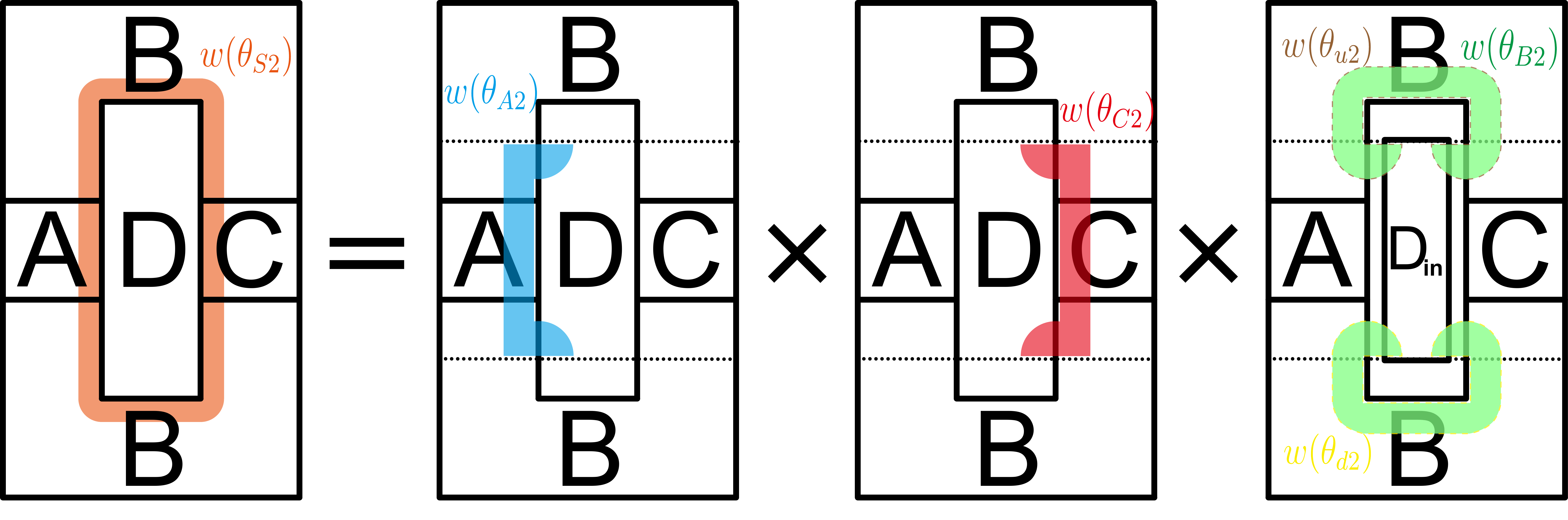}
 \caption{Decomposition of a stabilizer $w(\theta_{S2})$ (which corresponds to $\mathcal{S}$ in Fig.~\ref{fig: G_decomposition_ABC}) supported on $A\cup B\cup C$ into three stabilizers $w(\theta_{A2})$, $w(\theta_{B2})$, and $w(\theta_{C2})$ (which correspond to $\mathcal{J}_A,\mathcal{J}_B,\mathcal{J}_C$ in Fig.~\ref{fig: G_decomposition_ABC}). $w(\theta_{B2})$ consists of stabilizer generators supported entirely within the infinite strip obtained by extending the upper and lower boundaries of $D_{\mathrm{in}}'$ to infinity along the $x$ direction and it can be decomposed into two separate stabilizers $w(\theta_{d2})$ and $w(\theta_{u2})$. Each colored region indicates the support of the corresponding $w(\theta_{i2})$. At this stage, these partial products may still have residual support in region $D$ near the corners.}
 \label{fig: G_decomposition_ABC1}
\end{figure*}

At this stage, the decomposition is still not of the form required by Eq.~\eqref{eq: indivisibility condition}, because the partial products may retain some support in \(D\). The key point is that locality forces all such residual terms to lie near the four corners of \(D\).

From the definitions above,
\begin{equation}
\begin{aligned}
    &\supp(\theta_{A2})\subset [L-r,L+8r^3q^4+6r)\times [L+8r^3 q^4+6r,4L- 8r^3 q^4 - 7r),\\
    &\supp(\theta_{d2})\subset [L-r,2L)\times [L-r,L+8r^3 q^4+6r),\\
    &\supp(\theta_{u2})\subset [L-r,2L)\times [4L- 8r^3 q^4 - 7r,4L),\\
    &\supp(\theta_{C2})\subset [2L- 8r^3 q^4 - 7r,2L)\times [L+8r^3 q^4+6r,4L- 8r^3 q^4 - 7r).
\end{aligned}
\end{equation}
Applying locality of the generators then gives
\begin{equation}
\begin{aligned}
    &\supp(w(\theta_{A2}))\subset [L-r,L+8r^3q^4+7r)\times [L+8r^3 q^4+6r,4L- 8r^3 q^4 - 6r),\\
    &\supp(w(\theta_{d2}))\subset [L-r,2L+r)\times [L-r,L+8r^3 q^4+7r),\\
    &\supp(w(\theta_{u2}))\subset [L-r,2L+r)\times [4L- 8r^3 q^4 - 7r,4L+r),\\
    &\supp(w(\theta_{C2}))\subset [2L- 8r^3 q^4 - 7r,2L+r)\times [L+8r^3 q^4+6r,4L- 8r^3 q^4 - 6r).
\end{aligned}
\end{equation}
Therefore the only possible overlaps are
\begin{equation}
w(\theta_{A2})\cap w(\theta_{d2}),\qquad
w(\theta_{C2})\cap w(\theta_{d2}),\qquad
w(\theta_{A2})\cap w(\theta_{u2}),\qquad
w(\theta_{C2})\cap w(\theta_{u2}),
\end{equation}
and their supports are bounded by
\begin{equation}
    \begin{aligned}
        &\supp\bigl(w(\theta_{A2})\bigr)\cap \supp\bigl(w(\theta_{d2})\bigr)
        \subset  [L-r,L+8r^3q^4+7r)\times[L+8r^3 q^4+6r,L+8r^3 q^4+7r),\\ 
        &\supp\bigl(w(\theta_{C2})\bigr)\cap \supp\bigl(w(\theta_{d2})\bigr)
        \subset[2L- 8r^3 q^4 - 7r,2L+r)\times[L+8r^3 q^4+6r,L+8r^3 q^4+7r),\\
        &\supp\bigl(w(\theta_{A2})\bigr)\cap \supp\bigl(w(\theta_{u2})\bigr)
        \subset  [L-r,L+8r^3q^4+7r)\times [4L- 8r^3 q^4 - 7r,4L- 8r^3 q^4 - 6r),\\
        &\supp\bigl(w(\theta_{C2})\bigr)\cap \supp\bigl(w(\theta_{u2})\bigr)
        \subset[2L- 8r^3 q^4 - 7r,2L+r)\times[4L- 8r^3 q^4 - 7r,4L- 8r^3 q^4 - 6r).
    \end{aligned}
\end{equation}

We isolate these projections of the four corner pieces into $D$ by defining
\begin{equation}\label{eq: o_dL o_dR o_uL o_uR}
    \begin{aligned}
        &o_{dL}:=\pi^x_{L:L+8r^3q^4+7r}\pi^y_{L+8r^3 q^4+6r:L+8r^3 q^4+7r}w(\theta_{d2}),\\
        &o_{dR}:=\pi^x_{2L- 8r^3 q^4 - 7r:2L}\pi^y_{L+8r^3 q^4+6r:L+8r^3 q^4+7r}w(\theta_{d2}),\\
        &o_{uL}:=\pi^x_{L:L+8r^3q^4+7r}\pi^y_{4L- 8r^3 q^4 - 7r:4L- 8r^3 q^4 - 6r}w(\theta_{u2}),\\
        &o_{uR}:=\pi^x_{2L- 8r^3 q^4 - 7r:2L}\pi^y_{4L- 8r^3 q^4 - 7r:4L- 8r^3 q^4 - 6r}w(\theta_{u2}).
    \end{aligned}
\end{equation}

Since $s_{S2}$, the sum of four $w(\theta_{*2})$ terms, has no supports in $D$, and the support bounds above show that these are the only possible residual contributions in \(D\), the terms inside \(D\) must be corner-localized and must cancel pairwise. Therefore,
\begin{equation}\label{eq: pi_D w(theta)}
\begin{aligned}
    &\pi_D w(\theta_{A2})=-o_{dL}-o_{uL},\\
    &\pi_D w(\theta_{u2})=o_{uL}+o_{uR},\\
    &\pi_D w(\theta_{C2})=-o_{dR}-o_{uR},\\
    &\pi_D w(\theta_{d2})=o_{dL}+o_{dR}.
\end{aligned}
\end{equation}
Here \(\pi_D\) denotes projection into \(D\).

We now shift these corner terms into the bulk of $B$ by adding suitable stabilizers in $B\cup D$.

 \begin{figure}[t]
     \centering
     \includegraphics[width=1\linewidth]{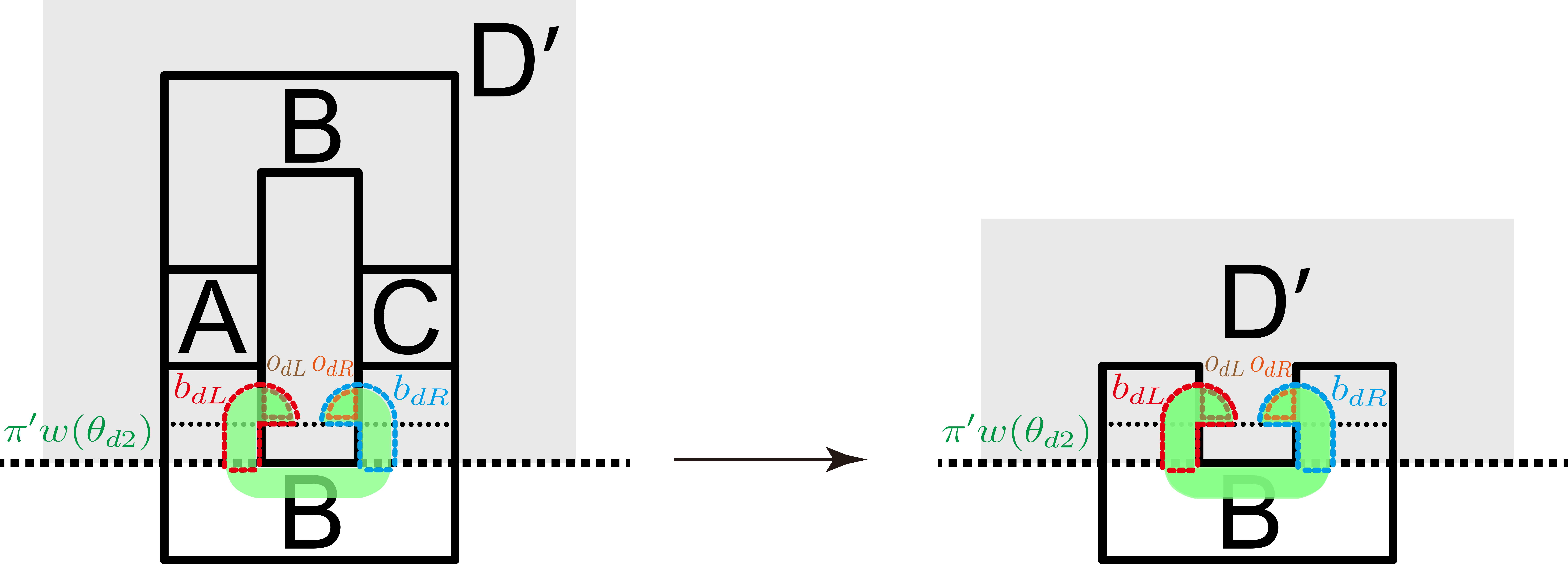}
     \caption{Restriction of $w(\theta_{d2})$ to the semi-infinite plane $D'$, which coincides with the lower boundary of $D$. The truncated Pauli operator factorizes as $b_{dL}, b_{dR}$, where $b_{dL}$ and $b_{dR}$ are boundary gauge operators supported near the left and right corners, respectively. The projection of $b_{dL}$ and $b_{dR}$ into $D$ is denoted by $o_{dL}$ and $o_{dR}$, respectively.}
     \label{fig: C8_2}
 \end{figure}

We begin with the lower strip \(w(\theta_{d2})\). Let
\begin{align}
    D'&:=\left\{(n,m)\in\Lambda \,\middle|\,\ L\le m \right\},
\end{align}
the semi-infinite plane sharing the lower boundary of \(D\), and let \(\pi'\) be the projection into \(D'\):
\begin{align}
    \pi'\!\left(\sum_i k_ix^{n_i}y^{m_i}\right)
    &=\sum_{(n_i,m_i)\in D'} k_ix^{n_i}y^{m_i}.
\end{align}
Its truncated support obeys
\begin{equation}
\begin{aligned}
\supp(\pi'w(\theta_{d2}))\subset [L-r,2L+r)\times [L,L+8r^3 q^4+7r).
\end{aligned}
\end{equation}
Accordingly, $\pi'w(\theta_{d2})$ can be vertically sliced into three terms:
\begin{equation}\label{eq: pi'w(theta_d2)}
\begin{aligned}
    \pi'w(\theta_{d2})
    =&\pi^x_{L-r:2L+r}\pi^y_{L:L+8r^3q^4+7r}w(\theta_{d2})\\
    =&\pi^x_{L-r:L}\pi^y_{L:L+8r^3 q^4+7r}w(\theta_{d2})\\
    +&\pi^x_{L:2L}\pi^y_{L:L+8r^3 q^4+7r}w(\theta_{d2})\\
    +&\pi^x_{2L:2L+r}\pi^y_{L:L+8r^3 q^4+7r}w(\theta_{d2}).
\end{aligned}
\end{equation}

When
\begin{equation}\label{eq: 8r^3q^4+7r<L}
    8r^3q^4+7r<L,
\end{equation}
the middle term is exactly the part of $\pi 'w(\theta_{d2})$ lying in \(D\), so using Eq.~\eqref{eq: pi_D w(theta)} we obtain
\begin{equation}\label{eq: pi^x pi^y p (theta_d2)}
\begin{aligned}
     \pi^x_{L:2L}\pi^y_{L:L+8r^3 q^4+7r}w(\theta_{d2})
     &=\pi_Dw(\theta_{d2})\\
     &=o_{dL}+o_{dR}.
\end{aligned}
\end{equation}
We therefore define
\begin{equation}\label{eq: b_dL b_dR}
    \begin{aligned}
        &b_{dL}:=o_{dL}+\pi^x_{L-r:L}\pi^y_{L:L+8r^3q^4+7r}w(\theta_{d2}),\\
        &b_{dR}:=o_{dR}+\pi^x_{2L:2L+r}\pi^y_{L:L+8r^3 q^4+7r}w(\theta_{d2}),
    \end{aligned}
\end{equation}
so that 
\begin{equation}\label{eq: pi'w(theta_d2=b_dL+b_dR)}
\pi'w(\theta_{d2})=b_{dL}+b_{dR},
\end{equation}
which means that, the projection of $w(\theta_{d2})$ into the semi-infinite plane $D'$ can be decomposed into two non-overlapping terms $b_{dL}$ and $b_{dR}$.

We then want to show that  $b_{dL}$ and $b_{dR}$ are boundary gauge operators themselves and then apply Lemma~\ref{lemma: height limit} to obtain a suitable decoration to shift them into the bulk of $B$.

By definition, a boundary gauge operator commutes with all bulk stabilizers, or equivalently with a generating set of bulk stabilizers. Near the boundary of \(D'\), however, this generating set is not limited to the original stabilizer generators fully supported in \(D'\): products of generators whose support outside \(D'\) cancels can also act as bulk stabilizers. We capture this enlarged effective generating set by computing a Gr\"obner basis of the stabilizer module in an anti-lexicographic TOP order.

The anti-lexicographic TOP order is defined as
\begin{equation}
    x^{n_1}y^{m_1}\prec x^{n_2}y^{m_2}
    \quad\text{if}\quad
    \left\{
    \begin{array}{l}
         m_1>m_2,\\
         \text{or}\\
         m_1=m_2,\ n_1<n_2.
    \end{array}
    \right.
\end{equation}
Denote the obtained Gr\"obner basis set by $g''^1,\dots ,g''^{N}$, Let
\begin{align}
&l''^\mu_{xi}:=x_{\max}(g''^\mu_i),\qquad
l''^\mu_{yi}:=y_{\max}(g''^\mu_i),
\qquad  i=1,\dots,2q,\ \mu=1,\dots,n''_S,\\
&l''^\mu_x:=\max_i\{l''^\mu_{xi}\},\qquad
l''^\mu_y:=\max_i\{l''^\mu_{yi}\},\\
&l''_x:=\max_\mu\{l''^\mu_x\},\qquad
l''_y:=\max_\mu\{l''^\mu_y\}.
\end{align}

Under this order, the leading term of an element is the term of smallest height. Therefore, by the same reasoning used to obtain Eq.~\eqref{eq: b'_Bn}, every bulk stabilizer \(s_{D'}\) in \(D'\) can be written as
\begin{equation}
    \begin{aligned}
            s_{D'}&=\sum_{\mu=1}^{n''_S} k''_\mu g''^\mu,
    \qquad
    k''_\mu\in R,\\
    y_{\min}(k''_\mu g''^\mu)&\ge y_{\min}(s_{D'}),
    \qquad \forall \mu.
    \end{aligned}
\end{equation}
Since \(s_{D'}\) is fully supported in \(D'\), every translated generator appearing in its decomposition must also be fully supported in \(D'\). Thus any bulk stabilizer in \(D'\) can be generated by translated \(g''^\mu\) supported entirely in \(D'\). It therefore suffices to show that \(b_{dL}\) commutes with every such translated generator.

For such a generator \(x^{n''}y^{m''}g''^\mu\), its support is given by
\begin{equation}
\begin{aligned}
    \supp \left(w(x^{n''}y^{m''}g''^\mu)\right)
    \subset [n'',n''+l''_x)\times [m'',m''+l''_y).
\end{aligned}
\end{equation}
On the other hand,
\begin{equation}\label{eq: max_x b_dL b_dR}
\begin{aligned}
    &x_{\max}(b_{dL})
    =x_{\max}\!\left(o_{dL}+\pi^x_{L-r:L}\pi^y_{L:L+8r^3q^4+7r}w(\theta_{d2})\right)
    <L+8r^3q^4+7r,\\
    &x_{\min}(b_{dR})
    =x_{\min}\!\left(o_{dR}+\pi^x_{2L:2L+r}\pi^y_{L:L+8r^3 q^4+7r}w(\theta_{d2})\right)
    \ge 2L- 8r^3 q^4 - 7r.
\end{aligned}
\end{equation}
Thus, if
\begin{equation}\label{eq: l''_x<L- 16r^3 q^4 - 14r}
    l''_x\le L- 16r^3 q^4 - 14r,
\end{equation}
then a translated generator in \(D'\) cannot overlap both \(b_{dL}\) and \(b_{dR}\). If it does not overlap \(b_{dL}\), then it trivially commutes with \(b_{dL}\). Otherwise,
\begin{equation}
\begin{aligned}
    b_{dL}*x^{n''}y^{m''}g''^\mu
    &=b_{dL}*x^{n''}y^{m''}g''^\mu+b_{dR}*x^{n''}y^{m''}g''^\mu\\
    &=\pi'(w(\theta_{d2}))*x^{n''}y^{m''}g''^\mu\\
    &=w(\theta_{d2})*\pi'x^{n''}y^{m''}g''^\mu\\
    &=w(\theta_{d2})*x^{n''}y^{m''}g''^\mu\\
    &=0.
\end{aligned}
\end{equation}
The second step uses Eq.~\eqref{eq: pi'w(theta_d2=b_dL+b_dR)}, the third uses Eq.~\eqref{eq: pi p_1*p_2}, and the fourth uses the fact that \(x^{n''}y^{m''}g''^\mu\) is fully supported in \(D'\). The fifth uses the fact that $w(\theta_{d2})$ and $x^{n''}y^{m''}g''^\mu$ are stabilizers. Hence \(b_{dL}\) is supported in \(D'\) and commutes with all bulk stabilizers in \(D'\), so it is a boundary gauge operator in \(D'\). The same holds for \(b_{dR}\).

By Lemma~\ref{lemma:secondary boundary gauge}, there exists \(\xi_{dL}\in S\) such that
\begin{equation}
    b_{dL}=w(\xi_{dL}),
\end{equation}
and its height is bounded by
\begin{equation}\label{eq: max_y min_y xi_dL}
    \supp(\xi_{dL})\subset \mathbb{Z}\times (L-r,L+8r^3q^4+5r+3\max\{r,h\}],
\end{equation}
where \(h\) is the height of \(b_{dL}\). From the definition of \(b_{dL}\),
\begin{equation}
    h<8r^3q^4+7r.
\end{equation}
Substituting this bound gives
\begin{equation}
    \supp(\xi_{dL})\subset \mathbb{Z}\times (L-r, L+32r^3q^4+26r).
\end{equation}

With a decomposition of an explicit height bound, we now slice \(\xi_{dL}\) to the finite horizontal window relevant for the decomposition:
\begin{equation}\label{eq: pi^x xi_dL}
\begin{aligned}
\xi'_{dL}&:=\pi^x_{L-r:2L}\xi_{dL}.
\end{aligned}
\end{equation}
Its projection into \(D\) is
\begin{equation}\label{eq: pi_Dw(xi'_dL)}
\begin{aligned}
        \pi_Dw(\xi'_{dL})
        &=\pi_Dw(\xi_{dL})-\pi_Dw(\pi^x_{-\infty:L-r}\xi_{dL}+\pi^x_{2L:+\infty}\xi_{dL})\\
        &=\pi_D\pi'w(\xi_{dL})+0\\
        &=\pi_D(b_{dL})\\
        &=o_{dL}.
\end{aligned}
\end{equation}
Hence, the residual contribution of $w(\xi'_{dL})$ in $D$ is exactly the same as those of $b_{dL}$, or those of $w(\theta_{d2})$ near the bottom-left corner.

Moreover,
\begin{equation}
    \begin{aligned}
    &\supp (\xi'_{dL})\subset [L,2L)\times [L-r,L+32r^3q^4+26r),\\
    &\supp (w(\xi'_{dL}))\subset [L,2L+r)\times [L-r,L+32r^3q^4+27r).
    \end{aligned}
\end{equation}
When
\begin{equation}\label{eq:  L+32r^3q^4+27r leq frac 54 L}
    L+32r^3q^4+27r\le 2L,
    \qquad
    0\le L-r,
\end{equation}
The support bound of \(w(\xi'_{dL})\) ensures that it does not reach \(A\) or \(C\), and hence is fully supported in \(B\cup D\).

The operator \(\xi'_{dR}\) is constructed in the same way and satisfies
\begin{equation}\label{eq: pi_Dw(xi'_dR)}
    \pi_Dw(\xi'_{dR})=o_{dR},
    \qquad
    \supp \bigl(w(\xi'_{dR})\bigr)\subset B\cup D.
\end{equation}
Using Eq.~\eqref{eq: pi_D w(theta)}, Eq.~\eqref{eq: pi_Dw(xi'_dL)}, and Eq.~\eqref{eq: pi_Dw(xi'_dR)}, we find that subtracting $\xi'_{dL}+\xi'_{dR}$ from $\theta_{d2}$ fully cancels its residual contributions in $D$:
\begin{equation}
   \pi_Dw(\theta_{d2}-\xi'_{dL}-\xi'_{dR})
   =o_{dL}+o_{dR}-o_{dL}-o_{dR}
   =0.
\end{equation}
On the other hand, Eq.~\eqref{eq: partition of theta} implies
\begin{equation}
    y_{\max}(w(\theta_{d2}))\le y_{\max}(\theta_{d2})+r\le 8r^3 q^4+7r.
\end{equation}
Hence, if
\begin{equation}\label{eq: L+8r^3 q^4+6r<frac 54 L}
L+8r^3 q^4+7r<2L,
\end{equation}
then \(w(\theta_{d2})\) is fully supported in \(B\cup D\) since it is not high enough to reach $A$ or $C$, and so is \(w(\theta_{d2}-\xi'_{dL}-\xi'_{dR})\). Since its projection into \(D\) vanishes, we conclude that
\begin{equation}
    \supp\bigl(w(\theta_{d2}-\xi'_{dL}-\xi'_{dR})\bigr)\subset B.
\end{equation}
In other words, by adding \(-w(\xi'_{dL}+\xi'_{dR})\) we decorate \(w(\theta_{d2})\) with stabilizers and push its support entirely into \(B\).

\begin{figure*}[thb]
    \centering
    \subfigure[ ]{\includegraphics[scale=0.08]{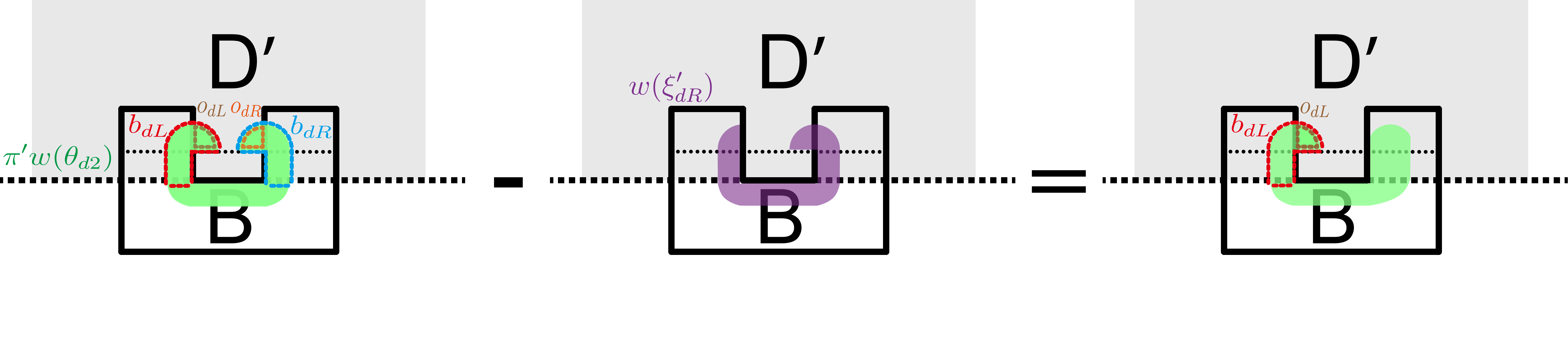}
}
    \\
    \hspace{0.02mm}
    \subfigure[ ]{\includegraphics[scale=0.08]{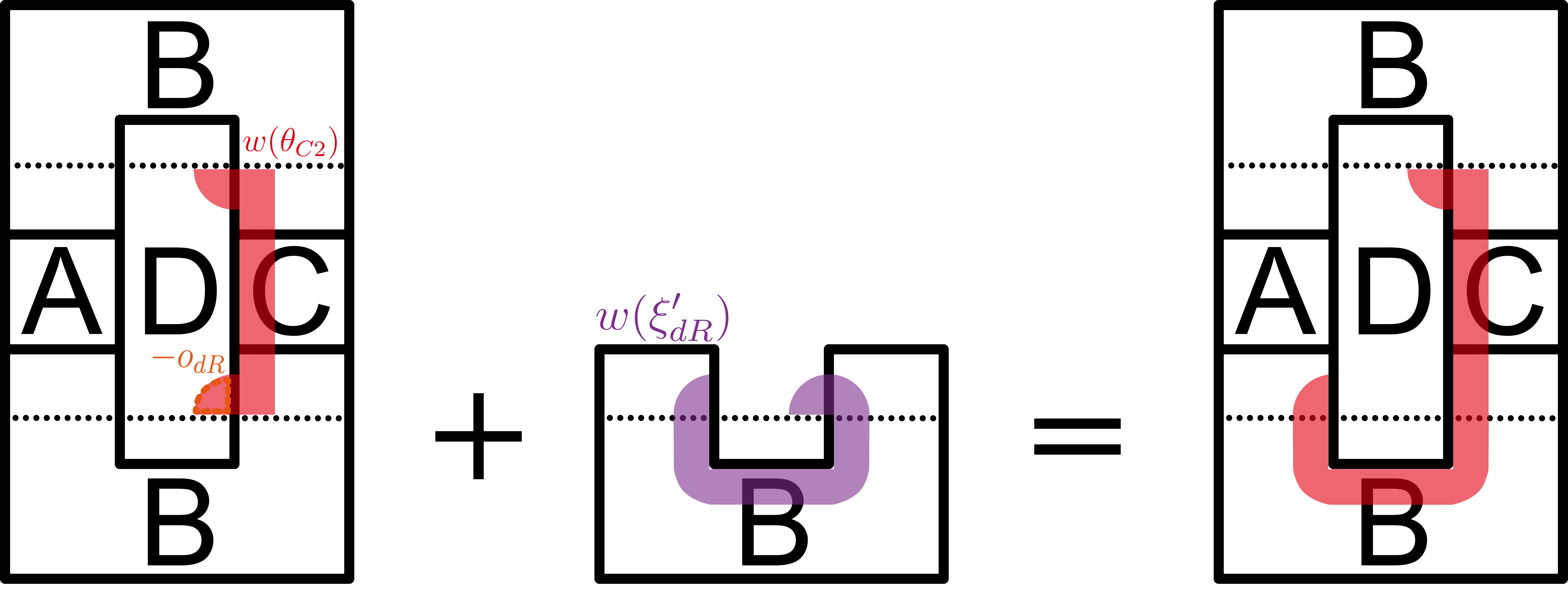}\label{fig: boundary_gauge_moving2_1}
}
    \caption{Eliminating residual support in region $D$ by translating boundary gauge operators. By subtracting stabilizer $w(\xi'_{dR})$, the supports of $b_{dR}$ can be shifted outward into $B$, canceling the residual contribution of $\pi'w(\theta_{d2})$ in $D$ around the bottom right corner. Moreover, adding $w(\xi'_{dR})$ cancels the residual contribution of $w(\theta_{C2})$ in $D$ around the bottom right corner.}
    \label{fig: boundary_gauge_moving_1}
\end{figure*}

The same construction applies to the upper strip. We obtain \(\xi'_{uL}\) and \(\xi'_{uR}\) such that
\begin{equation}
\begin{aligned}
     &\pi_Dw(\xi'_{uL})=o_{uL},
     \qquad
     \supp \bigl(w(\xi'_{uL})\bigr)\subset B\cup D,\\
     &\pi_Dw(\xi'_{uR})=o_{uR},
     \qquad
     \supp \bigl(w(\xi'_{uR})\bigr)\subset B\cup D,\\
     &\pi_Dw(\theta_{u2}-\xi'_{uL}-\xi'_{uR})=0.
\end{aligned}
\end{equation}
The residual contributions of $w(\theta_{A2})$ and $w(\theta_{C2})$ must cancel with those terms in $w(\theta_{u2})$ or $w(\theta_{d2})$, so using Eq.~\eqref{eq: pi_D w(theta)}, we obtain that adding those $\xi'_{*}$ terms cancels the residual contributions of $w(\theta_{A2})$ and $w(\theta_{C2})$ in $D$:
\begin{equation}
    \begin{aligned}
        &\pi_Dw(\theta_{d2}-\xi'_{dL}-\xi'_{dR})=0,\\
        &\pi_Dw(\theta_{u2}-\xi'_{uL}-\xi'_{uR})=0,\\
        &\pi_Dw(\theta_{A2}+\xi'_{dL}+\xi'_{uL})=0,\\
        &\pi_Dw(\theta_{C2}+\xi'_{dR}+\xi'_{uR})=0.
    \end{aligned}
\end{equation}
Combining this with
\begin{equation}
    \begin{aligned}
        &\supp\bigl(w(\theta_{d2})\bigr)\subset B\cup D,\\
        &\supp\bigl(w(\theta_{u2})\bigr)\subset B\cup D,\\
        &\supp\bigl(w(\theta_{A2})\bigr)\subset A\cup B\cup D,\\
        &\supp\bigl(w(\theta_{C2})\bigr)\subset B\cup D\cup C,
    \end{aligned}
\end{equation}
where the last three lines follow by the same geometric argument used for \(w(\theta_{d2})\), we obtain
\begin{equation}
    \begin{aligned}
        &\supp\bigl(w(\theta_{d2}-\xi'_{dL}-\xi'_{dR})\bigr)\subset B,\\
        &\supp\bigl(w(\theta_{u2}-\xi'_{uL}-\xi'_{uR})\bigr)\subset B,\\
        &\supp\bigl(w(\theta_{A2}+\xi'_{dL}+\xi'_{uL})\bigr)\subset A\cup B,\\
        &\supp\bigl(w(\theta_{C2}+\xi'_{dR}+\xi'_{uR})\bigr)\subset B\cup C.
    \end{aligned}
\end{equation}

Finally,
\begin{equation}
    s_{S2}=w(\theta_{S2})
    =w(\theta_{d2}-\xi'_{dR}+\theta_{u2}-\xi'_{uR}+\theta_{A2})
    +w(\theta_{C2}+\xi'_{dR}+\xi'_{uR}),
\end{equation}
and the first term can be regrouped as
\begin{equation}
    \begin{aligned}
        w(\theta_{d2}-\xi'_{dR}+\theta_{u2}-\xi'_{uR}+\theta_{A2})
        &=w(\theta_{d2}-\xi'_{dL}-\xi'_{dR})\\
        &\quad +w(\theta_{A2}+\xi'_{dL}+\xi'_{uL})\\
        &\quad +w(\theta_{u2}-\xi'_{uL}-\xi'_{uR}).
    \end{aligned}
\end{equation}
These three terms are supported in \(B\), \(A\cup B\), and \(B\), respectively. Hence
\begin{equation}
w(\theta_{d2}-\xi'_{dR}+\theta_{u2}-\xi'_{uR}+\theta_{A2})
\end{equation}
is fully supported in \(A\cup B\), while
\begin{equation}
w(\theta_{C2}+\xi'_{dR}+\xi'_{uR})
\end{equation}
is fully supported in \(B\cup C\). Both are stabilizers because they lie in the image of \(w\). Therefore \(s_{S2}\) admits a nontrivial decomposition of the required form, and thus \(s_{S2}\) is divisible.

\subsubsection{The divisibility of the product of stabilizer generators near the boundary of $E$}

We now show that \(s_{S3}\) is divisible. The strategy is to isolate the part of \(s_{S3}\) that can touch region \(C\). We first use the topological order condition to choose a local generator expansion of \(s_{S3}\), so that every generator term stays within a fixed finite neighborhood of \(A\cup B\cup C\cup D\). We then separate \(\theta_{S3}\) according to its \(x\)-support. Terms sufficiently far to the left cannot reach \(C\) by locality, so any contribution to \(C\) must come from a narrow strip near the right edge. We rewrite that strip using a Gr\"obner-basis decomposition adapted to the \(x\)-direction, and then extract the terms whose image intersects \(C\). These terms are shown to be supported entirely in \(B\cup C\). Subtracting them leaves a stabilizer supported in \(A\cup B\), which proves divisibility.

Since \(s_{S3}\) is a stabilizer fully supported in \(A\cup B\cup C\cup D\), the topological order condition implies that \(\theta_{S3}\) 
can be chosen as a vector with finite support, so that each of its nonzero monomial terms \(k_i x^{n_i}y^{m_i} e_\mu\), with \(\mu=1,\dots,n_S\), satisfies
\begin{equation}
    \supp\bigl(w(k_i x^{n_i}y^{m_i} e_\mu)\bigr)\subset b_H(A\cup B\cup C\cup D),
\end{equation}
where \(b_H(A\cup B\cup C\cup D)\) is the \(H\)-neighborhood of \(A\cup B\cup C\cup D\), and \(H\) is some finite nonnegative integer.

In particular, every such term obeys the uniform support bound
\begin{equation}\label{eq: bound of w(k_i x^n_i y^m_i e_mu)}
    \supp\bigl(w(k_i x^{n_i}y^{m_i} e_\mu)\bigr)\subset [-H,3L+H)\times [-H,5L+H).
\end{equation}

Now consider a term that is not fully supported in \(A\cup B\cup C\), but is still contained in the enlarged region \(A\cup B\cup C\cup E\). Such a term must cross at least one outer boundary of \(A\cup B\cup C\), and therefore must satisfy at least one of the following:
\begin{equation}\label{eq: not fully supported in A cup B cup C condition}
    \begin{aligned}
        &x_{\max}\bigl(w(k_i x^{n_i}y^{m_i} e_\mu)\bigr)\ge 3L,\qquad
        x_{\min}\bigl(w(k_i x^{n_i}y^{m_i} e_\mu)\bigr)<0,\\
        &y_{\max}\bigl(w(k_i x^{n_i}y^{m_i} e_\mu)\bigr)\ge 5L,\qquad
        y_{\min}\bigl(w(k_i x^{n_i}y^{m_i} e_\mu)\bigr)<0.
    \end{aligned}
\end{equation}

Using locality of the generators, these four conditions imply, respectively:
\begin{equation}\label{eq: not fully supported in A cup B cup C condition 3}
    \begin{aligned}
        &x_{\max}\bigl(w(k'_j x^{n_j}y^{m_j} e_\mu)\bigr)\ge 3L,\qquad
        x_{\max}\bigl(w(k'_j x^{n_j}y^{m_j} e_\mu)\bigr)\ge 3L-r,\\
        &x_{\min}\bigl(w(k'_j x^{n_j}y^{m_j} e_\mu)\bigr)<0,\qquad
        x_{\max}\bigl(w(k'_j x^{n_j}y^{m_j} e_\mu)\bigr)<r,\\
        &y_{\max}\bigl(w(k'_j x^{n_j}y^{m_j} e_\mu)\bigr)\ge 5L,\qquad
        y_{\min}\bigl(w(k'_j x^{n_j}y^{m_j} e_\mu)\bigr)\ge 5L-r,\\
        &y_{\min}\bigl(w(k'_j x^{n_j}y^{m_j} e_\mu)\bigr)<0,\qquad
        y_{\max}\bigl(w(k'_j x^{n_j}y^{m_j} e_\mu)\bigr)<r.
    \end{aligned}
\end{equation}
When
\begin{equation}\label{eq: r<frac 54 L}
    r<2L,
\end{equation}
none of the last three cases in Eq.~\eqref{eq: not fully supported in A cup B cup C condition 3} can intersect \(C\). Therefore, any contribution of \(s_{S3}\) to \(C\) must come from the terms satisfying the first condition, or, the narrow strip near the right edge. 
We separate these terms by considering $w(\pi^x_{3L-r:+\infty}\theta_{S3})$, whose support is given by:
\begin{equation}
    \supp\bigl(w(\pi^x_{3L-r:+\infty}\theta_{S3})\bigr)\subset [3L-r,3L+H)\times [-H,5L+H).
\end{equation}
However, locality of the generators implies
\begin{equation}
     \supp\bigl(w(\pi^x_{-\infty:3L-r}\theta_{S3})\bigr)\subset [-\infty:3L)\times \mathbb{Z},
\end{equation}
so $w(\pi^x_{-\infty:3L-r}\theta_{S3})$ cannot exceed the right boundary of $A\cup B\cup C$. Moreover, $w(\theta_{S3})=s_{S3}$ is fully supported in $A\cup B\cup C$ so it cannot exceed that boundary neither. Hence, their difference $w(\pi^x_{3L-r:+\infty}\theta_{S3})$, also cannot touch the right boundary of $A\cup B\cup C$ and the support of $w(\pi^x_{3L-r:+\infty}\theta_{S3})$ can be tightened as
\begin{equation}\label{eq: bound of pi^x_3L-r:+infty theta_S3}
    \supp\bigl(w(\pi^x_{3L-r:+\infty}\theta_{S3})\bigr)\subset [3L-r,3L)\times [-H,5L+H).
\end{equation}

To extract the terms overlapping $C$ in this strip while keeping the coordinate bound in $x$ direction, we compute a Gröbner basis of \(g^1,\dots,g^{n_S}\) using the lexicographic TOP order
\begin{equation}
    x^{n_1}y^{m_1}\prec x^{n_2}y^{m_2}
    \quad\text{if}\quad
    \left\{
    \begin{array}{l}
         n_1<n_2,\\
         \text{or}\\
         n_1=n_2,\ m_1<m_2.
    \end{array}
    \right.
\end{equation}
Denote the resulting Gröbner basis by \(g'''^1,\dots ,g'''^{n'''_S}\). As before, we define the corresponding module \(S'''\) and maps \(w'''\) and \(\varepsilon'''\) by
\begin{align}
    &w''':\hat S'''\rightarrow \hat P \text{ or }S'''\rightarrow P,\\
    &\varepsilon''':P\rightarrow S'''.
\end{align}
We also introduce the support ranges of the new generators:
\begin{align}
&l'''^\mu_{xi}:=x_{\max}(g'''^\mu_i),\qquad
l'''^\mu_{yi}:=y_{\max}(g'''^\mu_i),
\qquad  i=1,\dots,2q,\ \mu=1,\dots,n'''_S,\\
&l'''^\mu_x:=\max_i\{l'''^\mu_{xi}\},\qquad
l'''^\mu_y:=\max_i\{l'''^\mu_{yi}\},\\
&l'''_x:=\max_\mu\{l'''^\mu_x\},\qquad
l'''_y:=\max_\mu\{l'''^\mu_y\}.
\end{align}
Equation~\eqref{eq: bound of pi^x_3L-r:+infty theta_S3} implies that
\begin{equation}
    x^{-3L+r}y^{H}w(\pi^x_{3L-r:+\infty}\theta_{S3})\in (\mathbb{Z}_p[x,y])^{2q}.
\end{equation}
In other words, after shifting the coordinate origin, \(w(\pi^x_{3L-r:+\infty}\theta_{S3})\) is supported entirely in the first quadrant and can therefore be treated as an element of a module over the ordinary polynomial ring \(\mathbb{Z}_p[x,y]\), rather than over a Laurent polynomial ring. This is useful because working over \(\mathbb{Z}_p[x,y]\) gives direct control over the lower bound of \(x\)-coordinates of the terms produced in the Gröbner-basis reduction, which will later ensure that these terms do not overlap \(D\). Since \(x^{-3L+r}y^{H}w(\pi^x_{3L-r:+\infty}\theta_{S3})\) lies in the submodule generated by \(g^1,\dots,g^{n_S}\) over \((\mathbb{Z}_p[x,y])^{2q}\), reducing it by the Gröbner basis \(g'''^1,\dots,g'''^{n'''_S}\) gives
\begin{equation}
     x^{-3L+r} y^{H}w(\pi^x_{3L-r:+\infty}\theta_{S3})
     =\sum_{\mu=1}^{n'''_S} k_\mu g'''^\mu,
     \qquad
     k_\mu \in \mathbb{Z}_p[x,y].
\end{equation}
Moreover, the reduction guarantees that, for each nonzero monomial term \( c_l x^{n_l}y^{m_l} g'''^\mu\) appearing in \(k_\mu g'''^\mu\), where $ c_l\in \mathbb{Z}_p$, the reduction algorithm guarantees that
\begin{equation}
    \lm(c_l x^{n_l}y^{m_l} g'''^\mu)\preceq \lm(x^{-3L+r}y^{H}w(\pi^x_{3L-r:+\infty}\theta_{S3})),
    \qquad \forall \mu.
\end{equation}
Under the order used here, the leading term is the term with largest \(x\)-coordinate. Therefore,
\begin{equation}
\begin{aligned}
    &x_{\max}(c_l x^{n_l}y^{m_l} g'''^\mu)\le x_{\max}(x^{-3L+r}y^{H}w(\pi^x_{3L-r:+\infty}\theta_{S3}))<r,\\
    &x_{\min}(c_l x^{n_l}y^{m_l} g'''^\mu)\ge 0.
\end{aligned}
\end{equation}
The first inequality follows from the support bound in Eq.~\eqref{eq: bound of pi^x_3L-r:+infty theta_S3}. The second inequality follows because \(c_l x^{n_l}y^{m_l} g'''^\mu \in (\mathbb{Z}_p[x,y])^{2q}\).

Undoing the shift gives
\begin{equation}\label{eq: bound of tilde k_mu g'''^mu}
    \begin{aligned}
         &w(\pi^x_{3L-r:+\infty}\theta_{S3})
         =\sum_{\mu=1}^{n'''_S}\tilde  k_\mu g'''^\mu,
         \qquad
         \tilde k_\mu=x^{3L-r}y^{-H}k_\mu\in R.
    \end{aligned}
\end{equation}
For each nonzero monomial term \(\tilde c_l x^{n_l}y^{m_l} g'''^\mu\) appearing in \(\tilde k_\mu g'''^\mu\), where $\tilde c_l\in \mathbb{Z}_p$,
\begin{equation}
    x_{\max}\bigl(\tilde c_l x^{n_l}y^{m_l} g'''^\mu\bigr)< 3L,
    \qquad
    x_{\min}\bigl(\tilde c_l x^{n_l}y^{m_l} g'''^\mu\bigr)\ge 3L-r.
\end{equation}
Hence every translated copy of \(g'''^\mu\) appearing in this decomposition is confined to the same narrow vertical strip.
When
\begin{equation}\label{eq: 2r<L}
    2r<L,
\end{equation}
such a term cannot reach \(A\) or \(D\), simply because its \(x\)-support is too far to the right.

Now suppose \(w'''(k_l x^{n_l}y^{m_l} e_\mu)\) intersects \(C\). Then locality of $g'''^\mu$s implies that its \(y\)-support must satisfy
\begin{equation}\label{eq: bound of w(k_l x^n_l y^m_l e_mu)}
    y_{\max}\bigl(w'''(k_l x^{n_l}y^{m_l} e_\mu)\bigr)< 3L+l'''_y,
    \qquad
    y_{\min}\bigl(w'''(k_l x^{n_l}y^{m_l} e_\mu)\bigr)\ge 2L-l'''_y.
\end{equation}
Furthermore, if
\begin{equation}\label{eq:  l'''_y<frac 54L}
    l'''_y\le 2L,
\end{equation}
then such a term cannot reach \(E\): it is neither high enough nor low enough. Since it also cannot touch \(A\) or \(D\), it must be fully supported in \(B\cup C\).

We now collect precisely those terms. Let \(\theta_{C3}\in S'''\) be the sum of all monomial terms \(k_i x^{n_i}y^{m_i}e_\mu\) that the coefficient $k_i x^{n_i}y^{m_i}$ appears in \(\tilde k_\mu\) whose image \(w'''(k_i x^{n_i}y^{m_i}e_\mu)\) intersects \(C\). Then
\begin{equation}
    \supp\bigl(w'''(\theta_{C3})\bigr)\subset B\cup C,
\end{equation}
and, by construction,
\begin{equation}\label{eq: pi_Cw(theta_C3)=pi_Cw(pi^x_3L-r:3L+r theta_S33)}
    \pi_C w'''(\theta_{C3})=\pi_C w(\pi^x_{3L-r:+\infty}\theta_{S3}).
\end{equation}
Here \(\pi_C\) denotes projection into \(C\):
\begin{align}
    \pi_{C}\!\left(\sum_i k_ix^{n_i}y^{m_i}\right)
    &=\sum_{(n_i,m_i)\in C} k_ix^{n_i}y^{m_i}.
\end{align}
Remember that any contribution of \(s_{S3}\) to \(C\) must come from the narrow strip near the right edge, or, the terms in $\pi^x_{3L-r:+\infty}\theta_{S3}$. Hence
\begin{equation}
\pi_Cw'''(\theta_{C3})=\pi_C w(\theta_{S3}).
\end{equation}

We can now separate off the \(B\cup C\) part of \(s_{S3}\):
\begin{equation}\label{eq: s_S3=w(theta_S3)+w(theta_C3)}
    s_{S3}=w(\theta_{S3})
    =\bigl(w(\theta_{S3})-w'''(\theta_{C3})\bigr)+w'''(\theta_{C3}).
\end{equation}
The second term is fully supported in \(B\cup C\). For the first term, since \(w(\theta_{S3})\) is supported in \(A\cup B\cup C\) and \(w'''(\theta_{C3})\) is supported in \(B\cup C\), it follows that
\begin{equation}
    \supp\bigl(w(\theta_{S3})-w'''(\theta_{C3})\bigr)\subset A\cup B\cup C.
\end{equation}
However,  from the argument above, the first term have no supports in $C$, so
\begin{equation}
    \supp\bigl(w(\theta_{S3})-w'''(\theta_{C3})\bigr)\subset A\cup B.
\end{equation}

Finally, although the two terms in Eq.~\eqref{eq: s_S3=w(theta_S3)+w(theta_C3)} are written using \(w\) and \(w'''\), they are both stabilizers, because \(g^1,\dots ,g^{n_S}\) and \(g'''^1,\dots ,g'''^{n'''_S}\) generate the same submodule:
\begin{equation}
    \mathrm{im}(w''')=\mathrm{im}(w).
\end{equation}
Thus the two terms on the right-hand side of Eq.~\eqref{eq: s_S3=w(theta_S3)+w(theta_C3)} are stabilizers supported in \(A\cup B\) and \(B\cup C\), respectively. Hence \(s_{S3}\) is divisible, completing the proof.
\begin{figure}
    \centering
    \includegraphics[width=0.8\linewidth]{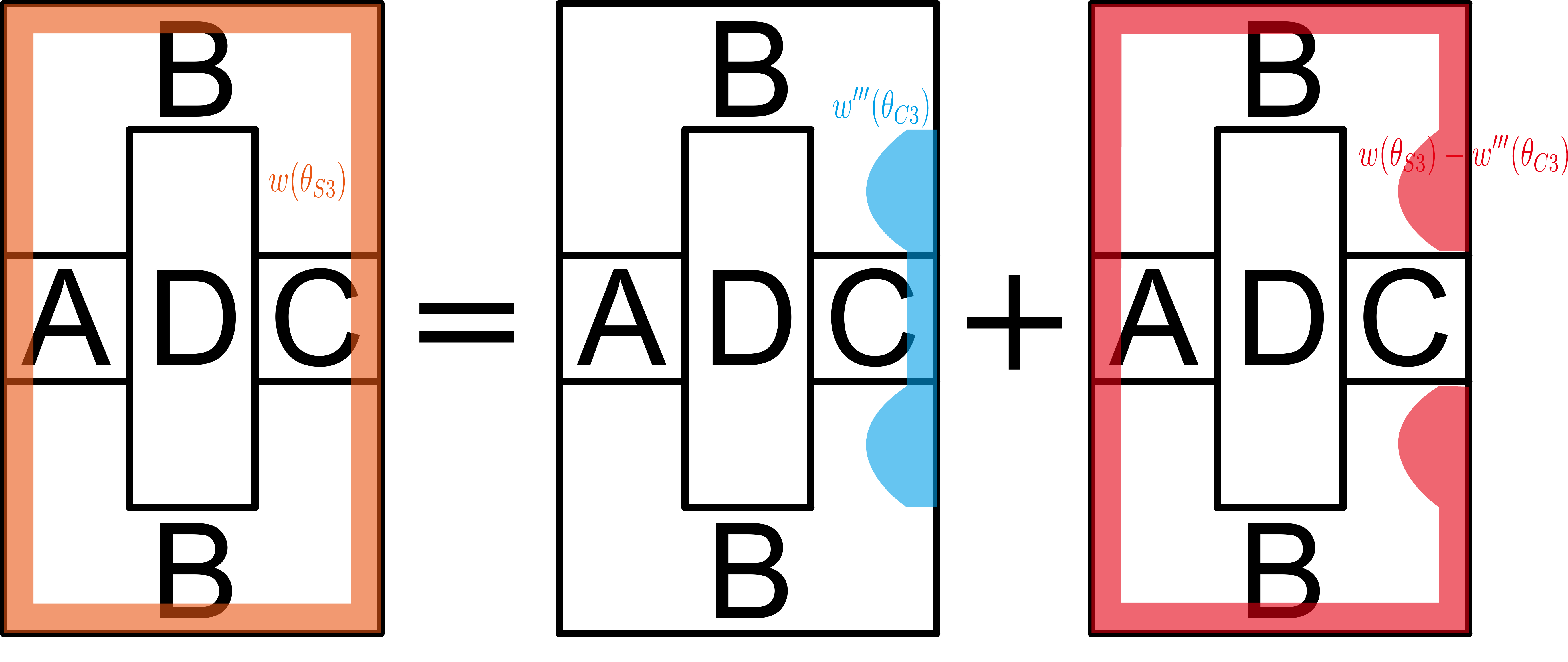}
    \caption{Decomposing $w(\theta_{S3})$ into two stabilizers by picking out all stabilizer generators obtained under a Gr\"obner basis reduction that overlap $C$. Here $w''(\theta_{C3})$ is fully supported in $B\cup C$ and $w(\theta_{S3})-w'''(\theta_{C3})$ is fully supported in $A\cup B$. Here, the lexicographic order chosen and the condition $l_y'''<2L$ guarantee that the two terms will not have supports in $E$, and when $2r<L$, the residual terms around the overlap between the two terms will not have supports in $D$. Orange, blue and red region indicates the support of $w(\theta_{S3}),w'''(\theta_{C3})$ and $w(\theta_{S3})-w'''(\theta_{C3})$, respectively.}
    \label{fig: division of w(theta_S3)}
\end{figure}

\subsubsection{The lower bound on $L$ and conclusion}

It remains to determine a sufficient lower bound on the system size. The proof above used the constraints
Eqs.~\eqref{eq: r<L}, \eqref{eq: 8r^3q^4+7r<L}, \eqref{eq: l''_x<L- 16r^3 q^4 - 14r},
\eqref{eq:  L+32r^3q^4+27r leq frac 54 L}, \eqref{eq: L+8r^3 q^4+6r<frac 54 L},
\eqref{eq: r<frac 54 L}, \eqref{eq: 2r<L}, and \eqref{eq:  l'''_y<frac 54L}. Combining these requirements gives
\begin{equation}\label{eq: the three size bound of L}
\begin{aligned}
    &L\ge 32r^3q^4+27r,\\
    &L\ge l''_x+16r^3q^4+14r,\\
    &2L\ge l'''_y.
\end{aligned}
\end{equation}

We now bound $l''_x$ and $l'''_y$ using Dub\'e's theorem for submodules. The degree bounds for the Gr\"obner basis elements
$g''^1,\dots,g''^{n''_S}$ and $g'''^1,\dots,g'''^{n'''_S}$ imply corresponding bounds on their spatial extent, since the exponent of each variable in a monomial is bounded by the total degree. Thus
\begin{equation}
\begin{aligned}
    &l''_x\le \frac{1}{8}(D_g+2q)^4,\\
    &l'''_y\le \frac{1}{8}(D_g+2q)^4,
\end{aligned}
\end{equation}
where $D_g$ denotes the degree bound for the original generators $g^1,\dots,g^{n_S}$. Since each generator has range at most $r$, we have
\begin{equation}
    D_g\le 2r.
\end{equation}
Substituting these bounds into Eq.~\eqref{eq: the three size bound of L}, we obtain the sufficient conditions
\begin{equation}
\begin{aligned}
    &L\ge 32r^3q^4+27r,\\
    &L\ge 2(r+q)^4+16r^3q^4+14r.
\end{aligned}
\end{equation}
For $q\ge 1$ and $r\ge 1$, both conditions are implied by the single bound
\begin{equation}
    L\ge 2r^4+32r^3q^4+27r+1.
\end{equation}
This completes the proof of Theorem~\ref{thm: improved Levin-wen}.

\end{document}